\documentclass{article}
\usepackage[utf8]{inputenc}
\usepackage[T1]{fontenc}
\usepackage[english]{babel} 
\usepackage{arxiv}

\usepackage{amsmath,amssymb,units,amsthm, subcaption}
\usepackage{graphicx}
\usepackage{hyperref}
\usepackage{dsfont}
\usepackage{url}
\usepackage{tikz}
\usepackage{multirow}
\usepackage{float} 
\usetikzlibrary{arrows.meta}

\newtheorem{lemma}{Lemma}
\newtheorem*{example}{Example}

\title{Adaptive Designs in Fast-Track Registration Processes for Digital Health Applications}
\author{
 Liane Kluge\\
  Competence Center for Clinical Trials Bremen\\
  University of Bremen \\
  \texttt{liane@uni-bremen.de}\\
  \And
 Werner Brannath \\
  Institute for Statistics and\\
  Competence Center for  Clinical Trials Bremen\\
  University of Bremen \\
  \texttt{brannath@uni-bremen.de}\\
}
\begin{document}
\newcommand{\alrel}{\alpha_{\text{rel}}}
\newcommand{\delrel}{\delta_{\text{rel}}}
\newcommand{\etarel}{\eta_{1,\text{rel}}}
\newcommand{\estone}{\hat{\vartheta}_1}
\newcommand{\cp}{\text{CP}_{\delta}}
\newcommand{\cps}{\text{cp}}
\newcommand{\ac}{\alpha_c}
\newcommand{\af}{\alpha_f}
\newcommand{\zf}{z_f}
\newcommand{\lonemin}{\lambda_{1,\min}}
\newcommand{\irel}{I_{\text{rel}}}
\newcommand{\trel}{t_{\text{rel}}}
\newcommand{\itwomin}{I_{2,\min}}
\newcommand{\ionemin}{I_{1,\min}}
\newcommand{\ionemax}{I_{1,\max}}
\newcommand{\Succ}{\mathcal{S}}
\newcommand{\itwoconst}{I_{2,\text{const}}}
\maketitle

\begin{abstract}
Fast-track procedures play an important role in the context of conditional registration of medical devices, such as listing processes for digital health applications. They offer the potential for earlier patient access to innovative products and involve two registration steps. The applicants can apply first for conditional registration. A successful conditional registration provides a limited funding or approval period and time to prepare the application for permanent registration (the second registration step).  For conditional registration, products have to fulfill only a part
of the requirements necessary for permanent registration. There is interest in valid and efficient study designs for fast-track procedures. This will be addressed in this paper. A motivating example is the German fast-track registration process of digital health applications (DiGA) for reimbursement by statutory health insurances. The main focus of the paper is the systematic statistical  investigation of the utility of adaptive designs in the context of fast-track registration processes like the DiGA fast-track. We demonstrate that, in most cases, such designs are much more efficient than the current standard of two separate studies. A careful statistical discussion of the registration requirements and their consequences is also included. The results are based on numerical calculations supported by mathematical arguments.
\end{abstract}
\emph{\textbf{Keywords: }}\normalsize  Confirmatory adaptive design, efficient study design, digital therapeutics (DTx), conditional approval, conditional registration
\section{Introduction}\label{sec1}

The interest in the so-called \emph{fast-track procedures} for the registration of health products and medical devices is currently growing worldwide. Examples are listing processes for digital health applications, conditional marketing authorisation of the EMA \cite{EMACondMarkAuth1, EMACondMarkAuth2}, accelerated approval of the FDA \cite{FDAAccApp1, us2024expedited}, and other accelerated approval pathways summarized in the work of Gupte et al \cite{gupte2025regulatory}. 
The concept of fast-track procedures is as follows: The applicants have the opportunity to first apply for a conditional registration that provides a limited funding period and time to prepare the application for permanent registration. For conditional registration, products have to fulfill less
strict requirements than those necessary for permanent registration, e.g., with regard to the evidence. After conditional registration, applicants have a specific maximum period of time to meet all the requirements for permanent registration, including convincing evidence. The aim of fast-track procedures is that patients can benefit quickly from innovative products, particularly if a product fulfills an unmet medical need. Fast-track procedures are currently becoming increasingly important in the field of digital medical devices. There are many countries that have established such procedures for digital medical devices. This includes countries within (like France, Germany, Belgium) and also outside the EU (like Korea and the United Kingdom), as summarized in an OECD Working Paper of Chapman \cite{chapman2025towards}. There are also many countries that are in the discussion of implementing fast-track procedures, see for instance Tarricone et al \cite{tarricone2024towards} or \cite{overviewFastTrackEnglish}. Furthermore, the existing registration procedures have also become a sales market for foreign companies like Czechia, Poland, and the United States \cite{digaDirectory, mHealthBelgiumWebsite}.

Our paper is mainly motivated by the German fast-track registration procedure for digital health applications (DiGA), as it has been and still is a pioneer for the establishment of fast-track programs for digital medical devices in many other countries \cite{lantzsch2022digital}.  If a DiGA successfully completes the assessment of the BfArM, it is registered for reimbursement by statutory health insurance companies in Germany. The Federal Institute for Drugs and Medical Devices (BfArM) is the German regulatory authority responsible for medicines and digital health applications. 

In this paper, we consider different design options for fast-track registration processes, whereby we account for the requirements of a conditional registration as well as the requirements for permanent registration. On the one hand, this provides a guide for applicants on how to design efficient studies for a fast-track registration process. On the other hand, this will include a careful statistical discussion of the different requirements found in the guidelines, in particular, the BfArM guideline for the registration of DiGA in Germany \cite{leitfadendiga} (an old BfArM guideline version is also available in English \cite{oldLeitfadenEnglish}). The discussion of these requirements is important, as governments and health insurance companies of many countries are increasingly faced with the challenge of deciding which health technologies they should pay for \cite{chapman2025towards}, and on the other side, the community of applicants wants fair requirements as it has been indicated in a survey by Ataiy et al \cite{ataiy2024digital}.  Taking into account the discussed requirements, our focus is on overall valid, efficient designs. We consider a design as overall valid if it provides overall type I error rate control, in the sense that the probability for a permanent registration of a DiGA without a positive effect in the considered endpoint describing the potential health care effect is bounded by some level $\alpha$, e.g. for the one-sided level of $\alpha=0.025$. Our considerations will include the current standard of separate studies for conditional and permanent registration, as well as two-stage adaptive designs (see, for instance, the work of Wassmer and Brannath \cite{BuchWassmerBrannath2016}) that permit the use of the data of the conditional registration study also for the permanent registration.  To the best of our knowledge, this is the first paper to systematically investigate the usefulness of adaptive designs in the context of a fast-track registration process for digital health applications. We illustrate that under common scenarios, adaptive designs are much more efficient than the standard approach of two separate studies for conditional and permanent registration. An important aspect of this assessment will be the overall performance of the whole two-step registration process, such as the overall success probability (power to receive permanent registration) and, as already mentioned, the type I error rate. 
One important result of our requirements discussion is the identification of scenarios for which a two-stage registration process is less efficient than the direct application for a permanent registration without an intermediate conditional registration step. For such situations, we suggest a data-driven choice based on the first-stage data between the fast-track procedure and the direct application for permanent registration. In particular, we will demonstrate that a mixture of an adaptive two-stage and single-stage design can provide an efficient alternative to applying for permanent registration directly.  Such options are possible with the current DiGA registration regulations in Germany. 
 
Our article is arranged as follows. Section~\ref{sectionExaplesFastTrack} gives a brief overview of the regulations and the established fast-track procedures of the different countries. Section~\ref{subsectionRequirements} provides a general discussion of requirements for conditional and permanent registration. Section~\ref{NewChapterStudyDesignsFastTrack} discusses study designs under the assumption that a permanent registration is impossible or of no further interest after an unsuccessful conditional registration process. Reasons for this restriction could be a lack of resources without conditional registration or regulatory concerns with regard to an initially unsuccessful registration process. In Section~\ref{ChapterNonBindingFutility}, we will consider designs with a possibility of a data-dependent decision between the fast-track procedure and direct permanent registration, which guarantee overall type I error rate control. The paper concludes in Section~\ref{sec:Summary} with a detailed summary and discussion of our results as well as their implications.

\section{Examples of current Fast-Track Registration Processes}\label{sectionExaplesFastTrack}
We briefly outline the legal basis and regulatory requirements of the fast-track procedures for digital medical devices established to date. The ``app on prescription'' in Germany has a pioneering role and was introduced by the Digital Healthcare Act, which came into force in 2019. Since then, digital health applications (\emph{Ger.} Digitale Gesundheitsanwendungen, DiGA) can be prescribed by physicians and psychotherapists and are reimbursed by statutory health insurance companies. The currently reimbursable DiGA are listed in the DiGA directory \cite{digaDirectory}. To be listed in the directory, the applications must successfully complete the assessment of the BfArM. The details of the registration process and the requirements for reimbursable DiGA are regulated by the Digital Health Applications Ordinance (DiGAV). A concrete interpretation of that and more practical details can be found in the DiGA guideline of the BfArM \cite{leitfadendiga}. An old guideline version is also available in English \cite{oldLeitfadenEnglish}. 

There are two different possible ways of permanently being listed in the directory: following the fast-track procedure or applying for a permanent listing directly. The requirements to permanently be listed in the directory are the demonstration of at least one positive healthcare effect (via statistical significance) for a specific patient group, as well as all quality and compliance standards from Sections 3 to 6 of the DiGAV. In order to prove a positive healthcare effect, a quantitative comparative study should be conducted and submitted to the BfArM. It must also be registered in a public trial registry. Compared to the application for a permanent listing, the requirements for a conditional listing are partly less strict because the study for the proof of a positive healthcare effect must not yet have been conducted or concluded. 
Instead, the DiGA applicants must submit a so-called \emph{systematic evaluation of data} where data from the use of the DiGA must be analyzed. This data, which we call \emph{pilot data}, must plausibly illustrate that a subsequent comparative study will be successful in proving a positive healthcare effect of the DiGA \cite[pp.~109-110]{leitfadendiga}. According to the DiGA guide, showing this initial evidence includes a relevant (numeric) positive healthcare effect that must be calculated from the pilot data  \cite[p.~110]{leitfadendiga}. Further requirements are specified in the DiGA guide. The pilot data collection is intended to prepare the final study for the proof of the positive healthcare effect needed for permanent registration. For instance, this includes case numbers, measuring instruments, endpoints, and recruitment methods. The applicants have to submit the evaluation concept (including study protocol) for the final study when applying for conditional listing. Ideally, the patient populations and endpoints investigated in the systematic evaluation of data and the final study should be the same. If an applicant receives conditional registration from the BfArM, the DiGA is listed for reimbursement in the directory for one year. The proof of a positive healthcare effect via a comparative study must be submitted to the BfArM within the one year of conditional listing. Otherwise, the DiGA is removed from the directory, and the applicant can apply its DiGA for a further listing after one year at the earliest. In summary, the applicants have, with the fast-track procedure, one year more to conduct the comparative study to prove a positive healthcare effect while already being listed as a reimbursable DiGA. 

The ``app on prescription'' in Germany and the fast-track registration process for DiGA have inspired France to establish an analogue regulatory framework for digital medical devices, \emph{DMD} for short, (\emph{FR.} dispositifs médicaux numériques (DMN)) that are reimbursable by health insurance companies. In 2023, France adopted the fast-track registration process, which is called \emph{PECAN} (\emph{FR.} Processus de Certification des Applications Numériques de Santé) \cite{overviewFastTrackFrance} with Decree 2023-232. Belgium has also established an authorisation pathway for digital health applications (\emph{FR.} Applications de santé mobiles) with the \emph{mHealth pyramid}, see \cite{mHealthBelgiumWebsite}, via the royal decree \cite{RoyalDecree2021}. Many other EU countries are currently in the process of introducing or in discussion of implementing fast-track procedures for digital health applications, such as Austria, Italy, and more \cite{tarricone2024towards, overviewFastTrackEnglish, AustriaDiGA}. There are also ambitions to create a uniform assessment framework for digital interventions within the EU via the ASSESS DHT \cite{euProject1} or EDiHTA \cite{euProject2} project.

Outside the EU, for instance, in Korea, the fast-track pathway \emph{Integrated Review and Assessment Programme for Innovative Medical Devices} was established in 2022 and 2023, which enables provisional insurance listing based on initial evidence \cite{chapman2025towards}. In the UK, there is the \emph{NICE’s Early Value Assessment (EVA)} pathway, a fast-track process for (digital) technologies which offers the possibility of a conditional recommendation which can provide a basis for conditional funding \cite{chapman2025towards, evaUK}.
\section{Requirements for conditional registration and consequences}\label{subsectionRequirements}
This section addresses the regulatory and strategic requirements for conditional and permanent registration motivated by the German fast-track procedure. Firstly, we make a few general assumptions. Usually, the data collection for the conditional and for the permanent registration is realized through two separate studies \cite{digaDirectory}. In this paper, we consider the scenario in which the two separate studies have the same design, with the same primary endpoint and the same study population. Later, we propose using two-stage adaptive designs as an alternative to non-adaptive separate study designs. In the following, the term ``stage'' is used to refer to the two studies of a non-adaptive design and the two stages of an adaptive design. 

In the numerical investigations presented below, we assume a health care effect 
parameter $\vartheta$ with $H_1:\vartheta>0$ corresponding to a positive healthcare effect and $H_0:\vartheta\le 0$ to an inefficient product. We will further assume that at both stages, $k=1,2$, we obtain a normally distributed estimate $\hat{\vartheta}_k\sim\mathcal{N}(\vartheta\,,\,I_k^{-1})$ of the healthcare effect with associated one-sided p-values $p_k=1-\Phi(\hat{\vartheta}_1\sqrt{I_k})$, where $\Phi$ is the standard normal distribution function and $I_k$ is the (Fisher) information of the data from stage $k$. Loosely speaking, the information measures how precisely a parameter can be estimated from the data — the larger it is, the more precise the estimate. It is typically proportional to the sample sizes. The stage-wise effect estimates are assumed to be calculated from disjoint cohorts and thereby stochastically independent. We further assume that for a permanent registration, statistical significance at one-sided level $\alpha=0.025$ is required (for the fast-track procedure and direct permanent registration). We assume that applicants must announce before conducting a stage/study, whether it is intended for conditional or permanent registration, in order to avoid inflation of the Type I error rate $\alpha$.

A possible additional requirement for permanent registration, besides statistical significance, in the case of the fast-track procedure, could be that the application for conditional registration must have been successful (for whatever reason). We will leave this open for now and discuss the fast-track procedure in this section with and without this additional requirement.

\subsection{Requirements motivated by the German DiGA guide}
For the conditional registration, motivated by the German DiGA guide, we assume that there is a minimal clinically relevant effect $\delrel>0$, which is strict in the sense that a product with a healthcare effect lower than $\delrel$ is considered as not having a relevant positive impact. Accordingly, a natural requirement for the conditional registration is that, at stage one, a point estimator of $\estone\geq\delrel$ is observed.

Another natural requirement for conditional registration is $p_1\le \ac$ for some significance level $\ac\in(0,1)$. This bounds the probability for a temporary registration of an ineffective product. There are several reasons why the conditional registration of an ineffective product is undesirable, even when the final study for a permanent registration fails. One reason are the costs coming with a temporary registration without a relevant healthcare effect. Users may even be exposed to a product with negative healthcare or side effects during the temporary period of registration.  In addition, applicants and
health insurance companies may suffer a considerable loss of reputation. This requirement is also mentioned in the German DiGA guide, whereby a two-sided significance level of $0.1$ is indicated as a possibility for conditional registration \cite{leitfadendiga}, which corresponds to the one-sided stage-one level of $\ac=0.05$. 

We list further possible requirements for conditional registration. In the DiGA guide, it is indicated that for a successful conditional registration, the study plan for the final stage must be submitted. In this context, it must plausibly be illustrated that a subsequent comparative study/stage will be successful in proving that the DiGA has a positive healthcare effect. This requirement is justified by the costs and complications (for the provider and patients) from a withdrawal of registration, which also comes with a loss of reputation. We propose the following specification of this requirement: The power calculation for the second stage should be done for a sufficiently high power, like $80\%$ (or higher) at the observed first stage healthcare effect, or a lower value.

If the application for a permanent registration is impossible after an unsuccessful study for conditional registration, the impact of the requirements $\estone\geq\delrel$ and $p_1\leq\ac$ for a conditional registration is substantial and therefore requires a careful investigation. Therefore, we will investigate these regulatory requirements and their impact on the strategic requirements of the applicants, such as power considerations.

A priori, it is not clear which of the above-mentioned two requirements is more stringent.  
It can be calculated that $\estone\geq\delrel$ is equivalent to $p_1\le \alrel:= 1-\Phi(\delrel\sqrt{I_1})$. Hence, the first requirement is equivalent to the second one with a specific stage-one level. We introduce the following notations to visualize this level in the relative scale:
\begin{align}\label{eq:trel}
\trel(I_1):=
{I_1}/{\irel}\quad\text{with}\quad \irel:=
{\eta_f^2}/{\delrel^2} \text{ where } 
\eta_f:=\Phi^{-1}(1-\beta)+\Phi^{-1}(1-\alpha)~.
\end{align}
The parameter $\eta_f$ defines the noncentrality parameter and $\irel$ the information
needed to achieve a power of $1-\beta$ at $\delrel$ with significance level $\alpha$. 
In our numerical investigations, we assume the commonly used $\alpha=0.025$ (one-sided) and $1-\beta=0.8$, because these values are used in the majority of DiGA studies for permanent registration.  This leads to $\eta_f= 2.8$. 
The significance level $\alrel$ is shown in Figure~\ref{fig:sigLevel} (solid decreasing line) in dependence on $\trel(I_1)$. The $\trel(I_1)$-axis ends at $1$, because for $\trel(I_1)>1$, it does not make sense to decide in favor of the fast-track procedure: A first stage information $I_1>\irel$ ensures a power of at least $1-\beta$ for rejecting $H_0$ at level $\alpha$ with a single fixed-size sample test for all $\delta\geq \delrel$. In this case, a study for a direct, permanent registration is more efficient than the pilot study to achieve conditional registration only. Thus $\irel$ can be considered as a universal upper bound for $I_1$.
\begin{figure}[h]
\centering
\includegraphics[width=0.5\textwidth]{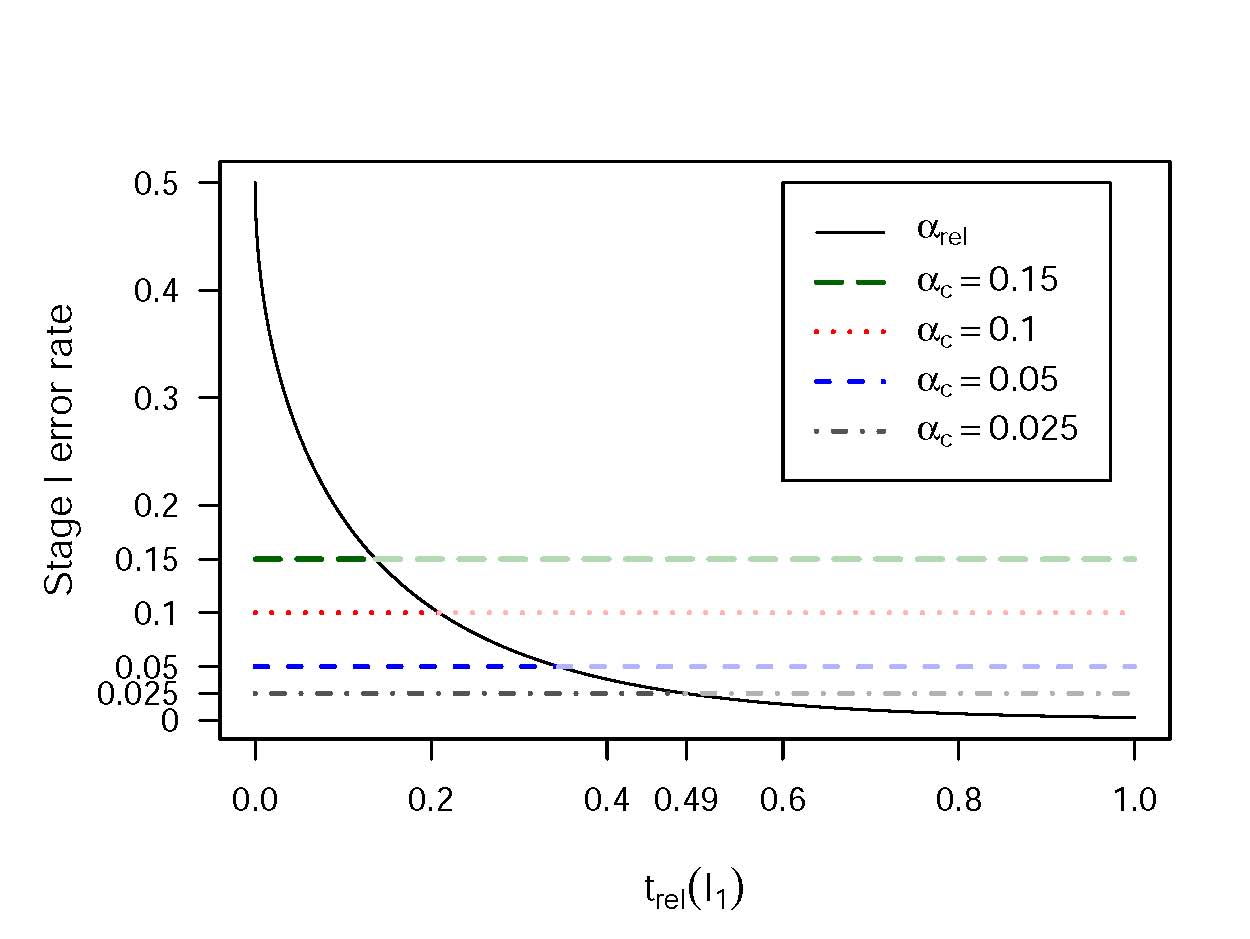}
\caption{\label{fig:sigLevel} Type I error rate of conditional registration resulting from the requirement $\estone\geq\delrel$ in comparison to different levels $\ac$ (one-sided).}
\end{figure}

It can be seen from Figure~\ref{fig:sigLevel}  that demanding $\estone\geq\delrel$ for conditional registration does not automatically control the Type I error rate for erroneously receiving a conditional registration under $H_0$. Accounting for both requirements, one must achieve 
\begin{align*}
    p_1\leq\af:=\min\bigg\{1-\Phi\bigg(\delrel\sqrt{I_1}\bigg),\ac\bigg\}=\min\{\alrel,\ac\}
\end{align*}
in the pilot study. The levels $\ac=0.05, 0.1, 0.15$ (one-sided) are indicated in Figure~\ref{fig:sigLevel} by horizontal lines. For the different choices of $\ac$, the Type I error rate $\af$ for conditional registration can be read off Figure~\ref{fig:sigLevel} by using the minimum of the corresponding horizontal and decreasing solid line. For example, with the choice $\ac=0.05$ indicated in the DiGA guide \cite{leitfadendiga} and $I_1\le 0.3447\cdot \irel$, the requirement for a type I error rate control at level $\ac$ is  more stringent than
the requirement of a point estimator above $\delrel$. 

\begin{example}
    As an illustrative example, we consider a case inspired by the
    DiGA \emph{elevida}, which targets fatigue in patients with multiple sclerosis and has been permanently registered in the German DiGA directory via the fast-track procedure \cite{digaElevida}. The study for permanent registration was a parallel group, two-arm RCT \cite{pottgen2018randomised}. The primary outcome was the physical and cognitive fatigue measured on the Chalder Fatigue Scale. The healthcare effect was defined by the mean difference between a waitlist control group and the treatment group (using DiGA for 12 weeks). As mentioned in \cite{pottgen2018randomised}, the minimally important difference for improvement on the Chalder Fatigue Scale was estimated between $0.7$ and $1.4$. In the following, we will work with $\delrel=1$.  We assume hypothetically that a fast-track procedure is planned for a similar DiGA by another applicant, based on the information of the DiGA elevida. In this fast-track procedure, we use parallel group designs for both conditional and permanent registration, each with balanced samples and normal mean differences as estimates for the healthcare effect. The informations for the first and second stage are then given by $I_k=n_k/(2\sigma^2)$ where $n_k\in\mathbb{N}$ is the sample size per group for stage $k$, $k=1,2$. A sample size calculation with $\alpha=0.025$ and $1-\beta=0.8$ at $\delrel=1$ would give the information $\irel=2.8^2/1^2=7.84$. Using the estimated variance $\hat{\sigma}^2=5.17^2$ from \cite{pottgen2018randomised}, the corresponding per-group sample size would be $n_{1,\text{rel}}:= 420\approx 2\hat{\sigma}^2\irel$. This implies, for example, that a pilot study with $\trel(I_1)=0.5$ corresponds to a study with per-group sample size of $n_1=0.5\cdot 420=210$.
\end{example} 

It is relevant to investigate whether the conditions set for conditional registration are less strict than those for permanent registration. To have a less stringent condition for conditional registration than for permanent registration, $\af>\alpha$ should hold true, which implies $\ac>\alpha$ and $1-\Phi(\delrel\sqrt{I_1})>\alpha$.  The latter condition is fulfilled when  
\begin{align}
    I_1< \ionemax := 
    {\Phi^{-1}(1-\alpha)^2}/{\delrel^2}~.\label{UniversalUpperBoundInformation}
\end{align}
A pilot study with information $I_1\ge \ionemax$ and the requirement $p_1\le \af\ (\le \alpha)$ wouldn't make sense for the applicants, because it would be less or as efficient as a similarly large study for permanent registration for which only $p_1\le\alpha$ is required. Therefore, we only consider pilot information $I_1$ that is lower than the universal upper bound $\ionemax$.

Note that $\trel(\ionemax)$ is independent from $\delrel$ and only depends on $\alpha$ and $\beta$ (which determine $\eta_f$). For $\alpha=0.025$ and $\beta=0.2$ we obtain $\trel(\ionemax)=0.49$ which is marked in Figure~\ref{fig:sigLevel}. In our example fast-track procedure with $n_{1,\text{rel}}=420$, this implies that the per group sample size for conditional registration should be smaller than $n_{1,\max}:= 206\approx 0.49\cdot 420$.
\subsection{Consequences in conjunction with required conditional registration}\label{ConsequencesRequirementsBindCondRegits}

This subsection analyses the consequences of the 
requirements for a conditional registration for situations where an unsuccessful
conditional registration renders a study for a permanent registration impossible (for whatever reasons), i.e. where conditional registration is required. A justified aim of the applicants, but also regulators, should be to guarantee a certain probability of overall success, which we understand as the probability of achieving permanent registration. This leads to the requirement that the success probability of the conditional registration needs to be sufficiently large. For example, to reach an overall success probability of at least $1-\beta$ at some $\delta\geq\delrel$, we must have a pilot study with ``power''
\begin{align}
    \mathbb{P}_{\delta}\left(Z_1\geq z_f\right)> 1-\beta~,\label{conditionCondApprovalExtended}
\end{align}
where $Z_1:=\sqrt{I_1}\estone$ and 
\begin{align}\label{definition_zf}
z_f:=\max\left\{\sqrt{I_1}\delrel,\Phi^{-1}(1-\ac)\right\}\,.
\end{align}
Note at first that this condition cannot be fulfilled (for $1-\beta>0.5$) if $\delta=\delrel$: In this case, the power in (\ref{conditionCondApprovalExtended}) cannot exceed $0.5$ because $Z_1\ge z_f$ implies 
$\estone\ge\delrel$. So if $\delta=\delrel$ is considered highly likely and worthwhile for registration, conducting a single trial for permanent registration is more promising than a fast-track registration process with required conditional registration. 

Let us now assume that investigators a priori assume a healthcare effect $\delta>\delrel$ and aim at a sufficient success probability under this assumption. Due to $Z_1\sim \mathcal{N}(\delta\sqrt{I_1},1)$ condition (\ref{conditionCondApprovalExtended}) is equivalent to
\begin{align}
I_1\geq \ionemin:=\max\left\{\left(\frac{\Phi^{-1}(1-\beta)}{(\xi-1)\cdot\eta_f}\right)^2, \left(\frac{\Phi^{-1}(1-\ac)+\Phi^{-1}(1-\beta)}{\xi\cdot\eta_f}\right)^2\right\}\cdot\irel\label{lambdaleftMax}
\end{align}
where
\begin{align*}
    \xi:={\delta}/{\delrel}\in(1,\infty)~.
\end{align*}
Figure~\ref{minimumNeededBoth} (left) shows, in dependence on $\xi$, the universal lower bound (\ref{lambdaleftMax}) in the relative scale $\trel(\ionemin)={\ionemin}/{\irel}$. Note that $\trel(\ionemin)$ is independent from $\delrel$ and only depends on $\alpha$, $\beta$, $\ac$ and $\xi$. 
The first part of the maximum in (\ref{lambdaleftMax}) comes from the requirement $p_1\leq\alrel$ and is represented in Figure~\ref{minimumNeededBoth} (left) by the solid decreasing line. The other lines represent the second part of the maximum in (\ref{lambdaleftMax}) (which comes from the requirement $p_1\leq\ac$) for the different (one-sided) levels $\ac=0.05, 0.1, 0.15$. The minimum possible relative Information $\trel(\ionemin)$ can also be read off Figure~\ref{minimumNeededBoth} (left) by using for each $\ac$ the maximum of the corresponding horizontal line and the decreasing solid line.

The universal upper bound $\ionemax$ for the information of a reasonable pilot study is also shown in Figure~\ref{minimumNeededBoth} (left) on the relative scale, i.e. $\trel(\ionemax)$, namely by the horizontal dot-dashed line. Recall that this universal upper comes from the regulatory requirement $\estone\geq\delrel$ and the reflection that the requirements for conditional registration should not be more stringent than for a permanent registration. One can see that the universal lower bound $\trel(\ionemin)$ can exceed this universal upper bound, in which case the fast-track procedure with a required conditional registration does not make sense. This is the case for $\xi$ falling short of the intersection point between the decreasing solid line and dot-dashed line, which equals 
\begin{align}
\xi_{\min}:= 1+ \frac{\Phi^{-1}(1-\beta)}{\Phi^{-1}(1-\alpha)}\,.
\end{align}
With $\alpha=0.025$ and $\beta=0.2$ we obtain $\xi_{\min}\approx 1.43$. This implies that the here considered fast-track process is useful only if we can assume a health care effect that is at least 43\% higher than the minimal relevant effect $\delrel$. Note that this effect would even be larger with a larger overall power of e.g. $1-\beta=0.9$.
All the following considerations of the fast-track procedure with required conditional registration will be made under the assumptions $\alpha=0.025$ and $\beta=0.2$.

\begin{figure}[h]
\centering
\includegraphics[width=1\textwidth]{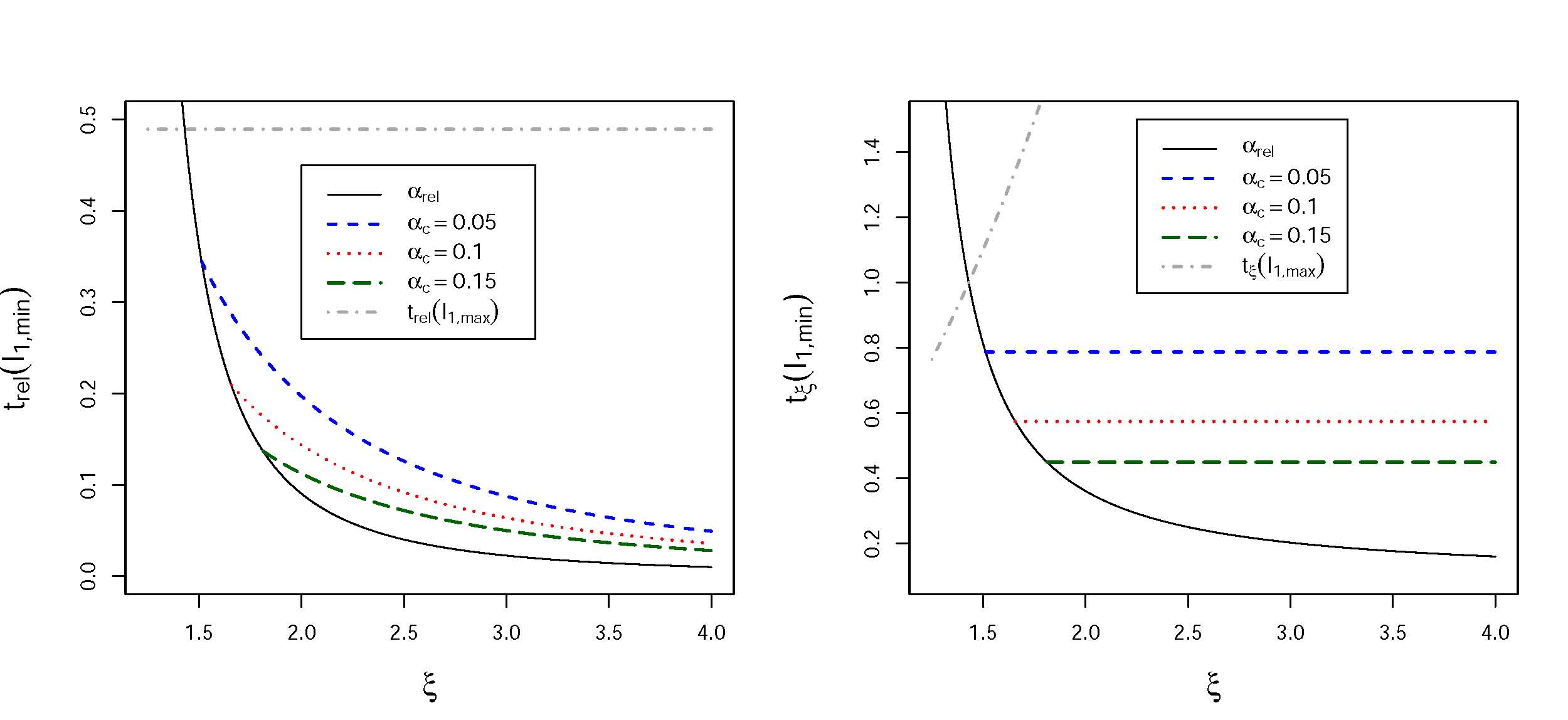}
\caption{Visualisation of the minimum information needed for the first period of data collection. The universal lower bound for $\trel(I_1)={I_1}/{\irel}$ in dependence on $\xi=\delta/\delrel\ge \xi_{\min}=1.43$ (left) and the universal lower bound for $t_{\xi}(I_1)={I_1}/{I_{\delta}}$ in dependence on $\xi\ge \xi_{\min}$ (right) are shown.\label{minimumNeededBoth}}
\end{figure}
\begin{example}[continued]
    Since the observed healthcare effect in the conducted trial \cite{pottgen2018randomised} was $2.74$, a reasonable and conservative a priori healthcare effect assumption is $\delta=2=2\cdot\delrel$, i.e. $\xi=2$. We can read off Figure~\ref{minimumNeededBoth} (left), that $\trel(\ionemin)=\ionemin/\irel=n_{1,\min}/n_{1,\text{rel}}=0.2$ for the choice of $\ac=0.05$. This implies that,  to meet (\ref{conditionCondApprovalExtended}), the fast-track procedure needs a minimum first stage sample size per group of $n_{1,\min}=84= 0.2\cdot 420$.
    For comparison, a fixed-size sample test with significance level $\alpha=0.025$ and power $1-\beta=0.8$ at $\delta$ would need the 
    information of $I_{\delta}:=\eta_f^2/\delta^2=2.8^2/2^2=1.96$, i.e. a per group sample size of $n_{\delta}:= 105\approx 2\hat{\sigma}^2 I_{\delta}$.    
\end{example}

Next, we discuss the choice of $\ac$ for the fast-track procedure with required conditional registration with reference to Figure~\ref{minimumNeededBoth} (right).
This figure shows the universal lower bound for the pilot information on another scale, namely relative to  the information $I_{\delta}:={\eta_f^2}/{\delta^2}$ needed for achieving a power of $1-\beta$ at the true or a priori assumed alternative $\vartheta=\delta$ with a single fixed-size sample test: The vertical axis of this figure is
\begin{align*}
t_{\xi}(I_1):={I_1}/{I_{\delta}}=\trel(I_1)\xi^2
\end{align*}
and it shows $t_{\xi}(\ionemin)=\trel(\ionemin)\xi^2$ in dependence of $\xi$. Since this scale increases with $\xi^2$, the horizontal dotted-dashed line for the universal upper bound in Figure~\ref{minimumNeededBoth} (left) becomes a curved line in Figure~\ref{minimumNeededBoth} (right), and the non-solid curved lines in Figure~\ref{minimumNeededBoth} (left), representing
the part of the universal lower bound coming from the requirement $p_1\le \ac$, are becoming straight lines in Figure~\ref{minimumNeededBoth} (left). The universal lower bound for the pilot information (to achieve a power larger than $1-\beta$) is again represented by the maximum of the decreasing solid curve and the horizontal line corresponding to the respective $\ac$. One can see from this figure that with $\ac=0.05$ the minimum information for the pilot study must exceed $78\%\ (\approx 80\%)$ of the information $I_\delta$ required for a fixed-size sample test with a power of $1-\beta$ at the true or a priori assumed alternative $\vartheta=\delta$.
This indicates that $\ac=0.05$ may lead to rather large pilot studies and that a larger choice for $\ac$ would be more practical. Since the level $\ac=0.15$ leads to a universal lower bound below $I_\delta/2$ (more precisely, $44,94\%$ of $I_\delta$) and may still be considered an acceptable level for conditional registration, we will use this value in our further investigations. Recall that for our example fast-track procedure, we have for $\delta=2$ that $I_{\delta}=1.96$ with corresponding per group sample size $n_{\delta}=105$. From Figure~\ref{minimumNeededBoth} (right) for $\ac=0.15$ it follows that $0.45=t_{\xi}(\ionemin)=n_{1,\min}/n_{\delta}$. That means $n_{1,\min}=48\approx0.45\cdot n_{\delta}$ is now the minimum needed first-stage sample size per group and significantly smaller than $84$, which is the one needed if $\ac$ is not $0.15$ but $0.05$.

Let us further explain why we considered the two different scales $\trel$ and $t_{\xi}$: Depending on the applicant’s stance, i.e. the confidence in the assumed a priori effect $\delta$, one scale may be more relevant than the other. If there is little confidence in the assumed a priori effect $\delta$ and few or no data have been collected before starting to collect data as part of the fast-track procedure, the $\trel$-scale may be more interesting. However, if applicants have greater confidence in the correctness of the assumed a priori effect $\delta$, comparing the information with the one needed for a fixed-size sample test powered at $\delta$ may be more appropriate. In this case, applicants especially do not likely want the (minimum) information for the first stage to exceed $I_{\delta}$, i.e. $t_{\xi}(\ionemin)<1$ should hold.

In addition to the regulatory requirements, we recommend that applicants set their own strategic requirements, such as achieving a certain overall success probability, including the conditional and permanent registration process, to improve their likelihood for a successful registration.

\section{Study designs with required conditional registration}\label{NewChapterStudyDesignsFastTrack}
In this section, we assume that at the time of the planning of a fast-track procedure by the applicants, an unsuccessful study for conditional registration precludes any further attempt to receive a permanent registration. Reasons for this may be regulatory requirements or financial aspects, for example. We relax this assumption in Section~\ref{ChapterNonBindingFutility}. We will investigate the efficiency of the standard design with two separate studies when a certain overall power is anticipated, and compare them with adaptive designs after a brief recap of adaptive designs. 

\subsection{Overall performance of the standard fast-track registration process}\label{subsectionConsequences}
We consider the overall and second-stage power and corresponding sample sizes of the standard design with two separate stages (studies). The overall power calculation (including both stages) is done for some a priori assumed effect $\delta$ with $\xi=\delta/\delrel>\xi_{\min}$ for the power $1-\beta$. We further assume that, after a successful first stage, the second stage is powered for the same value $1-\beta$, but now at the observed first stage estimate $\estone$, using the well-known sample size formula for z-tests. Recall that 
a successful first stage implies $\estone\ge\delrel$. 

Since the calculated second stage information can become arbitrarily small for large interim effect estimates, but the second stage study should to be sufficiently large to be sufficiently convincing for a permanent registration, we also assume a specific minimum information $\itwomin$ that is never undercut, even if the second stage sample size calculation falls short of this minimum. The selection of this minimum stage-two information $\itwomin$ also provides a way to control the overall power with a fixed stage-one information. 
In the following, we will consider the minimal information $\itwomin$ required for an overall success probability ($1-\beta$) as a function of a (given) first-stage information $I_1$ assuming a specific health care effect $\delta=\xi\delrel$ for $\xi>\xi_{\min}$. This will inform us about the efficiency of the standard fast-track procedure under different assumptions and compare it with alternative designs utilizing adaptive design methodology.
 
As described, we assume that the fast-track procedure leads to a successful permanent registration if the first stage is successful (with conditional registration) and, at the second stage, the hypothesis $H_0$ is rejected by the z-test with second-stage p-value $p_2\leq\alpha$. Therefore, the overall power can be written as
\begin{align}\label{IntegralOverallPowerSharperBorder}
\Succ(\itwomin,I_1):=\, & \mathbb{P}_{\delta}\left(Z_1\geq\zf,p_2\le \alpha\right)=\int\limits_{z_f}^{\infty}\mathbb{P}_{\delta}(p_2\le \alpha~|~Z_1=z_1)\,\phi\big(z_1-\sqrt{I}_1\delta\big)dz_1\notag\\ 
&=\int\limits_{z_f}^{\infty}\left\{1-\Phi\left(\Phi^{-1}(1-\alpha)-\sqrt{\max\big\{\itwomin,I_1\eta_f^2/z_1^2\big\}}\delta\right)\right\}\phi\big(z_1-\sqrt{I}_1\delta\big)dz_1,
\end{align}
where $\phi$ is the density of the standard normal distribution, $\eta_f$ is defined as in \eqref{eq:trel}, $z_f$ is as in \eqref{definition_zf} and
\begin{align}\label{RelativeInformationZweiteStudieKlassisch}
I_2(z_1):=I_2(z_1,\itwomin,I_1):=\max\big\{\itwomin,I_1\eta_f^2/z_1^2\big\}=\max\left\{\itwomin,I_1\left\{\Phi^{-1}(1-\beta)+\Phi^{-1}(1-\alpha)\right\}^2/{z_1^2}\right\}
    \end{align}
is the information for the second stage when $z_1\geq\zf$. 
\begin{example}[continued]
For our example fast-track procedure, we assume that $n_1=63=0.6\cdot n_{\delta}>n_{1,\min}=48$ is chosen as a per group sample size for the first stage, i.e. $I_1=n_1/(2\hat{\sigma}^2)=1.18$  and $t_{\xi}(I_1)=0.6$. Then $\zf=\max\{\sqrt{n_1/(2\hat{\sigma}^2)}\delrel,\Phi^{-1}(1-0.15)\}=1.09$ (see (\ref{definition_zf})). If, for instance, $\estone=1.2$ is the observed estimate after the first stage with corresponding (observed) noncentrality parameter $z_1=\sqrt{I_1}\estone=1.4$, the fast track procedure is continued because $z_1>\zf$. The sample size per group for the second stage is than according to (\ref{RelativeInformationZweiteStudieKlassisch}):
\begin{align}
    n_2(z_1):=2\hat{\sigma}^2 I_2(z_1)=\max\left\{n_{2,\min},2\hat{\sigma}^2\left\{\Phi^{-1}(1-\beta)+\Phi^{-1}(1-\alpha)\right\}^2/\estone^2\right\}.\label{formulaSecondStageSampleSizeExampleSeparate}
\end{align}
If $\itwomin$ and thus $n_{2,\min:}=2\hat{\sigma}^2\itwomin$ is set, the letter per group sample size can explicitly be calculated.
\end{example}

Recall from Section~\ref{subsectionRequirements} that to achieve the overall power $1-\beta$ at the assumed a priori effect $\delta=\xi \delrel$ for $\xi>\xi_{\min}$, the information of the first stage must be chosen greater than $\ionemin$ defined in \eqref{lambdaleftMax}. This results from the fact that the overall power is less than the power for a successful conditional registration. It can be shown that for $\xi>\xi_{\min}$ there exists for all $I_1>\ionemin$ a unique minimum information $\itwomin$ such that $\Succ(\itwomin,I_1)=1-\beta$ if $\Succ(0,I_1)\leq 1-\beta$, and if $\Succ(0,I_1)>1-\beta$ the overall power is achieved without a minimum information for the second stage, in which case we set $\itwomin=0$. The proof of this result can be found in the Appendix.

Again, for a more universal investigation, we consider relative informations rather than the absolute informations $\itwomin$ and $I_1$. Since we believe that in practice, the fixed-size sample test for the a priori assumed $\delta=\xi\delrel$ is the more relevant benchmark than the fixed-size sample test for $\delrel$, we will present the results on the $t_{\xi}$-scale, however, now for a fixed $\xi$.
Because $$\sqrt{I}\delta =\sqrt{I/I_{\delta}}\,\eta_f
=\sqrt{t_\xi(I)}\,\eta_f \quad\text{and}\quad 
\zf=\max\{\sqrt{t_{\xi}(I_1)}\eta_f/\xi, \Phi^{-1}(1-\ac)\} \text{ by \eqref{definition_zf}},
$$ we can write $\mathcal{S}(\itwomin,I_1)$ as function of $t_{\xi}(\itwomin)$ and $t_{\xi}(I_1)$ 
for given $\alpha$, $\ac$, $\beta$ and $\xi$.

The dotted lines in Figure~\ref{fig:MinInform} show for $\xi=1.75$ and $\xi=2$ in dependence on $t_{\xi}(I_1)$, where $t_{\xi}(\ionemin)<t_{\xi}(I_1)\leq t_{\xi}(\ionemax)$, the minimum relative information $t_{\xi}(\itwomin)$ required with the standard fast track process with required conditional registration to have overall power $1-\beta=0.8$. The other lines will be described and discussed in the next section. The R-code used for the calculation of $t_{\xi}(\itwomin)$ and for all subsequent images can be downloaded from \url{https://github.com/LianeKluge/Adaptive-Designs-in-Fast-Track-Registration-Processes}. As it can be seen from Figure~\ref{fig:MinInform}, when $I_1$  comes close to the critical value $\ionemin$, the second stage information becomes large. This is because when $I_1$ comes close to the critical value, the power for the conditional registration comes close to $1-\beta$. Then, the second stage information required for the anticipated overall power becomes larger to compensate for the requirement to also reject $H_0$ at the second stage at level $\alpha$. This behavior applies to any design for which $Z_1\ge z_f$ is a requirement for a final success. We note that for the standard design with two separate stages, the requirement for a specific overall power demands that the minimum information for the second stage is always greater than $I_{\delta}$, i.e.\ the sample size required for a fixed-size sample test with a power $1-\beta$ at the a priori assumed effect $\delta$. For our example fast-track procedure we have for $\xi=2$ due to the choice of $n_1=63$ that $t_{\xi}(I_1)=n_1/n_{\delta}=63/105=0.6$. According to Figure~\ref{fig:MinInform} (right) the minimum needed relative information for the second stage corresponding to $t_{\xi}(I_1)=0.6$ is $1.41=t_{\xi}(\itwomin)=n_{2,\min}/n_{\delta}$. That means the minimum needed per group sample size for the second stage is $n_{2,\min}=1.41\cdot 105=148$.

In the next subsection, we will discuss adaptive designs for fast-track procedures, which have some important advantages. Compared to the non-adaptive design, where the overall level of achieving both the conditional and the permanent registration is $\af\cdot\alpha$, the overall level is $\alpha$ for adaptive designs. This leads to greater efficiency when aiming to achieve a certain overall power. Moreover, in contrast to the standard design with separate studies, the data from the first stage can also be used for the final confirmatory hypothesis test of $H_0$. This increases efficiency even more. The performance of two adaptive designs is teased in Figure~\ref{fig:MinInform} by the dashed line and the solid line.
\begin{figure}[h]
\includegraphics[width=1\textwidth]{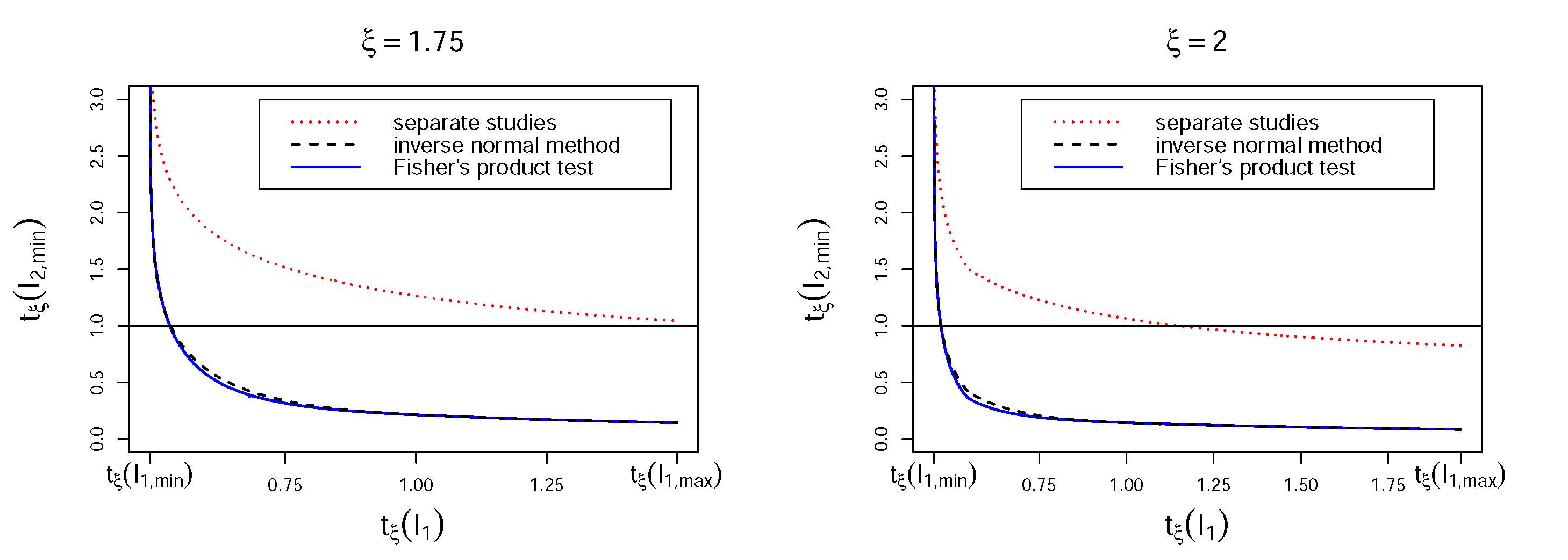}
\caption{Minimum needed information for the second stage to achieve overall power of $1-\beta=0.8$ for $\xi=1.75$ (left) and $\xi=2$ (right) for the common non-adaptive design and the adaptive designs, inverse normal method and Fisher’s product test.}
\label{fig:MinInform}
\end{figure}

\subsection{Adaptive Study Designs for Fast-Track Processes}\label{subsectionAdaptiveDesigns}
We will start with a brief recap of confirmative adaptive designs. Confirmatory (two-stage) adaptive designs are organized into two sequential stages of data collection, with an interim analysis at the end of the first stage (see, for instance,  the works of Wassmer and Brannath \cite{BuchWassmerBrannath2016} and Bretz et al \cite{ueberblickAdaptiveDesignsBretz}). The interim data can be used to decide whether to terminate or continue the trial. The study can be terminated due to futility (with an acceptance of $H_0$) or due to success, i.e.\ early rejection of a hypothesis. 
So far, adaptive study designs play an important role in clinical drug development (see, for example, the works of Pallmann et al \cite{pallmann2018adaptive} and Cerqueira et al \cite{cerqueira2020adaptive}).  
They allow the user to learn from the data of an ongoing trial by performing design modifications at interim when the trial is not terminated. Thus, adaptive designs can be considered as extensions of group sequential designs \cite{BuchWassmerBrannath2016}. Possible design modifications at interim include sample size re-estimation, selection of a primary endpoint, or dropping of treatment arms in combined phase II/III trials while controlling the type I error rate. Moreover, the data from both stages can be used for the final confirmatory hypothesis test. That means, in contrast to the common situation of an external pilot study that prepares an independent confirmatory study, adaptive designs permit the use of the pilot data for the final confirmatory hypothesis test. As we will see in this section, in the context of fast-track procedures, the use of adaptive designs can increase efficiency substantially.

There are various strategies for constructing adaptive tests, for instance, the \emph{combination test} approach introduced by Bauer \cite{bauer1989multistage} and Bauer and Köhne \cite{bauer1994evaluation} and the \emph{conditional error function} approach introduced by Proschan and Hunsberger \cite{proschan1995designed}. There is also a direct relationship between them \cite{BuchWassmerBrannath2016}. In this paper, we work with the conditional error function approach. The principle is as follows: 
Significance levels with $0\le \alpha_1<\alpha_0\le 1$ for early rejection and acceptance, respectively, may be predefined. Furthermore, 
a non-decreasing function $A:\mathbb{R}\rightarrow[0,1]$, referred to as the \emph{conditional error function} must be specified in advance. Let $Z_k=\hat{\vartheta}_k\sqrt{I_k}\sim \mathcal{N}(\vartheta\sqrt{I_k},1)$ denote the z-statistic calculated from the stage $k$ data for testing a hypothesis $H_0$. If after the first stage $Z_1\geq\Phi^{-1}(1-\alpha_1)$, the trial is stopped and $H_0$ is rejected. If $Z_1<\Phi^{-1}(1-\alpha_0)$, the trial is stopped due to futility with an acceptance of $H_0$. In the case $\Phi^{-1}(1-\alpha_0)\leq Z_1< \Phi^{-1}(1-\alpha_1)$ the trial is continued to stage 2 and then $H_0$ is rejected iff $Z_2\geq \Phi^{-1}(1-A(Z_1))$. With our assumptions under $H_0$, the distribution of $Z_1$ and the conditional distribution of $Z_2$ given $Z_1$ are stochastically smaller than or equal to the standard normal distribution. Thus, if the trial is continued to the second stage, the conditional rejection probability is less than or equal to $A(Z_1)$. If $Z_1\geq\Phi^{-1}(1-\alpha_1)$ is observed after the first stage, the conditional rejection probability $A(Z_1)$ can be set to $1$ and in the case of stopping the trial due to a \emph{binding futility stopping rule} to  $0$. The Type I error rate for $H_0$ is under control if this conditional error function fulfills
 \begin{align}
    \int_{-\infty}^{+\infty}\, A(z_1)\phi(z_1)\mathrm{d}z_1 =\alpha_1+\int_{\Phi^{-1}(1-\alpha_0)}^{\Phi^{-1}(1-\alpha_1)}A(z_1)\phi(z_1)\mathrm{d}z_1\leq \alpha~.\label{generalConditionCef}
\end{align} 
For flexible decision making with regard to the futility stopping, stopping with $Z_1<\Phi^{-1}(1-\alpha_0)$ can be ignored by not setting $A$ to 0 in this case, or equivalently, replacing the lower stopping boundary with~$-\infty$
in \eqref{generalConditionCef}. The study can still be stopped with acceptance of $H_0$. A futility level $\alpha_0<1$ that is ignored in (\ref{generalConditionCef}) is often called \emph{non-binding} futility level, and the resulting flexibility comes with the price of reduced power (if the futility stop is executed). In this paper, we consider scenarios with binding (Section~\ref{NewChapterStudyDesignsFastTrack}) and non-binding (Section~\ref{ChapterNonBindingFutility}) futility stopping. 

The Type I error rate is still under control if the design for the first stage is determined or modified after the first stage based on the interim data, given that the conditional distribution of the resulting second stage z-statistic $Z_2$ given $Z_1$ is stochastically smaller than or equal to the standard normal distribution under $H_0$. This follows from what is sometimes denoted as \emph{conditional invariance principle}; see e.g. the work of Brannath et al \cite{brannath2007multiplicity}. 

One example of a design modification is adjusting the second stage sample size, or more generally, the second stage information. A typical goal of this is to achieve a sufficiently high conditional probability, i.e.\ a certain \emph{conditional power}, of rejecting $H_0$ when the trial is continued to the second stage. 
Thus, in the context of a conditional DiGA registration, a high conditional power at the end of the first stage should be a convincing argument for a successful subsequent stage in proving that the DiGA has a positive healthcare effect. It thereby meets the corresponding requirement mentioned on pages~109-110 in the DiGA guideline \cite{leitfadendiga}.

The conditional power is the conditional probability of rejecting $H_0$ in the second stage, conditional on the first stage data. For the true effect  $\vartheta$, the conditional power is given by 
\begin{align*}
    \mathbb{P}_{\vartheta}\left(Z_2\geq \Phi^{-1}(1-A(z_1))\right)=1-\Phi\left(\Phi^{-1}(1-A(z_1))-\vartheta\sqrt{I_2}\right)~.
\end{align*}
Since the true effect is unknown, the sample size or information of the second stage can be calculated such that a conditional power of $1-\beta$ is achieved at the observed estimate $\estone$ or noncentrality parameter $z_1$ of the first stage. There are other strategies for adjusting the information used; see, for example, Section~7.4 in the work of Wassmer and Brannath \cite{BuchWassmerBrannath2016}.

We will now describe how two-stage adaptive designs can be used in the two-stage fast-track procedures. For the application for conditional registration, the data from the first stage is used, and for the application for permanent registration, the data from the first and second stages are used. We consider the case where there is no possibility of an early rejection after the first stage, which means we set $\alpha_1:=0$. We demand that the conditional error function $A$ fulfills $A(z_1)\leq 0.5$ for all $z_1\in\mathbb{R}$. This ensures that $H_0$ can only be rejected after the second stage when $Z_2\geq 0$, i.e. a non-negative effect $\hat{\vartheta}_2\geq 0$ is observed. After the second stage, the hypothesis $H_0$ is rejected iff $Z_2\geq \Phi^{-1}(1-A(Z_1))$. We will use two different kinds of conditional error functions. The first one is a slight modification of the conditional error function of the \emph{inverse normal method} with equal weights \cite{lehmacher1999adaptive}. The second considered conditional error function is a slight modification of the conditional error function of \emph{Fisher's product test} \cite{fisher1932statistical} (proposed in the works of Bauer \cite{bauer1989multistage} and Bauer and Köhne \cite{bauer1994evaluation} for use in adaptive designs): We define for $z_0\in\mathbb{R}\cup\{-\infty\}$
\begin{align}
A_{I,z_0}(z_1):=
    \begin{cases}
			\bigg\{1-\Phi\left(\frac{\Phi^{-1}(1-c_{I}(z_0))-w_1 z_1}{w_2}\right)\bigg\} \wedge 0.5, & \text{if $z_1\geq z_0$}\\
            0, & \text{otherwise}
		 \end{cases}\label{cefInverNormal}
\end{align}
where $w_1=w_2=\sqrt{0.5}$ and 
\begin{align}
A_{F,z_0}(z_1):=
    \begin{cases}
			\frac{c_F(z_0)}{1-\Phi(z_1)} \wedge 0.5, & \text{if $z_1\geq z_0$}\\
            0, & \text{otherwise}
		 \end{cases}~.\label{cefFisher}
\end{align}
The constants $c_{I}(z_0)$ and $c_F(z_0)$ are chosen such that they are nonnegative and fulfill (with $A:=A_{I,z_0}$ or $A:=A_{I,z_0}$) condition (\ref{generalConditionCef}). The chosen $z_0$ can be interpreted as a binding futility bound on the z-scale with corresponding early stopping level $\alpha_0=1-\Phi(z_0)$. If there is no possibility for permanent registration after a non-successful conditional registration, then the binding futility boundary $z_0=z_f$ is a logical choice. This implies that the considered conditional error functions depend on $t_{\xi}(I_1)$, $\xi$ (and the fixed $\eta_f$) because $z_f$ depends on these parameters.

Similar to Section~\ref{subsectionConsequences}, we will investigate the overall power. The requirements for conditional registration remain unchanged: $\hat{\vartheta}_1\geq\delrel$ and $p_1\leq\ac$ (below $\ac=0.15$) needs to be shown. We also assume that after the first stage, the second stage is powered such that a conditional power of $1-\beta$ is achieved in the observed noncentraility parameter $z_1$ of the first stage. Thus, we define 
\begin{align}\label{adaptiveStudyDesigInformQuot}
    I_2^{A}(z_1):=I_2^{A}(z_1,\itwomin,I_1):=\max\left\{\itwomin,{I_1\left\{\Phi^{-1}(1-\beta)+\Phi^{-1}\big(1-A(z_1)\big)\right\}^2}/{z_1^2}\right\}, \quad z_1\geq\zf,
\end{align}
where the minimum information $\itwomin$ remains to be defined. Because $H_0$ is tested at the second stage by a z-test at level $A(z_1)$, we obtain the overall power 
\begin{align}\label{allgemeineConditionOverallPowerAdaptivD}
\nonumber
\Succ_A(\itwomin,I_1) & :=\mathbb{P}_{\delta}\left(Z_1\geq z_f,\text{reject~}H_0\right)\\
& =\int\limits_{z_f}^{\infty}
\left\{
1-\Phi\left(\Phi^{-1}(1-A(z_1))-\eta_f{\sqrt{I_2^{A}(z_1)}}/{\sqrt{I_{\delta}}}\right)
\right\}
\phi\big(z_1-\sqrt{I_1}\delta\big) d z_1~.
\end{align}

Similar to Subsection~\ref{subsectionConsequences}, for achieving an overall power of at least $1-\beta$, the information $I_1$ must be strictly greater than $\ionemin$. Furthermore, it can be shown that for $\delta >\delrel$ and $\eta_f$ and all $I_1>\ionemin$ with $\Succ_A(0,I_1)\le 1-\beta$ there exists a unique minimum information $\itwomin$ such that $\Succ_A(\itwomin,I_1)=1-\beta$. For $I_1$ with $\Succ_A(0,I_1)\geq 1-\beta$, there is no minimum information needed for the second stage. In the latter case, we set $\itwomin=0$. For proof, we refer to the Appendix. For our example fast-track procedure, the calculation of the second-stage sample size can be done similarly to (\ref{formulaSecondStageSampleSizeExampleSeparate}) by replacing $\alpha$ with $A(z_1)$.

The dashed and solid lines in Figure~\ref{fig:MinInform} show the minimum required sample sizes relative to $I_\delta$, i.e.\ $t_{\xi}(\itwomin)$, for the two different $\xi$
with the conditional error functions $A_{I,z_f}$ and $A_{F,z_f}$, respectively. It can be seen that, with regard to the minimum needed information for the second stage, the control of the overall power ($1-\beta=0.8$) at $\delta=\xi\delrel$ comes for the adaptive designs with a much smaller requirement for the minimum information than for the non-adaptive design. For the adaptive designs, the minimum second stage information becomes quickly less than the information $I_\delta$ required for the fixed-sample size test and is less than 50\% of this information if the interim information is $0.55\cdot I_\delta$ or more. For the example fast-track procedure, recall that for $\xi=2$ the relative information  $t_{\xi}(I_1)=n_1/n_{\delta}=63/105=0.6$ corresponds to the chosen first stage sample size per group $n_1=63$. According to Figure~\ref{fig:MinInform} (right), the minimum relative information needed for the second stage is now $t_{\xi}(\itwomin)=0.33$ and $t_{\xi}(\itwomin)=0.29$ for the inverse normal method and Fisher's product test, respectively. This implies $n_{2,\min}=0.33\cdot n_{\delta}\approx 35$ and $n_{2,\min}=0.29\cdot n_{\delta}\approx 31$, respectively. These sample sizes are much lower than $n_{2,\min}=148$ needed for the non-adaptive design.

\subsection{Comparison of different study designs}\label{subsectionComparisonDesigns}
We now want to compare the efficiency of conducting an adaptive design (with conditional error function $A\leq 0.5$) to a design with two separate stages with regard to the expected and maximum information under the same assumptions as in the previous sections. In particular, this implies that for a given interim information $I_1$, the applicant achieves the overall power of $1-\beta$ at the a priori assumed effect $\delta=\xi\cdot\delrel$ by choice of the required minimum second stage information $\itwomin$. We will compare the expected second stage information in dependence of $I_1>\ionemin$ assuming that the true effect is~$\delta$. 
Like for the first stage, expected and maximum informations will be compared for different $\xi$ and again relative to $I_\delta$, i.e.\ on the $t_\xi$-scale for a most general comparison which only depends on $\ac$, $\alpha$, $\beta$ and $\xi$.

Since the second stage information \eqref{adaptiveStudyDesigInformQuot} is decreasing in $A(z_1)$, which itself is a non-decreasing function in $z_1$, the maximum sample size is achieved for $z_1=z_f$ or equals $\itwomin$.

Figure~\ref{MeanInformationBothXi} shows for $\xi=1.75$ and $\xi=2$ the (relative) mean information for the second stage for the two adaptive designs with conditional error functions $A_{I,\zf}$ and $A_{F,\zf}$ and the non-adaptive design. It can be seen that the adaptive designs have for all $t_{\xi}(I_1)$ with $t_{\xi}(\ionemin)<t_{\xi}(I_1)\leq t_{\xi}(\ionemax)$ a lower (relative) mean information. Furthermore, for the adaptive designs for the most $t_{\xi}(I_1)$, the relative mean information is (clearly) less than $1$. Thus, the mean information needed for the second stage is less than the information needed for a (single) fixed-size sample test with power $1-\beta$ (direct application for permanent registration). Furthermore, it can be numerically calculated that the mean information over both stages is not much larger than $1$. The corresponding Figures can be found in the Supplement.

\begin{figure}[h]
\includegraphics[width=1\textwidth]{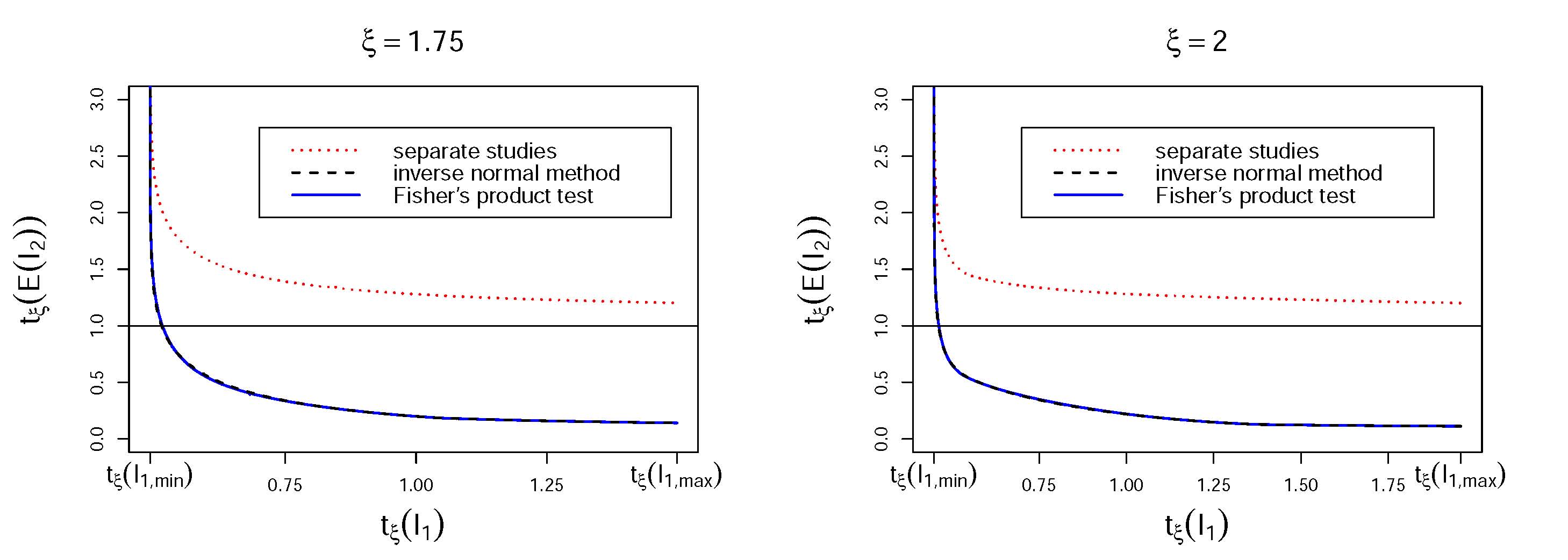}
\caption{Visualization of the mean information for the second stage for $\xi=1.75$ (left) and $\xi=2$ (right) for adaptive designs and a non-adaptive design.}\label{MeanInformationBothXi}
\end{figure}

The maximum informations of the considered designs are shown in Figure~\ref{fig:MaximumInformationBothXi}. Figure~\ref{MeanInformationBothXi} and Figure~\ref{fig:MaximumInformationBothXi} have in common that if $t_{\xi}(I_1)$ comes closer to $t_{\xi}(\ionemin)$, the mean and maximum information for the second stage become large and tend to infinity. This is
because the minimum information $t_{\xi}(\itwomin)$ tends to infinity when $t_{\xi}(I_1)$ comes closer to the critical value $t_{\xi}(\ionemin)$ as described in the previous subsections. As can be seen from Figure~\ref{fig:MaximumInformationBothXi} (right) for $\xi=2$, the maximum information of all designs has a linear increasing part between $0.45$ and $0.55$. This is because, for $\xi=2$ for all $t_{\xi}(I_1)\leq 0.55$, the border $\zf$ is the same: It holds $\zf=1-\Phi^{-1}(1-\ac)= 1.04$ (i.e. the demand $p_1\leq\ac$ is stricter than having a numeric relevant estimation). This implies that the maximum information is a locally linear function of $t_{\xi}(I_1)$. For $\xi=1.75$, there is no such interval on which $\zf$ is constant. Thus, there is no increasing part in the maximum information in Figure~\ref{fig:MaximumInformationBothXi} (left). For both adaptive designs for $\xi=1.75$ and $\xi=2$, the relative maximum information remains constant for $t_{\xi}(I_1)>1.06$ and $t_{\xi}(I_1)>1.38$, respectively, because the conditional error functions are then constantly equal to $0.5$. Figure~\ref{fig:MaximumInformationBothXi} shows that the adaptive designs are also convincing with regard to the maximum information when compared with the non-adaptive design. We analyze this result.
\begin{figure}[h]
\includegraphics[width=1\textwidth]{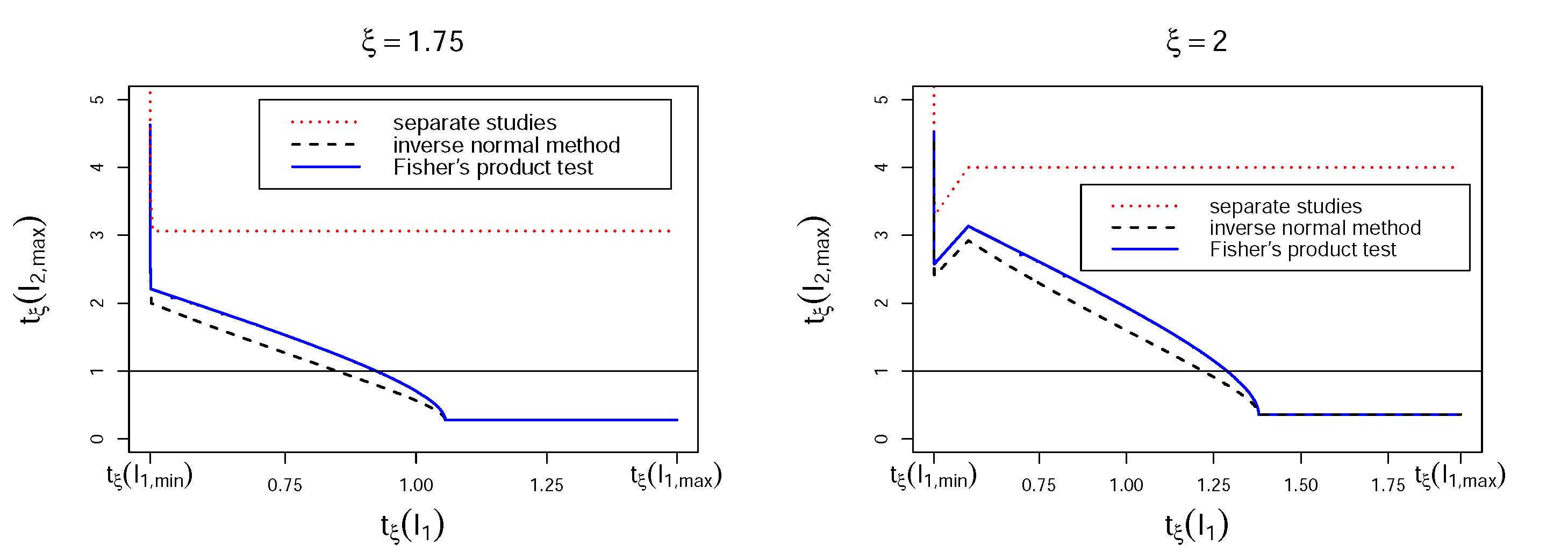}

\caption{Visualization of the maximum information for the second stage for $\xi=1.75$ (left) and $\xi=2$ (right) for adaptive designs and a non-adaptive design.}\label{fig:MaximumInformationBothXi}
\end{figure}

Recall that the maximum sample size is achieved for $z_1=\zf$ or equals $\itwomin$. It seems obvious that an adaptive design with conditional error function $A$ has lower maximum information than the non-adaptive design if $A(\zf)>\alpha$ is fulfilled. Nevertheless, since the maximum information can also be equal to $\itwomin$, a formal proof is needed (see the Appendix). For our considered conditional error functions for $\xi=1.75$ and $\xi=2$ for all $I_1>\ionemin$ it holds indeed $A_{I,\zf}(\zf)>\alpha$ and $A_{F,\zf}(\zf)>\alpha$ which gives a formal explanation for the lower maximum informations in Figure~\ref{fig:MaximumInformationBothXi}. It is important to note that the binding futility bound enables the possibility to construct conditional error functions $A$ which fulfill for all $z_1\geq\zf$ that $A(z_1)>\alpha$ and thus $A(\zf)>\alpha$ while having the Type I error controlled at level $\alpha$. In the most extreme case, the conditional error function, which is constant equals ${\alpha}/{\af}>\alpha$, can be chosen with the binding futility bound $Z_1<\zf$. This design would probably only be accepted if the authorities were informed about it before the start of data collection of the first stage.

\begin{example}[continued]
For our example fast-track procedure for the a priori assumed $\xi=2$ and the considered first stage sample size of $n_1=63$ (i.e. $t_{\xi}(I_1)=0.6$) Figure~\ref{fig:MaximumInformationBothXi} (right) yields $t_{\xi}(I_{2,\max})=4,2.78,3$ for the non-adaptive design, the inverse normal method and Fishers' product test, respectively. Using the relation $t_{\xi}(I_{2,\max})=n_{2,\max}/n_{\delta}$ and the results from the previous sections we have that the sample size per group for the second stage lies between: $n_{2,\min}=148$ and $n_{2,\max}=420$ (non-adaptive design), $n_{2,\min}=35$ and $n_{2,\max}=292$ (inverse normal method), $n_{2,\min}=31$ and $n_{2,\max}=315$ (Fisher's product test). Figure~\ref{MeanInformationBothXi} (right) yields $t_{\xi}(E(I_2))=1.41$ (non-adaptive design) and $t_{\xi}(E(I_2))=0.47$ (adaptive designs). This implies a second stage mean sample size per group of $E(n_2)=1.41\cdot n_{\delta}=149$ (non-adaptive design) and $E(n_2)=0.47\cdot n_{\delta}\approx 50$ (adaptive designs). The average sample sizes of the second stage are close to the minimum sample sizes, which implies that the distribution of the sample size is left-skewed in the sense that sample sizes close to the maximum ones are rare.
\end{example}
\section{Combination of conditional and permanent registration}\label{ChapterNonBindingFutility}
Recall from Section~\ref{ConsequencesRequirementsBindCondRegits} that, under the assumption that the applicants cannot apply for a permanent registration after an unsuccessful attempt to obtain conditional registration, the fast-track procedure does not make sense for all $\xi\le \xi_{\min}=1.43$. Then, the overall success probability of achieving permanent registration is given by $\mathbb{P}_{\delta}(\text{reject~}H_0)=\mathbb{P}_{\delta}(Z_1\geq\zf, \text{reject~}H_0)$, and we have seen that for $\xi\le \xi_{\min}=1.43$ there is no first-stage information $I_1$ for which an overall success probability of $1-\beta=0.8$ can be achieved, while the requirements for a conditional registration are less strict than those for a direct permanent registration (i.e., $\ionemin>\ionemax$). Motivated by this, we relax the assumption that a conditional registration is a requirement for permanent registration. We assume that, based on the pilot data, we can decide to either apply for conditional registration or abstain from this application and continue directly with the second stage for permanent registration. We will work again with $\ac=0.15$. This Section presents possible designs and design inspirations for all a priori effect assumptions $\xi\geq 1$, in particular for an a priori assumed (or true) effect of $\delta/\delrel\leq\xi_{\min}$.

\subsection{Suggestion for a combination strategy}

A reasonable strategy for this is to apply for conditional registration whenever $Z_1\geq \zf$ (since the requirements are met), and otherwise to abstain from this application and continue the two-step approach with the second stage for permanent registration. This process is illustrated in Figure~\ref{twoArms}. In order to meet the overall type I error rate $\alpha$, we utilize adaptive design methodology and use for both cases an overall conditional error function $A$ satisfying \eqref{generalConditionCef}. We will discuss possible choices for $A$ in the following. 

\begin{figure}[h]
\centering
\resizebox{!}{3.5cm}{
\begin{tikzpicture}
\node[rectangle, draw=none, minimum size=1pt] at (1, 0.5) {\bf 1.\ Stage};
\node[rectangle, draw=none, minimum size=1pt] at (7.5, 2.3) {\bf 2.\ Stage};
\draw[-{Latex[length=3mm, width=4mm]}] (0,0) -- (3,0);
\filldraw[black] (3.2,0) circle (4pt);
\draw[-{Latex[length=3mm, width=4mm]}] (3.4,0) -- (6,1) node[pos=0.4, above, sloped] {$Z_1\geq\zf$};
\draw[-{Latex[length=3mm, width=4mm]}] (3.4,0) -- (6,-1) node[pos=0.4, below, sloped] {$Z_1<\zf$};
\node[rectangle, draw, align=left] at (12, 1) {
\hspace{1em}$I_2(z_1)=\max\left\{\itwomin,I_1\left\{\Phi^{-1}(1-\beta)+\Phi^{-1}(1-A(z_1))\right\}^2/{z_1^2}\right\}$ \\[.25em] 
\hspace{1em}where $\itwomin$ such that \\[.25em]
\hspace{11em}$\mathbb{P}_{\delta}\left(
Z_2\ge \Phi^{-1}\big(1-A(Z_1)\big)\middle|Z_1\geq\zf\right)\ge 1-\beta$\hspace{1em}\ };
\node[rectangle, draw, align=left] at (12, -1) {
\hspace{1em}$I_{2}(z_1)=\itwoconst$ where $\itwoconst$ such that \\[.25em]
\hspace{11em}$\mathbb{P}_{\delta}\left(
Z_2\ge \Phi^{-1}\big(1-A(Z_1)\big)\middle|Z_1<\zf\right)= 1-\beta$\hspace{1em}\ };
\end{tikzpicture}}
\caption{Combination strategy with a data-driven decision to either apply or abstain from the application for conditional registration, where the overall type I and II errors are under control.}\label{twoArms}
\end{figure}
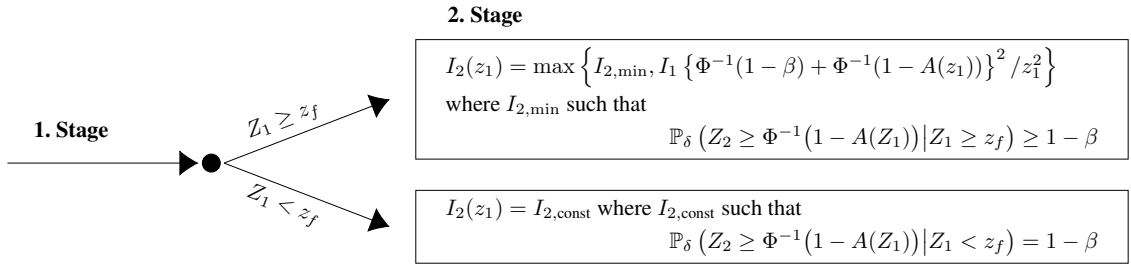

As indicated in the lower branch of Figure~\ref{twoArms}, we use for $Z_1<z_f$ a constant second-stage sample size $I_{2,\text{const}}$, namely the one that provides us with the conditional success probability of $1-\beta$ given $Z_1<z_f$, for the a priori assumed effect $\delta=\xi\delrel$.  We are not aiming for a specific conditional power in this case, because this is costly in terms of sample size, particularly if $Z_1<z_f$, and there is no conditional registration that needs to be extended with a sufficiently high (conditional) probability. We are using the a priori assumed effect $\delta$, because a mid-trial effect re-estimation appears questionable without conditional power control. 

In order to also achieve an overall success probability  of at least $1-\beta$ at $\delta$ (i.e.\ over both branches in Figure~\ref{twoArms}), we need to have a conditional success probability of at least $1-\beta$ also in the upper branch, i.e.\ we need to require that 
\begin{align} \label{UpperBranchPower}
\mathbb{P}_{\delta}\left(Z_2\ge \Phi^{-1}\big(1-A(Z_1)\big)\middle|Z_1\geq\zf\right)\ge  1-\beta\,.
\end{align}
Moreover, if $Z_1\geq \zf$, we choose, as in Section~\ref{NewChapterStudyDesignsFastTrack}, the second-stage sample size such that the conditional power $1-\beta$ is reached at the estimated effect estimate. As in Section~\ref{NewChapterStudyDesignsFastTrack}, condition \eqref{UpperBranchPower} can be satisfied with a sufficiently large minimum second stage information $\itwomin$. In contrast to the situation in Section~\ref{NewChapterStudyDesignsFastTrack}, we find now for all $\xi>  1$ such minimal second-stage information. This is because, for all $\delta>0$, the left-hand side of \eqref{UpperBranchPower} converges to 1 with increasing second-stage information.  

In the Appendix, it is shown that, for a design with conditional error function $A$ for all $\xi\geq 1$ and $I_1>0$, there exists a unique minimal information $\itwomin$ such that (\ref{UpperBranchPower}) can be met with equality if $\Succ_A(I_1,0)\leq (1-\beta)\cdot\mathbb{P}_{\delta}(Z_1\geq\zf)$, where $\Succ_A$ is as in \eqref{allgemeineConditionOverallPowerAdaptivD}. If $\Succ_A(I_1,0)> (1-\beta)\cdot\mathbb{P}_{\delta}(Z_1\geq\zf)$, condition (\ref{UpperBranchPower}) is achieved with an arbitrarily small information for the second stage. In this case, we can set $\itwomin=0$. 

We note that the use of the constant conditional error function, $A(z_1)=\alpha$ for all $z_1$, corresponds to the standard fast-track process if $Z_1\geq\zf$. For $Z_1<z_f$, it results in a new (second-stage) study for permanent registration with information $I_{2,\text{const}}=I_{\delta}$ and without using the pilot data. A potentially more efficient approach that utilizes the pilot data also for the permanent registration (in both cases) is to use a non-constant conditional error function like $A_{I,-\infty}$ or $A_{F,-\infty}$. Note that these conditional error functions should not include the binding futility stopping rule $Z_1<\zf$, because we are testing now $H_0$ also if $Z_1<\zf$. 

The proposed strategy of performing no adaptations at interim and so use a constant information $\itwoconst$ for the second stage whenever $Z_1<\zf$, particularly suggests to 
use a fixed-size sample z-test with 
the cumulative sample of both stages. That means, if $Z_1<\zf$ is observed, $H_0$ is rejected if 
\begin{align}
    Z=\sqrt{I_1/(I_1+\itwoconst)}Z_1+\sqrt{\itwoconst/(I_1+\itwoconst)}Z_2\geq \Phi^{-1}(1-\alpha)~.\label{fixedSizeSampleTest}
\end{align}

It can easily be seen (and is well known) that (\ref{fixedSizeSampleTest}) is equivalent to $Z_2\geq\Phi^{-1}(1-\tilde{A}_{Z,\alpha}(Z_1))$ for the specific inverse normal conditional error function 
$$\tilde{A}_{Z,\alpha}(z_1):=1-\Phi\left(\frac{\Phi^{-1}(1-\alpha)-\sqrt{I_1/(I_1+\itwoconst)}z_1}{\sqrt{\itwoconst/(I_1+\itwoconst)}}\right),\quad z_1\in\mathbb{R},$$ 
which satisfies level condition (\ref{generalConditionCef}). To complete the design, for the upper branch, the conditional error function must be extended in such a way that the type I error rate does not exceed $\alpha$. Here we choose $\tilde{A}_{Z,\alpha^{\prime}}$ where $\alpha\le \alpha^{\prime}\le 1$ is the largest possible value such that the Type I error rate of
\begin{align}\label{definition_az}
    A_Z(z_1):= 0.5\wedge\begin{cases}
			\tilde{A}_{Z,\alpha^{\prime}}(z_1), & \text{if $z_1\geq \zf$}\\
            \tilde{A}_{Z,\alpha}(z_1), & \text{if $z_1<\zf$}~.
		 \end{cases}
\end{align}
is equal to or less than $\alpha$ (i.e.\ (\ref{generalConditionCef}) holds true).
For all examples considered in the following, it can be seen that the significance level $\alpha$ is fully used, i.e. for $A_Z$ equality holds in (\ref{generalConditionCef}). Note that the proposed strategy is based on the Müller and Schäfer method \cite{muller2001adaptive} with the difference that the conditional error function $A_Z$ is truncated here by $0.5$.
In line with the proposed strategy in Figure~\ref{twoArms}, we assume that $\itwomin$ and $\itwoconst$ are chosen such that the aimed conditional success probabilities specified in Figure~\ref{twoArms} are achieved.

In the following, we investigate and compare the performance of the conditional error functions $A_Z$, $A_{I,-\infty}$, and $A_{F,-\infty}$ and the non-adaptive design.

\begin{figure}[h]
\includegraphics[width=1\textwidth]{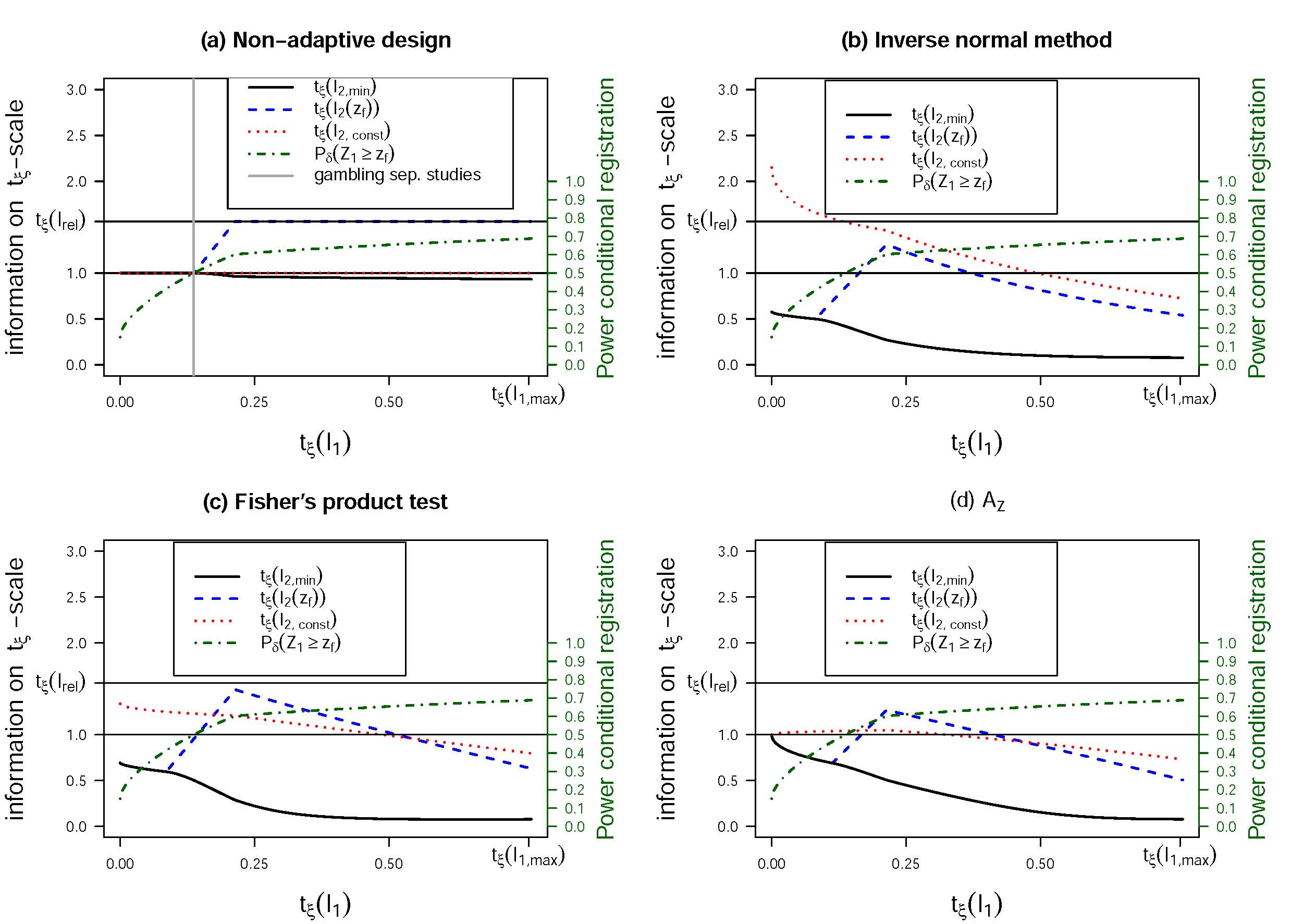}
\caption{Comparison of minimum and maximum information for the second stage of the suggested combination strategy with a non-adaptive and the adaptive designs when $\alpha_c=0.15$ and $\xi=1.25$. \label{fig:BS}}
\end{figure}

\subsection{Investigation of the second stage information $I_{2,\text{const}}$ for the case $Z_1<z_f$}\label{subsectionInvestigationI_2const}

The dotted lines in the plots (a) to (d) of Figure~\ref{fig:BS} show for the case $Z_1<z_f$ the constant second-stage information $t_\xi(I_{2,\text{const}})$ in dependence of $t_\xi(I_1)$ when
$\xi=1.25$ and $\alpha_c=0.15$. The plots are for the non-adaptive design  ($A\equiv\alpha$) and the adaptive designs with conditional error functions $A_{I,-\infty}$, $A_{F,-\infty}$ and the specific inverse normal conditional error function $A_Z$ that is line with $I_{2,\text{const}}$ defined in \eqref{definition_az}, respectively. In the Supplement, a figure is provided that allows a direct comparison of $t_{\xi}(\itwoconst)$ between the designs. For the non-adaptive design we get $I_{2,\text{const}}=I_{\delta}$ for all $t_\xi(I_1)$. We can see that for the adaptive designs, the information $I_{2,\text{const}}$ decreases with $t_\xi(I_1)$ and becomes smaller than $I_{\delta}$ if $t_\xi(I_1)>0.5$  (for $A_{I,-\infty}$ and $A_{F,-\infty}$) or $t_\xi(I_1)>0.33$ (for $A_{Z}$). 

For the numerical determination of $\itwoconst$, we used a bisection search. This requires that the conditional success probability is non-decreasing in $\itwoconst$. For the non-adaptive designs and for the adaptive designs with $A_{I,-\infty}$ and $A_{F,-\infty}$, this can easily be seen. For $A_Z$, the monotonicity is proven in the Appendix.

\begin{example}[continued]
For our example fast-track procedure, we will work with $\delrel=1.4$ from now on, which is also reasonable based on Pöttgen et al \cite{pottgen2018randomised}. We further assume that the planning is now based on the results of Van Kessel et al \cite{kesselPlaningStudy}, which were used for the planning of the confirmatory study of the DiGA elevida. Then, a reasonable and conservative a priori assumption for the healthcare effect is $\delta=1.75=1.25\cdot\delrel$, i.e. $\xi=1.25$. Additionally, we continue to assume that a parallel group design for both conditional and permanent registration, each with balanced samples, is used. Now, a information of $I_{\delta}:=\eta_f^2/\delta^2=2.8^2/1.75^2=2.56$, i.e. a per group sample size of $n_{\delta}:=2\hat{\sigma}^2 I_{\delta}\approx 137$, is required for a fixed-size sample test with significance level $\alpha=0.025$ and power $1-\beta=0.8$. This corresponds to the per group sample size of the study conducted for permanent registration of the DiGA elevida \cite{pottgen2018randomised}. If now $0.5=t_{\xi}(I_1)=n_1/n_{\delta}$, i.e. $n_1=0.5\cdot137=69$ would be the chosen first stage sample size per group, Figure~\ref{fig:BS} yields the sample size for the final stage without conditional registration $n_{2,\mathrm{const}}=t_{\xi}(I_{2,\mathrm{const}})\cdot n_{\delta}=1\cdot n_{\delta}=137$ with the non-adaptive design, the inverse normal method and Fisher's product test, or $n_{2,\mathrm{const}}=0.9\cdot n_{\delta}\approx 124$ with $A_Z$.
\end{example}

\subsection{Investigation of the second stage information for the case $Z_1\ge z_f$}
Figure~\ref{fig:BS} shows the minimum second-stage information $t_{\xi}(I_{2,\min})$ by the solid lines and the maximum second-stage information $t_{\xi}(I_{2,\max})$ by the dashed lines for the different designs for the case of a successful conditional registration ($Z_1\ge z_f$).
The maximum second-stage information is given by the adaptive second-stage information $I_2^{A}(z_f)$ (set $A\equiv\alpha$ for the non-adaptive design)  when it exceeds $I_{2,\min}$. These lines are linearly increasing in $t_\xi(I_1)$ as long as $\zf=\Phi^{-1}(1-\alpha_c)$, and they are constant or decreasing afterwards. 

With a small interim information, it can happen that the maximum information $I_2^{A}(\zf)$ becomes equal to the minimum information, in which case we need to use a constant second stage information also in the case $Z_1\ge z_f$, namely $I_{2,\min}$ (solid line). For the non-adaptive design ($A\equiv \alpha$), this constant information equals $I_{\delta}$. Since $I_{2,\text{const}}=I_\delta$ also for $Z_1<z_f$,  there is no difference between the lower and upper case in Figure~\ref{twoArms}.
As a consequence, the pilot study is only done to ``gamble'' for a conditional registration. With $\xi=1.25$
and $\alpha_c=0.15$, this happens for all $t_{\xi}(I_1)\leq 0.137$. Such a ``gambling'' with rather small pilot data appears to be a quite questionable strategy. The gambling area is indicated in plot (a) of Figure~\ref{fig:BS} by the vertical solid line. For the considered adaptive designs, a constant second-stage information after successful conditional registration occurs for smaller pilot informations than in the non-adaptive case. For the adaptive designs, the upper and lower cases can differ at least with regard to the constant second-stage information, and the pilot data are always used in the application for permanent registration. Hence, the pilot data are never used only to ``gamble'' for a conditional registration.


\subsection{Overall comparison with regard to the second stage information}

Obviously, for the case $Z_1\ge z_f$, the data-driven second-stage information is always in the region between the solid and dashed lines in Figure~\ref{fig:BS}. One can see from this figure that using a non-constant conditional error function requires much smaller second (and thereby overall) informations than the constant conditional error function $A\equiv\alpha$. With a sufficiently large pilot information, the region for the second-stage information can even be uniformly smaller for the adaptive designs than for the non-adaptive design. Moreover, also the constant information $\itwoconst$ for the case $Z_1<z_f$
becomes smaller with the non-constant conditional error functions. This clearly indicates the efficiency gain by using such adaptive designs.

Figure~\ref{fig:maxAndMeanInformNoBindConDReg} (left) allows for a comparison of the maximum information of the second stage between the different designs over both branches. It can be seen that the non-adaptive design leads to a larger maximum information than the adaptive design with $A_Z$ if $t_{\xi}(I_1)\geq 0.14$, with Fisher’s product test if $t_{\xi}(I_1)\geq 0.17$ and with the inverse normal method if $t_{\xi}(I_1)\geq 0.2$.

Figure~\ref{fig:maxAndMeanInformNoBindConDReg} (right) shows for $\xi=1.25$ and $\alpha_c=0.15$ the average second stage information (relative to $I_\delta$) for the different designs when $\delta=\xi \delrel$.  The non-adaptive design has a larger second-stage information than the inverse normal method for $t_{\xi}(I_1)\geq 0.1398$ and Fisher’s product test for $t_{\xi}(I_1)\geq 0.0742$. In particular, there is only a small non-gambling area for $t_\xi(I_1)$ where the non-adaptive design has a smaller mean information than the inverse normal method. For Fisher's product test and $A_Z$, there is no such area. The adaptive design with $A_Z$ even has a lower mean information than the non-adaptive design for all considered $t_{\xi}(I_1)$. A comparison of the maximum and mean information over both stages can be found in the Supplement.
\begin{example}[continued]
For our example fast-track procedure with assumption $\delrel=1.4$ and the considered relative first-stage information $t_{\xi}(I_1)=0.5$, i.e. $n_1=0.5\cdot n_{\delta}\approx 69$, Table \ref{tableResultsExample} shows the constant second-stage sample size per group $n_{2,\mathrm{const}}$ in case of no successful conditional registration and the minimum and maximum sample size per group in case of a successful conditional registration. The mean information over all $z_1\in\mathbb{R}$ and thus over both branches is also shown by $E(n_2)$. The values in brackets show the quantities added by $n_1$, i.e. over both stages. Depending on the confidence of the applicant in the assumed a priori effect, a comparison of the per group sample sizes of Table \ref{tableResultsExample} with $n_{\delta}=137$ or $n_{\mathrm{rel}}:=2\hat{\sigma}^2\cdot\eta_f^2/\delrel^2=2\cdot 5.17^2\cdot2.8^2/1.4^2\approx 214$ may be of interest.
\begin{table}[ht]
\centering
\begin{tabular}{ |c|c|c|c|c| } 
\hline
$A$ & $n_{2,\mathrm{const}}$ & $n_{2,\min}$ & $n_{2,\max}$ & $E(n_2)$\\
\hline
$A\equiv\alpha$ & $137 (206)$ & $130 (199)$ & $215 (284)$ & $139 (208)$\\ 
$A_{I,-\infty}$ & $137 (206)$ & $14 (83)$ & $137 (206)$ & $68 (137)$\\ 
$A_{F,-\infty}$ & $137 (206)$ & $12 (81)$ & $141 (210)$ & $70 (139)$\\ 
$A_{Z}$ & $124 (193)$ & $22 (91)$ & $124 (193)$ & $70 (139)$\\ 
\hline 
\end{tabular}
\vspace{0.5cm}
\caption{Sample size results for hypothetical fast-track procedure.}\label{tableResultsExample}
\end{table} 
\end{example}

\begin{figure}[h]
\includegraphics[width=1\textwidth]{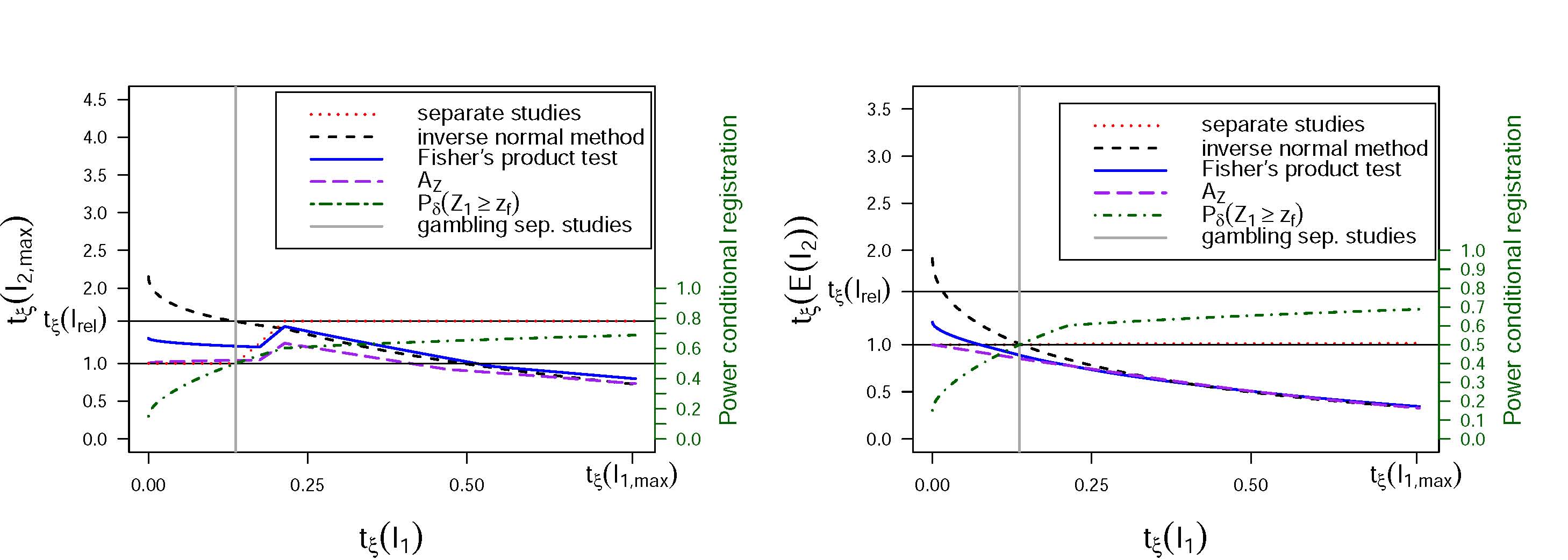}
\caption{Maximum information (left) and mean information (right) for $\xi=1.25$ over both branches.}\label{fig:maxAndMeanInformNoBindConDReg}
\end{figure}

\subsection{Probability for a conditional registration}

The probability for achieving conditional registration (i.e. $Z_1\ge z_f$) under the a priori assumed effect $\delta$ is shown by the dot-dashed lines with values given by the right vertical axis. This curve is the same for all designs and plots of Figures~\ref{fig:BS} and \ref{fig:maxAndMeanInformNoBindConDReg}. It illustrates the possibilities and limitations in reaching conditional registration for given first-stage information $t_\xi(I_1)$ when $\xi=\delta/\delrel$ is a small as $1.25$. For our example fast-track procedure, the power for a conditional registration at $\delta=1.75$ is $\mathbb{P}_{\delta}(Z_1\geq\zf)=0.65$.

\subsection{Additional results and final remarks}
In the Supplement, one can find similar results of  analogous investigations for the alternative value of $\xi=\xi_{\min}$.

We finally mention that there are several alternative strategies for designs with a data-driven decision to either apply for or abstain from the application for a conditional registration. 
For example, we could use an adaptive design with a benchmark boundary $z_{\text{s}}<\zf$, where for $z_{\text{s}}<Z_1<\zf$ one calculates the second stage information such that a conditional power of $1-\beta$ is achieved at the interim estimate $\max\{\delrel,\estone\}$, and to avoid a too large maximum second-stage information, a constant information is used for the second stage if $Z_1\leq z_{\text{s}}$,
like in the lower branch of Figure~\ref{twoArms}. This information could then be calculated to meet the power $1-\beta$ conditional on $Z_1\leq z_{\text{s}}$.
Furthermore, a binding futility bound less than $\zf$ could also be introduced, like a stopping if $\estone\leq 0$, i.e. $Z_1\leq 0$, or stopping for futility, if a significant negative effect is observed in the pilot study. With such futility bounds, we expect only small quantitative (and no qualitative) changes in the results and findings presented here, in particular, when the futility bound is non-binding.


Moreover, we would like to emphasize the importance of an overall type I error rate control when deciding in a data-driven way whether to apply for conditional registration or abstain from this application and continue with a study for permanent registration. Without such overall type I error rate control, we could end up with a substantially increased probability for a permanent registration of an inefficient device that is much larger than $\alpha$. Assume that we start, like in Section~\ref{NewChapterStudyDesignsFastTrack}, with a two-stage fast-track process using the binding futility rule $Z_1<z_f$ and a (non-constant) conditional error function $A$ at level $\alpha$. If, in the case of an unsuccessful attempt to register conditionally, we (always) start a new study for a direct, permanent registration at significance level $\alpha$, the overall type I error rate could become $(1+\Phi(z_f))\cdot\alpha$ which equals $0.04625$ for $\alpha=0.025$ and $z_f=\Phi^{-1}(1-\alpha_c)$ with $\alpha_c=0.15$. We would conclude that this implies the requirement to disclose in advance the attempt to continue after an unsuccessful pilot study with a study for direct permanent registration. Moreover, if this is intended, one would need to show that the anticipated strategy controls the overall type I error rate at the required level. An example of such a strategy was indicated in Figure~\ref{twoArms} and discussed in this section.   

\section{Summary and Discussion}\label{sec:Summary}
The main focus of this paper is on the utility of adaptive designs for fast-track registration processes (for digital health applications) like the DiGA fast-track in Germany. This refers to the investigation of required minimum, maximum, and average sample sizes under the aim of a certain overall success probability at an assumed a priori healthcare effect. Such an a priori assumption is not uncommon for adaptive designs, where we often speak of a \emph{sample size reassessment}, indicating that some a priori sample size calculation has been performed (see, for instance, Chapter~7 in the work of Wassmer and Brannath \cite{BuchWassmerBrannath2016}). We have clearly seen that the implementation of the fast-track procedure with an adaptive design is more efficient than the adopted approach with separate studies, whereby we worked with the assumption of the same design, same primary endpoints, and same study population for conditional and final registration. The use of adaptive designs often leads to smaller minimum, maximum, and average sample sizes, because they permit the use of the pilot data for the permanent registration study. 

Due to Section~\ref{NewChapterStudyDesignsFastTrack}, adaptive designs are particularly more efficient than the independent study approach if an unsuccessful study for conditional registration precludes any further attempt to receive a permanent registration. In this case, they are also less conservative than the separate study approach that has an effective overall type I error rate of $\alpha_f\cdot \alpha$ while the adaptive designs can achieve permanent registration with an overall type I error rate $\alpha$.

We have seen in Section~\ref{ConsequencesRequirementsBindCondRegits} that if an unsuccessful study for conditional registration precludes any further attempt to receive a permanent registration, a fast-track procedure is questionable when the a priori expectation on the health care effect does not go much beyond the minimal clinically relevant effect. This is because in these cases, the requirements for conditional registration are becoming more stringent than those for permanent registration. 

This observation has motivated us to consider in Section~\ref{ChapterNonBindingFutility} the second type of designs without a binding futility stop, which allow to decide based on the pilot data whether to apply for conditional registration or abstain from this application, and continue with a study for direct permanent registration. The presented designs can be considered for all a priori effect assumptions and are valid due to a strict control of the type I error rate. A prior announcement to the authorities of such a data-dependent mid-trial decision in advance ensures transparency. This appears to us to be an approach that is acceptable in a regulatory environment. Two different types of adaptive designs are presented in Section~\ref{ChapterNonBindingFutility}. The first type is based on a specific conditional error function selected for the fast-track route, which is then extended by the possibility of skipping the conditional registration. The inverse normal method and Fisher's product test are considered. The second type of design is based on the approach of Müller and Schäfer \cite{muller2001adaptive}. Starting with a test for direct permanent registration, the possibility of deciding data-dependent whether to apply for conditional registration is integrated. This is formalized by the conditional error function $A_Z$. Mathematically, the considered types of designs only differ in the way the conditional error function is selected. For the first type of design, the choice is ad hoc; for the other, the conditional error function of the permanent registration is used, which is in line with the Müller and Schäfer method. As we have seen, using the conditional error function $A_Z$ is convincing and outperforms the non-adaptive design across all considered efficiency parameters for at least sufficiently large pilot studies.

As already mentioned, this paper aims to illustrate the considerable potential of using adaptive designs in fast-track procedures such as those for digital health applications by presenting selected designs. However, this means that in certain situations, alternative designs may be more appropriate. For example, under the assumptions of Section~\ref{ChapterNonBindingFutility}, it is not necessary that both arms in Figure~\ref{twoArms} achieve power $1-\beta$ for rejecting the null hypothesis conditional on the respective arm being realized. For instance, a conditional success probability $\pi>1-\beta$ could be assigned to the arm that covers successful conditional registration and a lower one to the other arm, such that an overall success probability of $1-\beta$ is still given. That could increase the success probability for conditional registration.

Note that in Section~\ref{NewChapterStudyDesignsFastTrack}, where conditional registration was a requirement for permanent registration, we can understand the case of an unsuccessful registration as a binding futility stop for the two-stage designs. Instead, we could have dealt with this case like an unbinding futility stop, which would have led to a somewhat conservative design with somewhat larger sample sizes.  However, such designs provide the opportunity to continue with a second stage for permanent registration whenever there are good reasons for and sufficient resources have been achieved. This opportunity may outweigh the rather limited increase in the sample sizes. Corresponding graphs showing the efficiency of such designs can be found in the Supplement.

One should note that in most parts of the paper, the choice of the considered conditional error functions is somewhat arbitrary. Thus, there could be functions that lead to greater efficiency. For instance, one could consider the \emph{optimal conditional error functions} suggested by Brannath and Bauer \cite{brannath2004optimal}, which minimize the expected sample size under the sample-size recalculation approach to achieve a certain conditional power used in this paper.

An additional potential advantage of the adaptive designs mentioned in this paper is the robustness of the operation characteristics with regard to the a priori assumptions, the main reason to do a sample size recalculation at interim analyses or a pilot study. An investigation of the robustness of the adaptive designs in comparison to separate stage approaches would be an interesting topic for future investigations.

Other types of designs that also allow an early attempt to achieve permanent registration at the interim analysis are not considered here. This includes classical group sequential designs and adaptive designs with the possibility of early rejection of the null hypothesis while controlling the overall type I error rate (see, for instance, the work of Wassmer and Brannath \cite{BuchWassmerBrannath2016}). Such designs could be of particular interest when the sample sizes of the first stage are large, and so there might be sufficient evidence for a permanent registration; otherwise, the data could be used to apply for conditional registration. However, this requires the use of group sequential boundaries in order to account for the repeated attempts to reach permanent registration. Note that an adaptive design has the advantage over a classical group sequential design to permit a data-driven sample size calculation and the control of conditional power, which can be understood as a requirement for conditional registration.

At this point, it is important to emphasize that we do not consider a multiplicity adjustment with regard to conditional and permanent registration, i.e., we do not aim to control the probability of erroneously receiving either conditional or permanent registration (or both). We consider such an adjustment as unnecessary because of the temporary nature of the conditional registration that applies only to a rather short time interval (e.g. in Germany, one year, and in other countries, even less). Given this, it's very likely that a patient will be affected by only one registration process, i.e. receives an inefficient health application due to either a temporary or permanent registration. This is in line with the idea of controlling only relevant type I error rates; see, for instance the works of Brannath et al \cite{brannath2023population} and Brannath \cite{brannath2023discussion}.

Finally, we would like to mention some important aspects of a fast-track registration process that remained unaddressed in this paper because it's either out of the scope of our consideration or would have overloaded our investigation. Most of our discussion relies on the existence of a well-accepted minimally clinically relevant healthcare effect. As indicated in our running example, this may not be the case in practice for all endpoints. Moreover, there are quite a few different approaches to determine a minimally clinically relevant effect, and there seems to be no clear guidance on which approach to use; see, for instance, the works of Devji et al \cite{MCIDevji2020evaluating}, Franceschini et al \cite{MCIDfranceschini2023minimal} and Vach and Saxer \cite{MCIDvach2024anchor}. Moreover, as we have seen, the requirement that the point estimator of the healthcare effect exceeds the minimally clinically relevant effect can be stronger than the standard requirement of a p-value below $\alpha=0.025$ for sample sizes that are just as large as required for a power of $80\%$. This questions the use of a minimally clinically relevant effect as a requirement for permanent registration. However, this does not apply to pilot or interim data as long as their sample sizes are moderate, and hence a minimally clinically relevant effect may well be used as a requirement for conditional registration. In any case, for our exposition, the assumption of a minimal clinically relevant effect provided an insightful context for a comparison of different designs for a fast-track procedure.

Another unaddressed possibility is the use of surrogate endpoints for conditional registration, as it has been suggested in the pharmaceutical context \cite{fleming2012biomarkers, elliott2023surrogate}. This requires more discussions and investigations (see, for example, the works of Shih et al \cite{shih2003controlling} and Shih \cite{shih2006plan}), in particular when applied to digital health applications. Moreover, there is the possibility of selecting an endpoint at interim for conditional and permanent registration. In particular, the endpoints as well as the study populations may differ between conditional and permanent registration. Further investigations could also be conducted in this context.

A further unaddressed important issue that goes completely beyond our discussion of our paper are health economic aspects that consider the balance between gains and costs of a fast-track procedure.

\section*{Acknowledgments}
The authors kindly thank PD Dr. Norbert Benda for helpful discussions and comments. Further, the authors gratefully acknowledge the support of the Leibniz ScienceCampus Bremen Digital Public Health (www.digital-public-health.de), which is jointly funded by the Leibniz Association (W72/2022), the Federal State of Bremen, and the Leibniz Institute for Prevention Research and Epidemiology – BIPS.
\bibliographystyle{unsrt}  
\bibliography{references}  


\newpage
\appendix
\section{Proofs}\label{AppendixProofs}
\begin{lemma}\label{LemmaUniqueMinInformSecondStage}
Let $A:\mathbb{R}\rightarrow[0,0.5]$ be a conditional error function which is positive on $[\zf,\infty)$, $\zf>0$, and $\delta=\xi\delrel$ for $\xi\geq 0$ be the a priori assumed or true parameter and $I_1>0$, $\alpha,\beta\in(0,0.5)$. By setting $A$ constant to $\alpha$, the non-adaptive design can be considered. We define the success probability at $\delta$ by
\begin{align*}
\Succ_A(\itwomin,I_1)&=\mathbb{P}_{\delta}\left( Z_1\geq\zf,\text{reject~}H_0\right)\\
&=\int\limits_{z_f}^{\infty}\left\{1-\Phi\left(\Phi^{-1}(1-A(z_1))-\eta_f{\sqrt{I_2^{A}(z_1)}}/{\sqrt{I_{\delta}}}\right)\right\}\phi(z_1-\sqrt{I_1}\delta)\mathrm{d}z_1
\end{align*}
where 
\begin{align*}I_2^{A}(z_1):=\max\left\{\itwomin,I_1{\left\{\Phi^{-1}(1-\beta)+\Phi^{-1}(1-A(z_1))\right\}^2}/{\zf^2}\right\},~z_1\geq\zf~.
\end{align*}
Let $p>0$ such that $p<\mathbb{P}_{\delta}(Z_1\geq z_f)$, $Z_1=\sqrt{I_1}\estone\sim\mathcal{N}(\delta\sqrt{I_1},1)$. If $\Succ_A(0,I_1)\leq p$, we find a unique $\itwomin\geq 0$ such that $\Succ_A(\itwomin,I_1)=p$.
\end{lemma} 
\begin{proof}
It holds for all $z_1\geq\zf$ that $\sqrt{I_2^{A}(z_1)}\nearrow\infty$ if $\itwomin\nearrow\infty$. Thus, for $\itwomin\nearrow\infty$,
\begin{align*}
    z_1\mapsto 1-\Phi\left(\Phi^{-1}(1-A(z_1))-\eta_f{\sqrt{I_2^{A}(z_1)}}/{\sqrt{I_{\delta}}}\right)
\end{align*}
gives a pointwise non-decreasing sequence of functions which converge pointwise to the constant function $z_1\mapsto 1$. Then, due to the monotone convergence theorem, $\Succ_A(\itwomin,I_1)\nearrow 1\cdot\mathbb{P}_{\delta}(Z_1\geq z_f)>p$ which implies for sufficiently large $\itwomin$ it holds $\Succ_A(\itwomin,I_1)> p$. It remains to show that there exists a unique minimal information such that equality holds. On the one hand, $\Succ_A(\itwomin,I_1)$ strictly increases in $\itwomin$. This follows from the fact that the used sample size formula for z-tests is monotonically decreasing in $z_1$ and converges to zero as $z_1$ approaches infinity. Indeed, this yields if $\itwomin$ increases, $I_2^{A}$ increases strictly on some non-empty interval $(z_i,\infty)$, $z_i\geq \zf$ and thus $\Succ_A(\itwomin,I_1)$ strictly increases. Furthermore, the dominated convergence theorem gives the continuity of $\Succ_A(\itwomin, I_1)$ in $\itwomin$. By putting all that together, we obtain: If $\Succ_A(0,I_1)\leq p$, we find a unique $\itwomin\geq 0$ such that $\Succ_A(\itwomin,I_1)=p$.
\end{proof}
In the main article, we use the previous lemma twice. In the first part of the paper, we consider all $\delta$ and $I_1$ which fulfill $\mathbb{P}_{\delta}(Z_1\geq\zf)>1-\beta$ and thus the lemma can be used for $p:=1-\beta$. In the second part of the paper we use that $\mathbb{P}_{\delta}(\text{reject~}H_0~|~Z_1\geq \zf)=1-\beta$ is equivalent to $\mathbb{P}_{\delta}(\text{reject~}H_0,Z_1\geq\zf)=(1-\beta)\cdot\mathbb{P}_{\delta}(Z_1\geq\zf)$ and set $p:=(1-\beta)\cdot\mathbb{P}_{\delta}(Z_1\geq\zf)$.
\begin{lemma}\label{LemmaMaximumInform}
Let $A:\mathbb{R}\rightarrow(0,0.5]$ be a conditional error function and $\delta=\xi\delrel$ for $\xi>\xi_{\min}=1.43$ be the a priori assumed effect and $I_1>\ionemin$, $\alpha,\beta\in(0,0.5)$. For the adaptive design with conditional error function $A$ and the non-adaptive design, the overall power at $\delta$ is achieved by choosing the required $\itwomin$. To distinguish between the two designs, we use the notations $\itwomin^{A}$ and $\itwomin^{\alpha}$. Furthermore, let $\zf>0$ and the fast-track procedure is terminated if $Z_1<\zf$ is observed after the first stage (binding futility stop). If $A(\zf)>\alpha$, the adaptive design has a lower maximum information of the second stage, i.e.
\begin{align}
\max\nolimits_{z_1\geq z_f}I_2^{A}(z_1) < \max\nolimits_{z_1\geq z_f}I_2^{\alpha}(z_1)  \label{ToShowMaximumInformComp}
\end{align}
and thus a lower maximum sum of information over both stages than the non-adaptive design. 
\end{lemma}
\begin{proof}
It holds 
\begin{align}
        \max_{z_1\geq z_f}I_2^{A}(z_1,\itwomin^{A},I_1)=\max\left\{\itwomin^{A},I_1{\left\{\Phi^{-1}(1-\beta)+\Phi^{-1}(1-A(\zf))\right\}^2}/{\zf^2}\right\}~.\label{TermMaxInformation}
\end{align}
The same holds for the non-adaptive design when replacing $A$ with $\alpha$. The proof is straightforward: If $A(z_f)>\alpha$, then
\begin{align}
    I_1{\left\{\Phi^{-1}(1-\beta)+\Phi^{-1}(1-A(z_f))\right\}^2}/{\zf^2}<I_1{\left\{\Phi^{-1}(1-\beta)+\Phi^{-1}(1-\alpha)\right\}^2}/{\zf^2}\label{FallzahlFormelAdaptKleiner}
\end{align}
holds true. We now must only exclude that
 \begin{align}
        \itwomin^{A}\geq\max\nolimits_{z_1\geq z_f}I_2^{\alpha}(z_1)=\max\left\{\itwomin^{\alpha},{I_1\left\{\Phi^{-1}(1-\beta)+\Phi^{-1}(1-\alpha)\right\}^2}/{\zf^2}\right\}~.\label{WidSpruchCond1} 
\end{align}
If $\itwomin^{A}=0$, (\ref{WidSpruchCond1}) cannot be true because the second part of the maximum in (\ref{WidSpruchCond1}) is greater than $0$ due to our assumptions. Then, (\ref{ToShowMaximumInformComp}) is proven. If $\itwomin^{A}>0$ the overall power at $\delta$ is exactly $1-\beta$, i.e. $\Succ_{A}(\itwomin^{A},I_1)=1-\beta$. If (\ref{WidSpruchCond1}) were true, it would imply because (\ref{FallzahlFormelAdaptKleiner}) that for all $z_1\geq\zf$ it holds $I_2^{A}(z_1)=\itwomin^{A}$. Due to $A(z_f)>\alpha$ and due to the monotonicity of $A$ on $[z_f,\infty)$ which implies $A(z_1)\geq A(\zf)>\alpha$ it holds $\Phi^{-1}(1-A(z_1))<\Phi^{-1}(1-\alpha)$. Together with the fact that we have assumed (\ref{WidSpruchCond1}) and $I_{2}^{\alpha}(z_1)$ is decreasing in $z_1$, it can be concluded that the integrated of $\Succ_A(\itwomin^{A},I_1)$ in (\ref{allgemeineConditionOverallPowerAdaptivD}) is strictly greater than the integrand of $\Succ(\itwomin^{\alpha},I_1)$ in (\ref{IntegralOverallPowerSharperBorder}). Therefore, $\Succ_{A}(\itwomin^{A},I_1)>\Succ(\itwomin^{\alpha},I_1)$. Because $\Succ_A(\itwomin^{A},I_1)=1-\beta$ and $\Succ(\itwomin^{\alpha},I_1)\geq 1-\beta$, we have a contradiction. Thus, (\ref{WidSpruchCond1}) does not hold true, and (\ref{ToShowMaximumInformComp}) is fulfilled. 
\end{proof}

\section{Numerical Implementation}
The images were created using the R-software. For the calculations, numerical methods like bisection search and the \texttt{integrate} function from the package \texttt{stats} were used. The R-code can be downloaded from the following link: \url{https://github.com/LianeKluge/Adaptive-Designs-in-Fast-Track-Registration-Processes}.

The following result is required for the numerical calculations in Section~\ref{subsectionInvestigationI_2const}.
\begin{lemma}\label{LemmaUesSlepiansLemma}
Let $Z(I_2):=Z(I_1,I_2):=\sqrt{{I_1}/{(I_1+I_2)}} Z_1 +\sqrt{{I_2}/{(I_1+I_2)}} Z_2$ where $Z_k=\hat{\vartheta}_k\sqrt{I_k}\sim\mathcal{N}(\delta\sqrt{I_k},1)$, $k=1,2$, and $\xi\geq 1$, i.e. $\delta\geq\delrel>0$. For $\zf\in\mathbb{R}$ the conditional success probability $\mathbb{P}_{\delta}(Z(I_2)\geq z_{\alpha}|Z_1<\zf)$ where $z_{\alpha}=\Phi^{-1}(1-\alpha)$ is increasing in $I_2$.
\end{lemma}
\begin{proof}
Let $I_{2,1}<I_{2,2}$. We prove $\mathbb{P}_{\delta}(Z_1<\zf,Z(I_{2,1})\geq z_{\alpha})<\mathbb{P}_{\delta}(Z_1<\zf,Z(I_{2,2})\geq z_{\alpha})$ by using Slepian's inequality (see, for instance, Theorem~5.1.7 in the work of Tong \cite{tong2012multivariate}). Slepian's inequality compares the strength of the dependence of two multivariate normally distributed random vectors with the same mean and the same variance and further assumptions on the covariances. We define $X_1:=Y_1:=Z_1-\delta\sqrt{I_1}$ and $X_2:=-Z(I_{2,1})+\delta\sqrt{I_1+I_{2,1}}$,  $Y_2:=-Z(I_{2,2})+\delta\sqrt{I_1+I_{2,2}}$. Then $X=(X_1,X_2)$ and $Y=(Y_1,Y_2)$ are centered random vectors and all variances equal $1$. Furthermore, $\mathrm{Cov}(X_1,X_2)=-\mathrm{Cov}(Z_1,Z(I_{2,1}))=-\sqrt{I_1/{(I_1+I_{2,1})}}\leq-\sqrt{I_1/{(I_1+I_{2,2})}}=-\mathrm{Cov}(Z_1,Z(I_{2,2}))=\mathrm{Cov}(Y_1,Y_2)$. Slepian's inequality now gives for all $a_1,a_2\in\mathbb{R}$ that $\mathbb{P}_{\delta}(X_1< a_1, X_2\leq a_2)\leq\mathbb{P}_{\delta}(Y_1< a_1, Y_2\leq a_2)$. By setting $a_1=\zf-\delta\sqrt{I_1}$ and $a_2=-z_{\alpha}+\delta\sqrt{I_1+I_{2,1}}$ and using $\delta> 0$ which implies $a_2<-z_{\alpha}+\delta\sqrt{I_1+I_{2,2}}$, we obtain
\begin{align*}
    \mathbb{P}_{\delta}(Z_1<\zf,Z(I_{2,1})\geq z_{\alpha})&=\mathbb{P}_{\delta}(X_1< a_1, X_2\leq a_2)\leq\mathbb{P}_{\delta}(Y_1< a_1, Y_2\leq a_2)\\
    &<\mathbb{P}_{\delta}(Y_1<a_1,Y_2\leq-z_{\alpha}+\delta\sqrt{I_1+I_{2,2}})=\mathbb{P}_{\delta}(Z_1<\zf,Z(I_{2,2})\geq z_{\alpha})~.
\end{align*}
This directly implies the monotonicity of the conditional success probability in $I_2$.
\end{proof}
\section{Additional Plots for Section~\ref{NewChapterStudyDesignsFastTrack}}
\subsection{With binding futility stop}\label{AppendixPlotsNoCombinationStrat}
The plots in this section complete the discussion in Section~\ref{NewChapterStudyDesignsFastTrack}. The adaptive designs with conditional error functions $A_{I,\zf}$, $A_{F,\zf}$ with binding futility stop after unsuccessful conditional registration and the design with two separate studies are considered.
\begin{figure}[H]
\begin{subfigure}[c]{0.56\textwidth}
\includegraphics[width=0.8\textwidth]{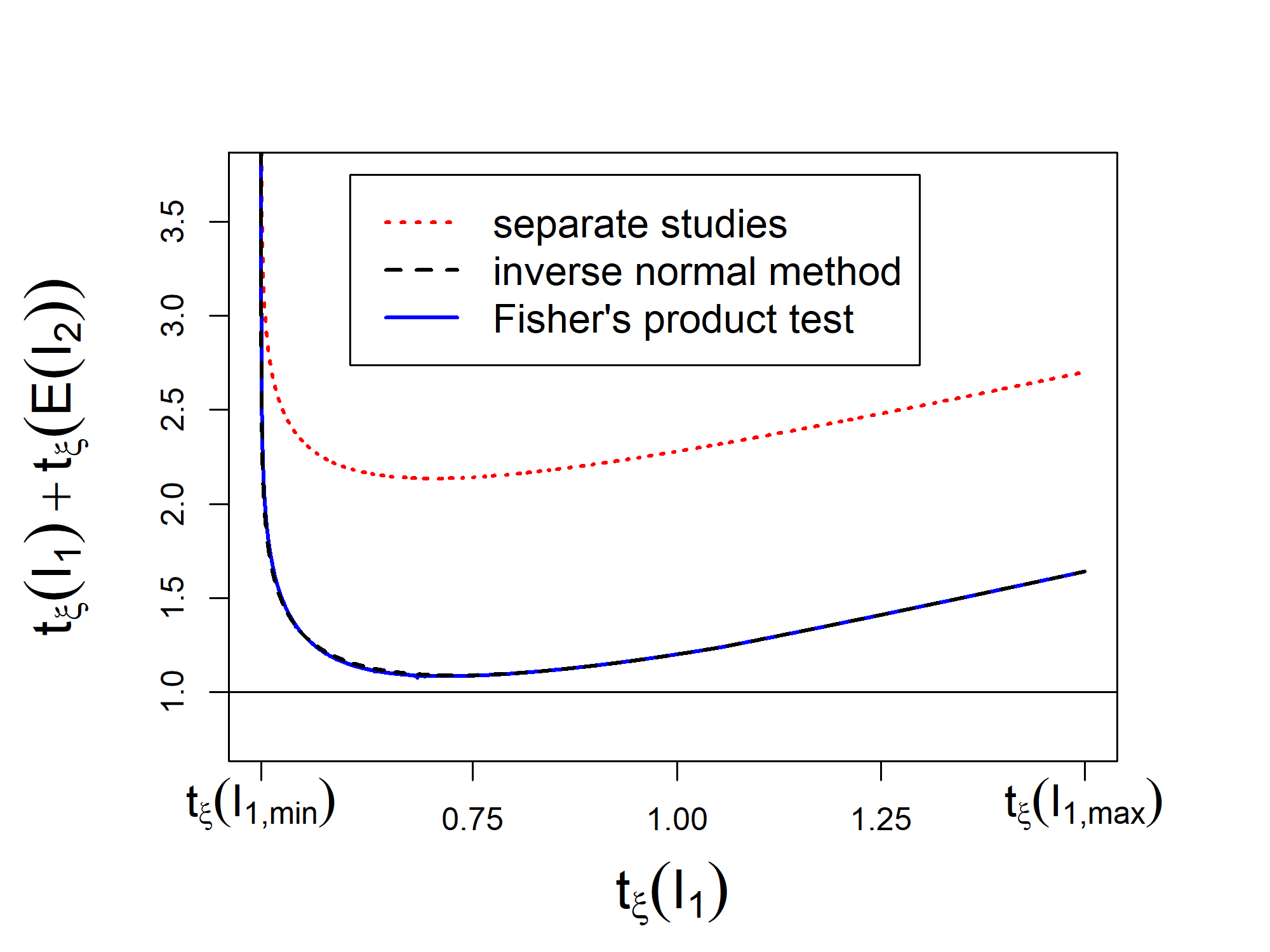}
\subcaption{$\xi=1.75$.}
\end{subfigure}
\begin{subfigure}[c]{0.56\textwidth}
\includegraphics[width=0.8\textwidth]{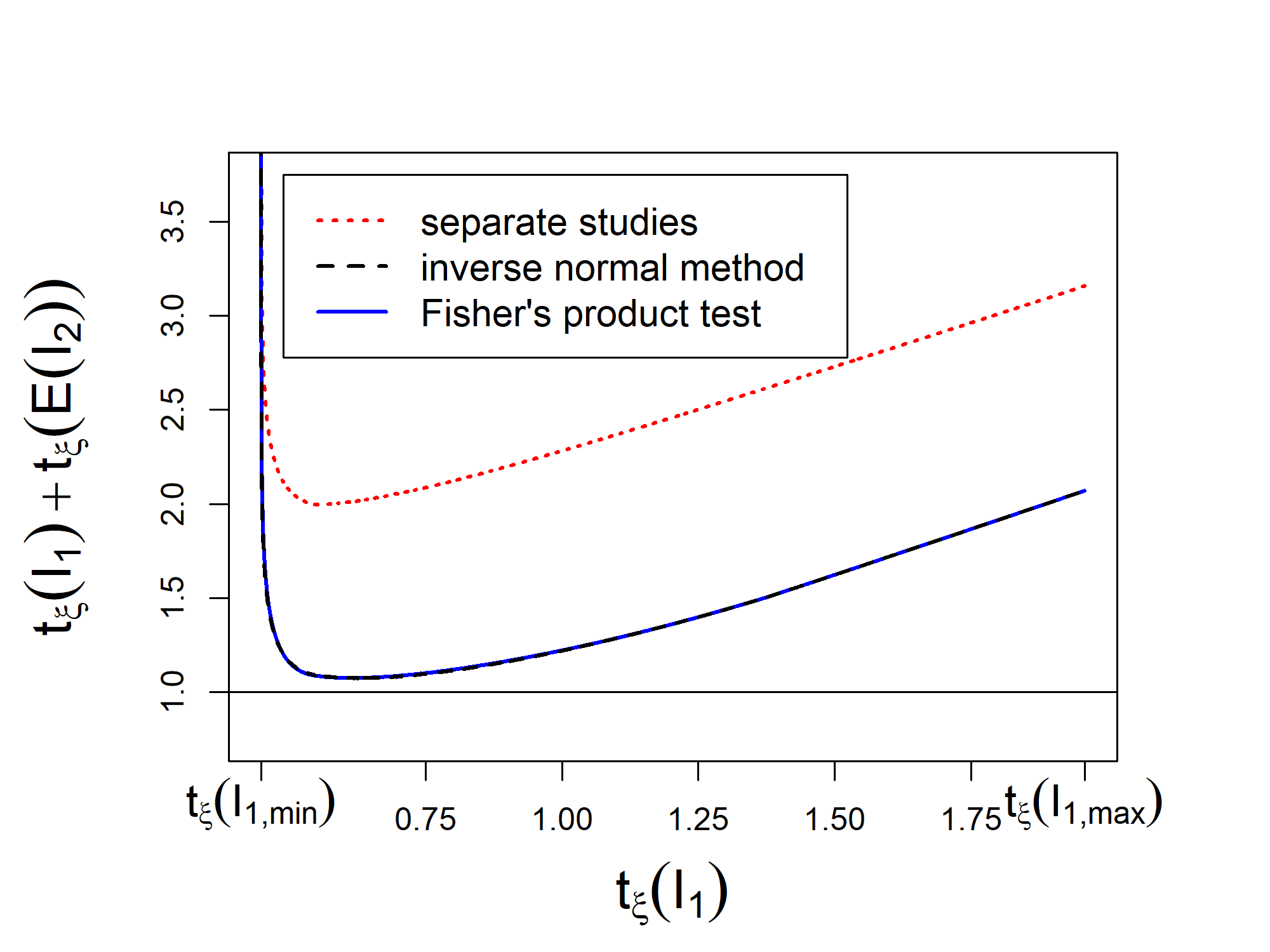}
\subcaption{$\xi=2$.}\label{fig:BF12}
\end{subfigure}
\caption{Mean information over both stages.}
\end{figure}
\begin{figure}[H]
\begin{subfigure}[c]{0.56\textwidth}
\includegraphics[width=0.8\textwidth]{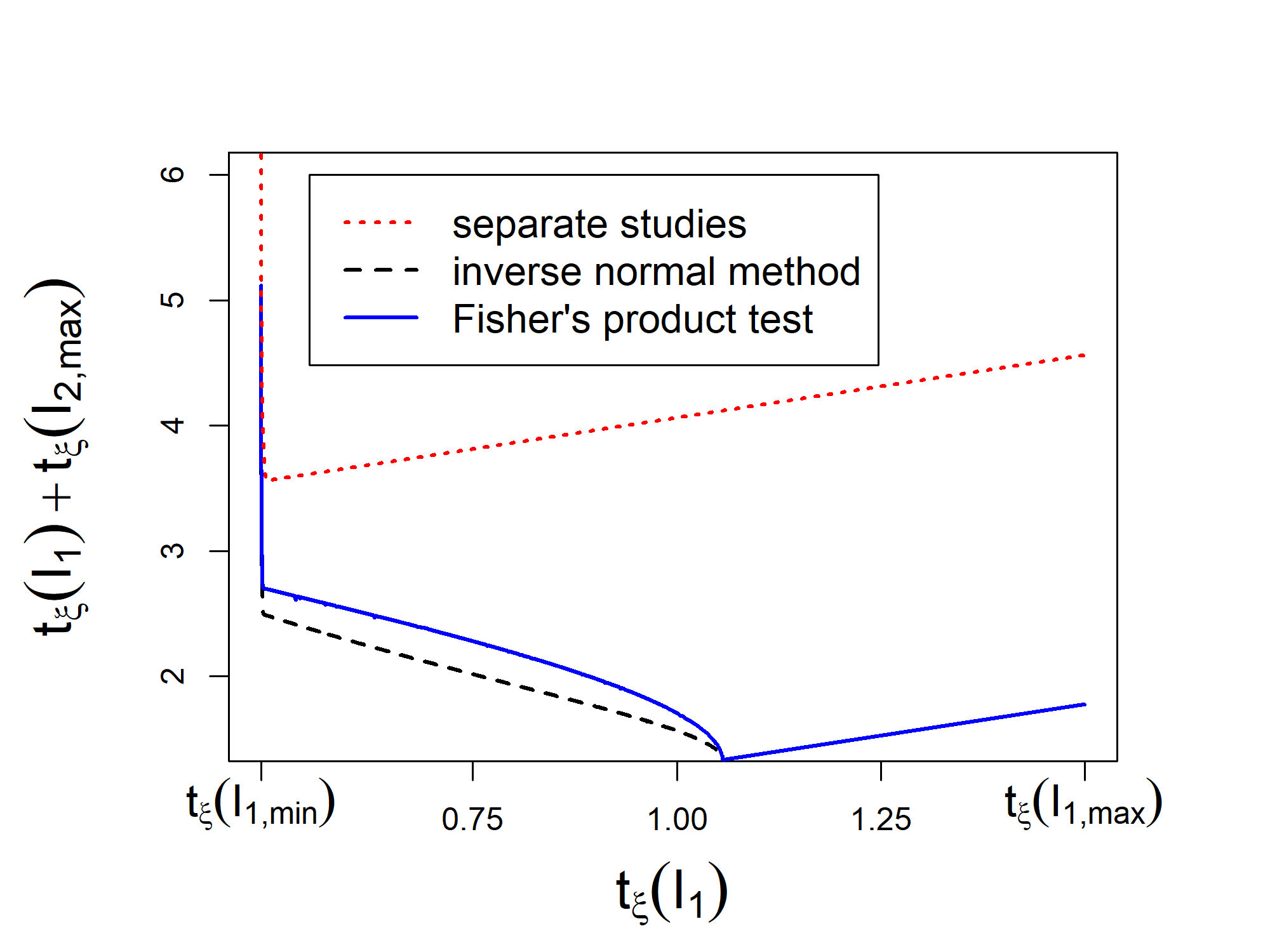}
\subcaption{$\xi=1.75$.}
\end{subfigure}
\begin{subfigure}[c]{0.56\textwidth}
\includegraphics[width=0.8\textwidth]{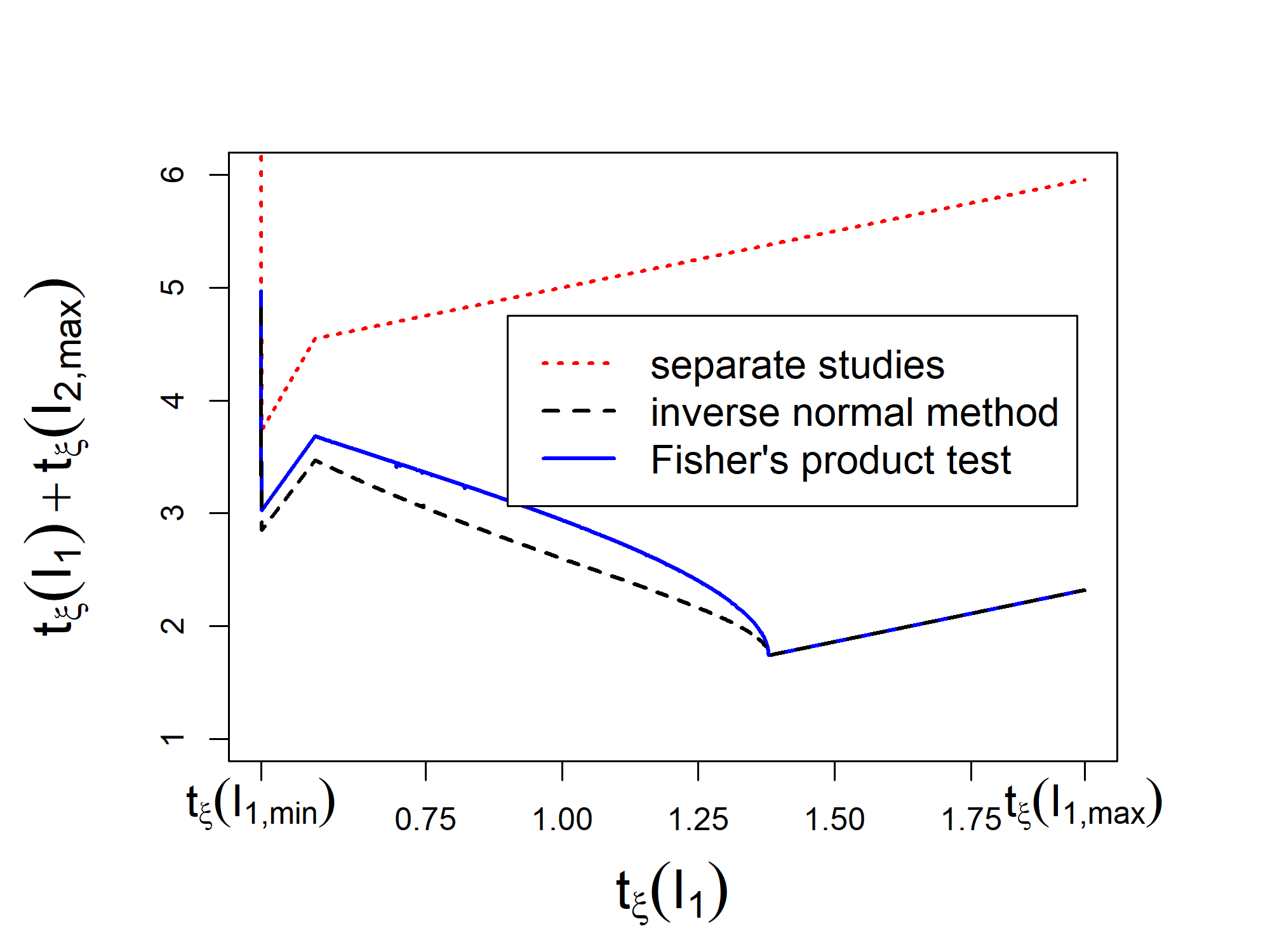}
\subcaption{$\xi=2$.}\label{fig:BF13}
\end{subfigure}
\caption{Maximum information over both stages.}
\end{figure}

\subsection{Without binding futility stop}\label{AppendixPlotsNoCombinationStratNonBinding}
In this Section, it is assumed as in Section~\ref{NewChapterStudyDesignsFastTrack} that at the time of the planning of a fast-track procedure by the applicants, a successful conditional registration is a requirement for permanent registration. Thus, we assume that $\mathbb{P}_{\delta}(Z_1\geq\zf, \text{~reject~}H_0)=1-\beta$ is striven at an a priori assumed or true effect $\delta=\xi\cdot\delrel$ where $\xi>\xi_{\min}$. We assume that the designs and informations are chosen such that this success probability is achieved, as in Section~\ref{NewChapterStudyDesignsFastTrack}. In contrast to Section~\ref{NewChapterStudyDesignsFastTrack}, the conditional error functions are not chosen in a way that there is a binding futility stop after unsuccessful conditional registration. Instead, in this section, we deal with this case like an unbinding futility stop and consider $A_{I,-\infty}$ and $A_{F,-\infty}$. These are somewhat conservative designs with somewhat larger sample sizes compared to $A_{I,\zf}$ and $A_{F,\zf}$. However, such designs provide
the opportunity to continue with a second stage for the permanent registration whenever there are good
reasons for and sufficient resources have been achieved. 
\begin{figure}[H]
\begin{subfigure}[c]{0.56\textwidth}
\includegraphics[width=.8\textwidth]{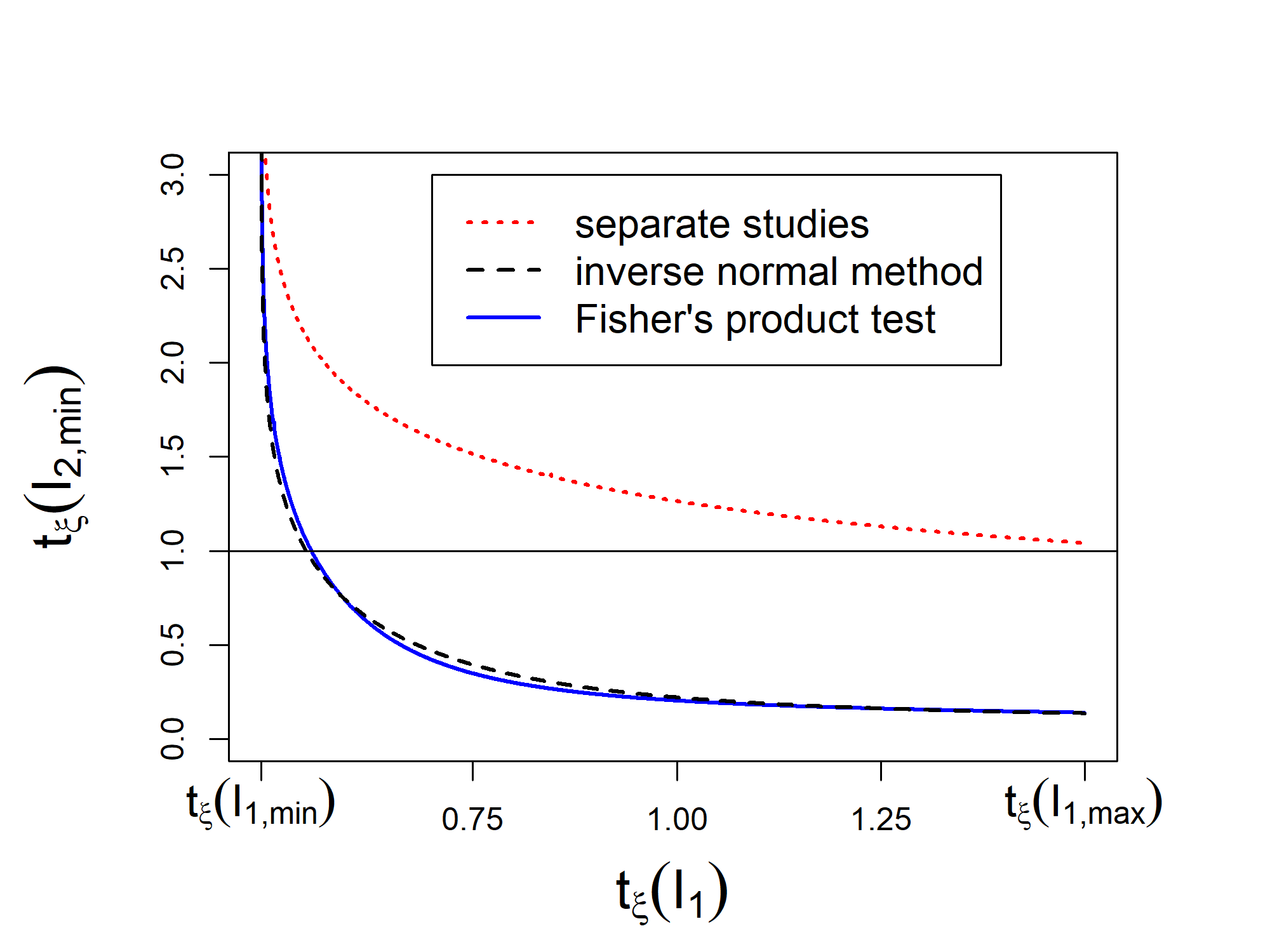}
\subcaption{$\xi=1.75$.}
\end{subfigure}
\begin{subfigure}[c]{0.56\textwidth}
\includegraphics[width=.8\textwidth]{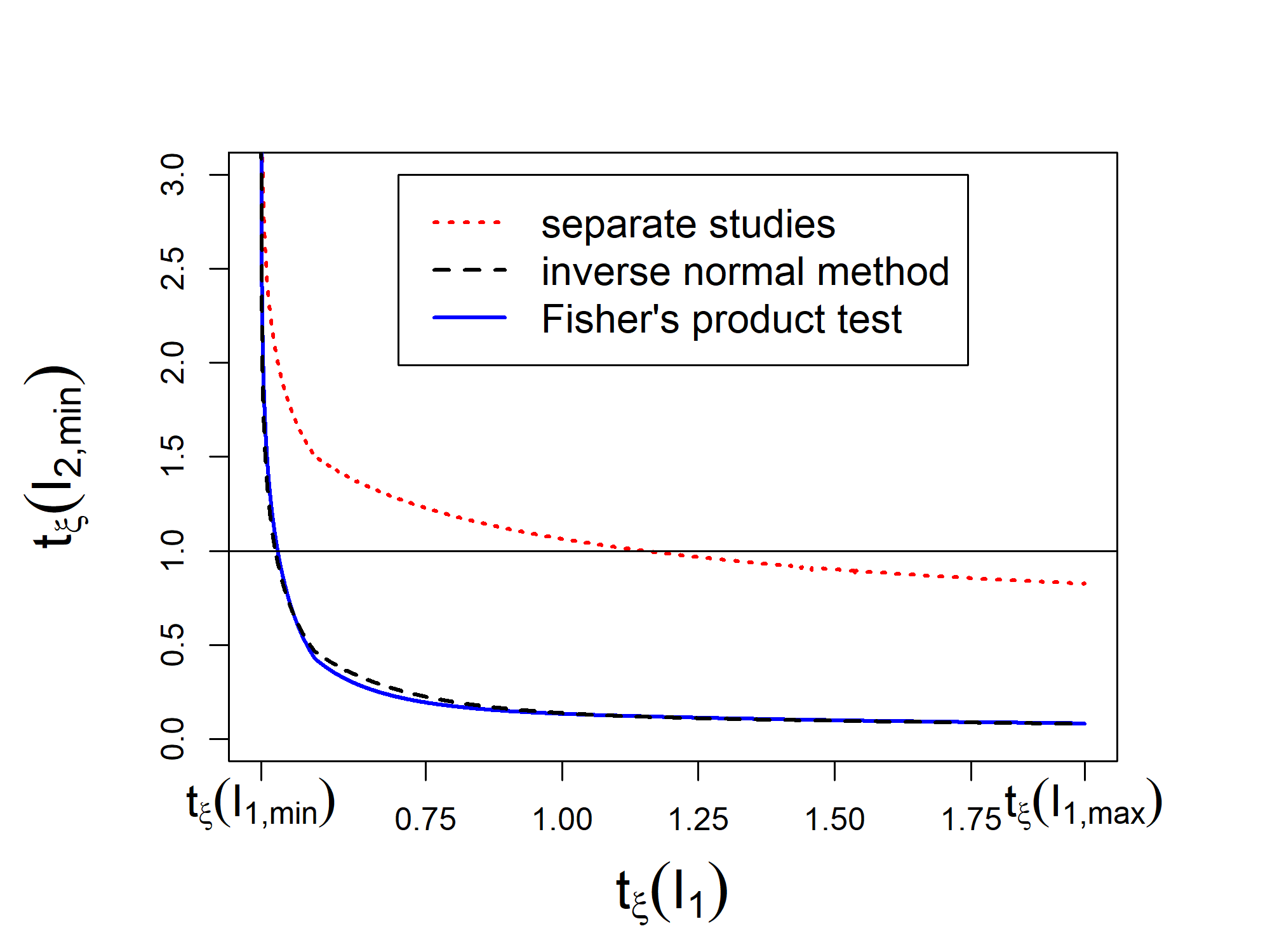}
\subcaption{$\xi=2$.}
\end{subfigure}
\caption{Minimum needed information for the second stage to achieve overall power of $1-\beta=0.8$ for the common non-adaptive design and the adaptive designs with conditional error functions $A_{I,-\infty}$ (inverse normal method) and $A_{F,-\infty}$ (Fisher's product test).}
\end{figure}
\begin{figure}[H]
\begin{subfigure}[c]{0.56\textwidth}
\includegraphics[width=.8\textwidth]{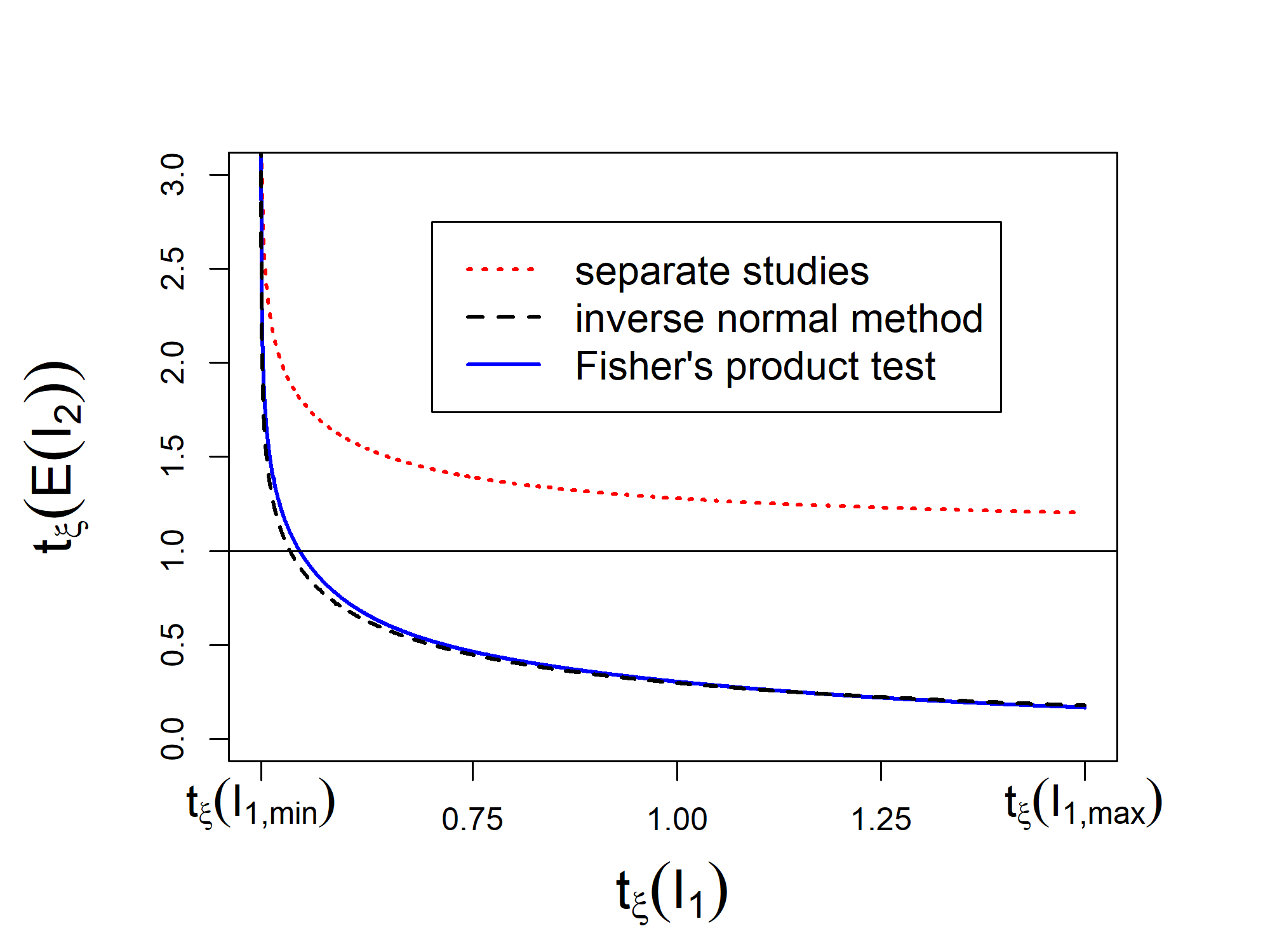}
\subcaption{$\xi=1.75$.}

\end{subfigure}
\begin{subfigure}[c]{0.56\textwidth}
\includegraphics[width=.8\textwidth]{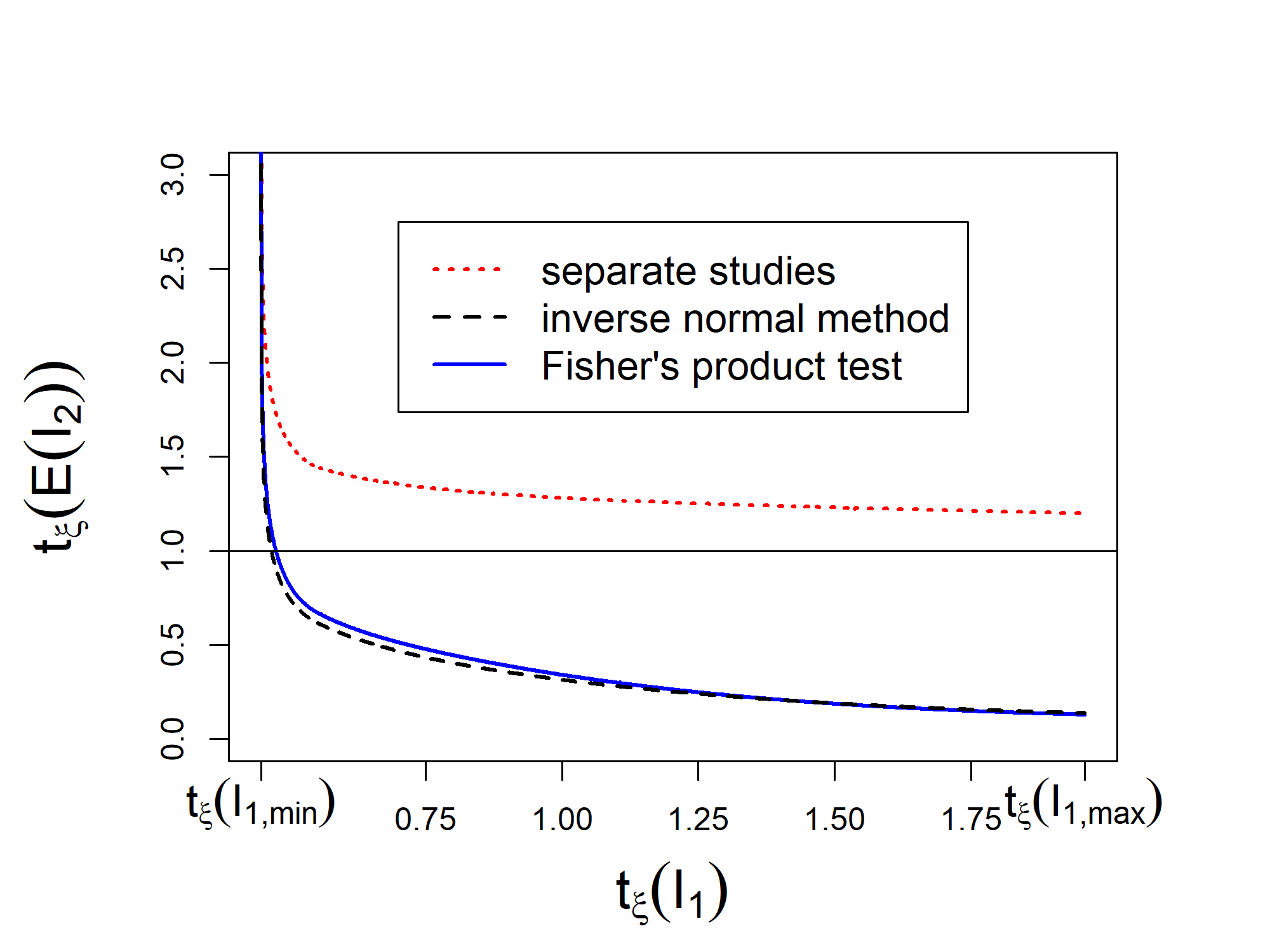}
\subcaption{$\xi=2$.}
\end{subfigure}
\caption{Visualisation of the mean information for the second stage for the non-adaptive design and the adaptive designs with conditional error functions $A_{I,-\infty}$ (inverse normal method) and $A_{F,-\infty}$ (Fisher's product test).}
\end{figure}
\begin{figure}[H]
\begin{subfigure}[c]{0.56\textwidth}
\includegraphics[width=.8\textwidth]{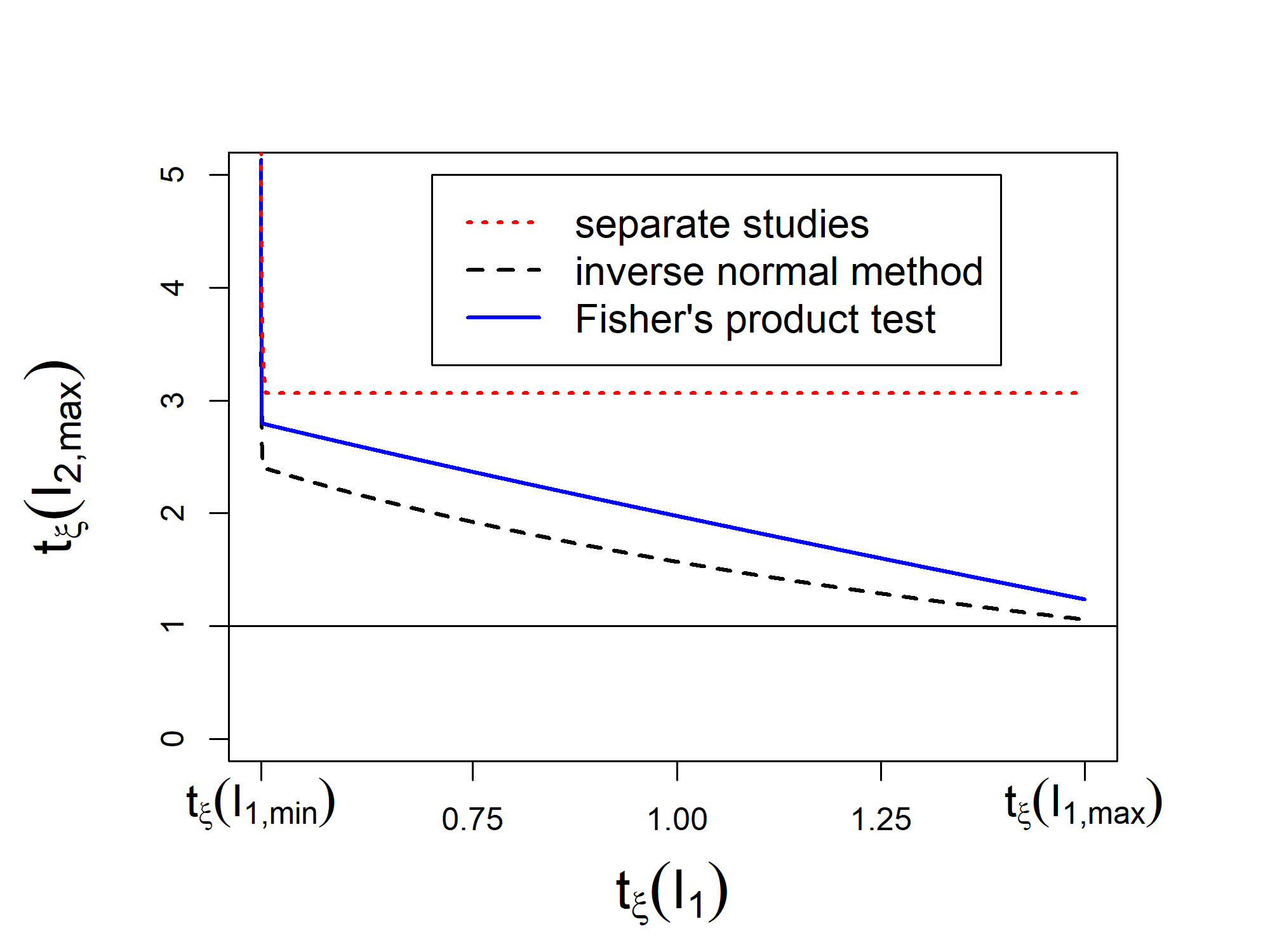}
\subcaption{$\xi=1.75$.}
\end{subfigure}
\begin{subfigure}[c]{0.56\textwidth}
\includegraphics[width=.8\textwidth]{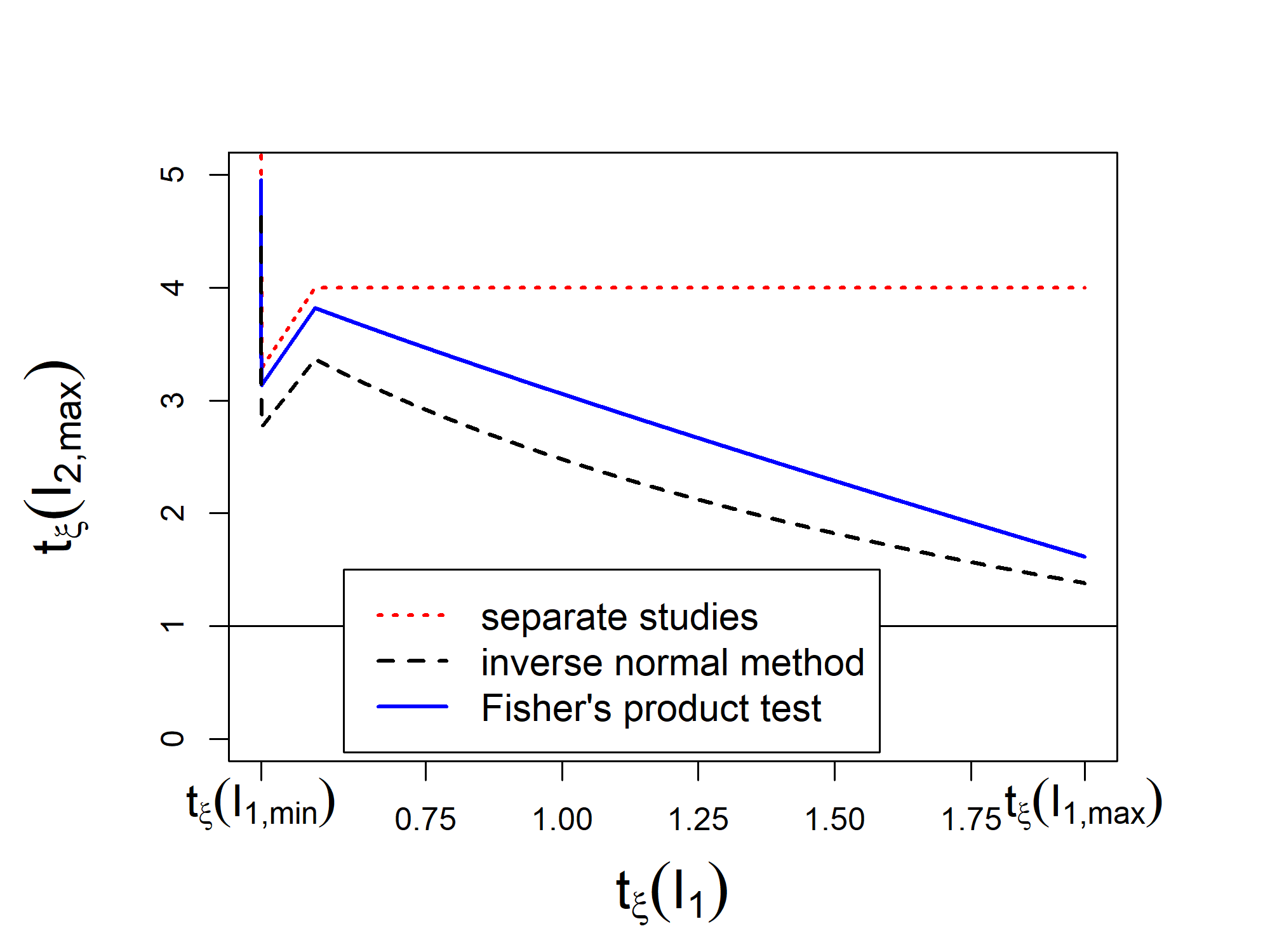}
\subcaption{$\xi=2$.}
\end{subfigure}
\caption{Visualisation of the maximum information for the second stage for the non-adaptive design and the adaptive designs with conditional error functions $A_{I,-\infty}$ (inverse normal method) and $A_{F,-\infty}$ (Fisher's product test).}
\end{figure}
\begin{figure}[H]
\begin{subfigure}[c]{0.56\textwidth}
\includegraphics[width=0.8\textwidth]{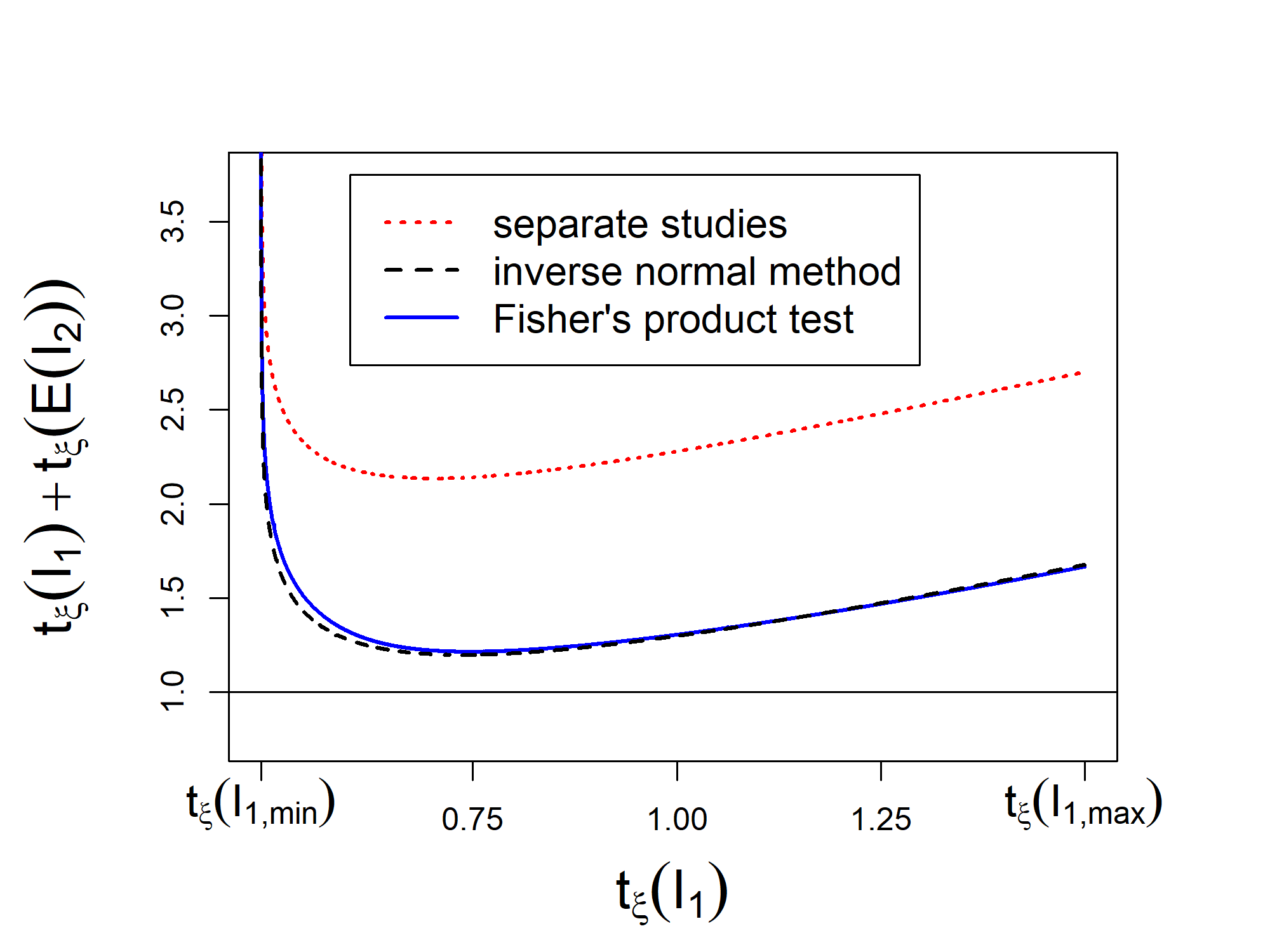}
\subcaption{$\xi=1.75$.}
\end{subfigure}
\begin{subfigure}[c]{0.56\textwidth}
\includegraphics[width=0.8\textwidth]{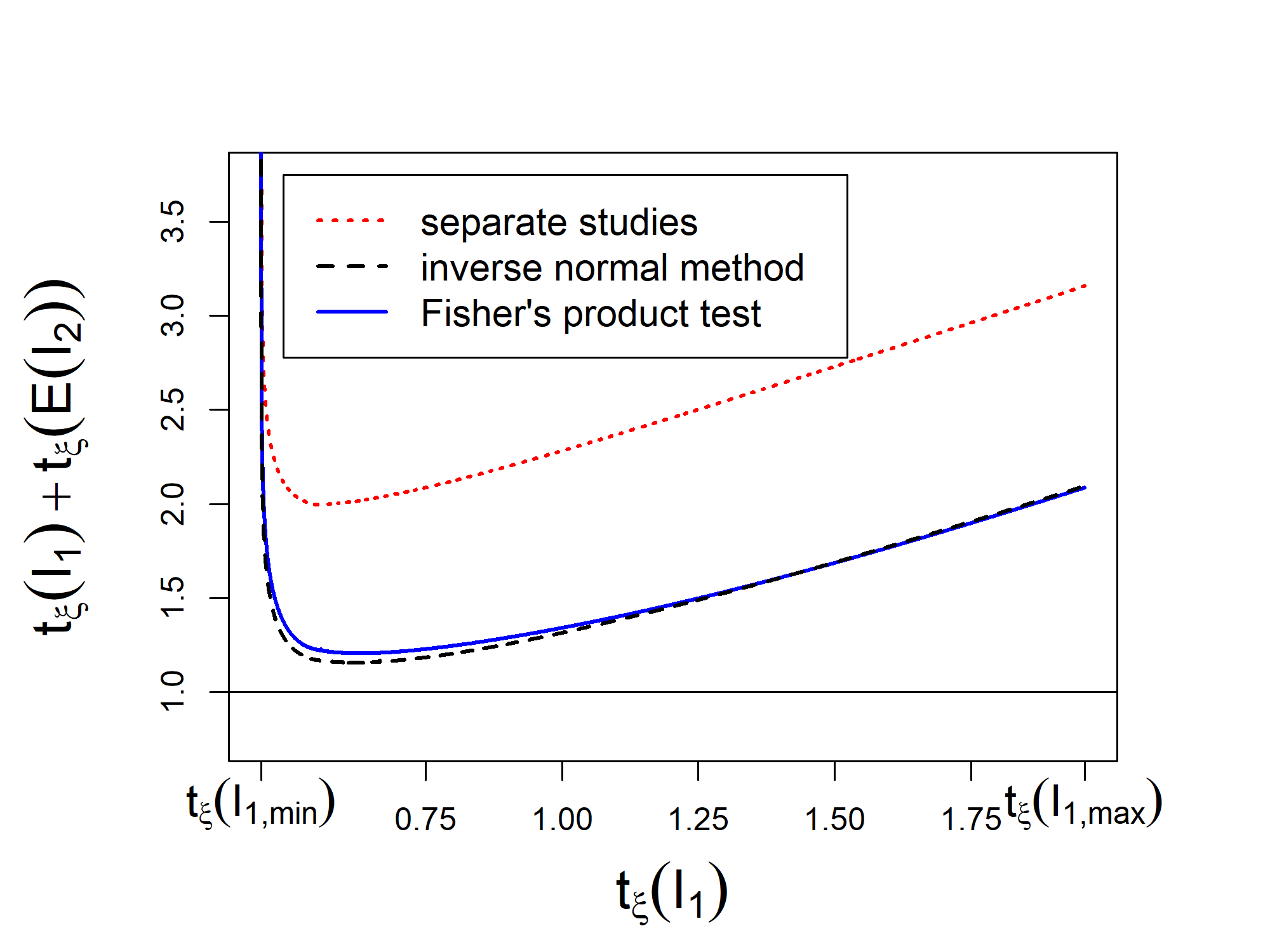}
\subcaption{$\xi=2$.}
\end{subfigure}
\caption{Mean information over both stages.}
\end{figure}

\begin{figure}[H]
\begin{subfigure}[c]{0.56\textwidth}
\includegraphics[width=0.8\textwidth]{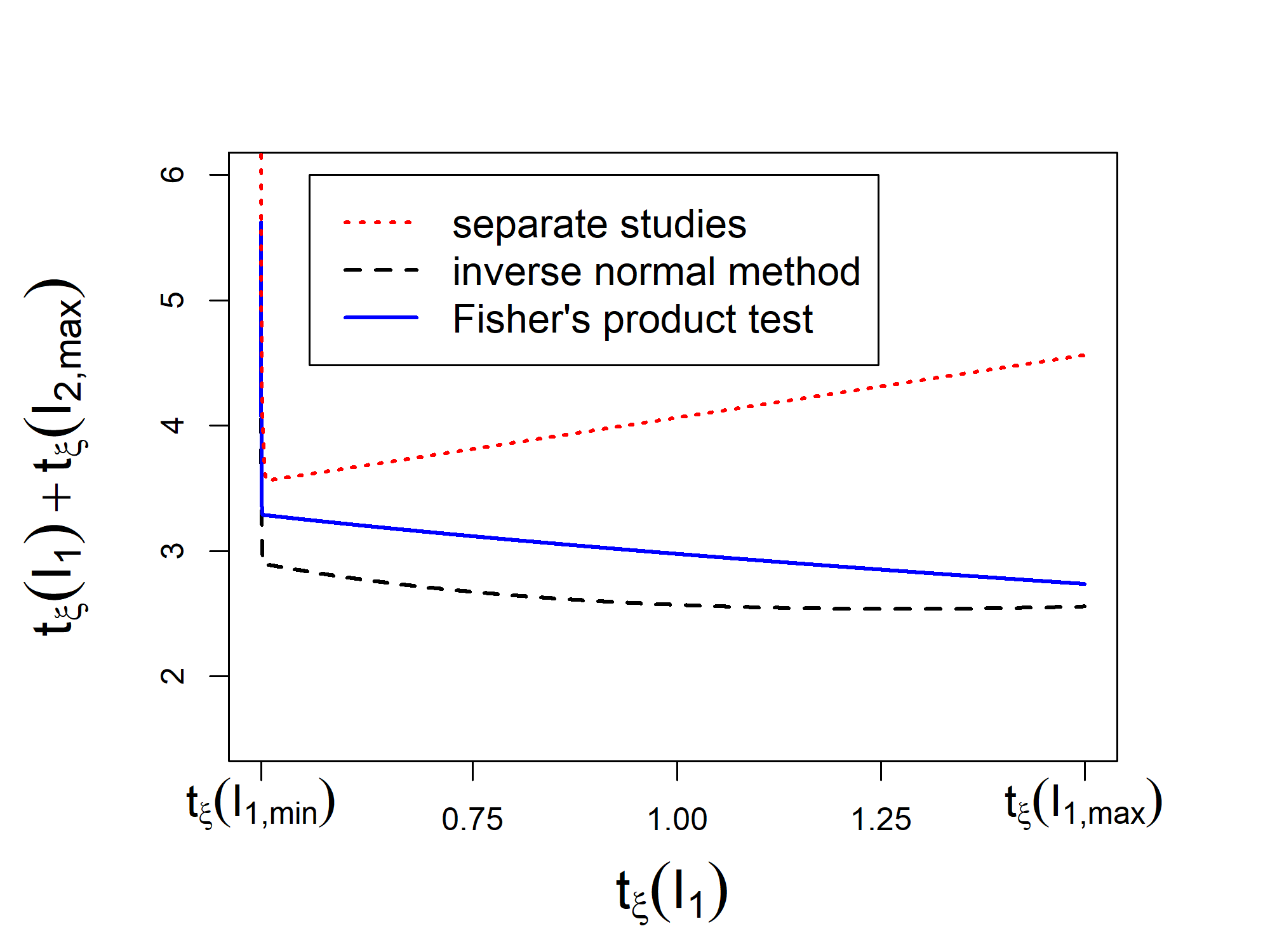}
\subcaption{$\xi=1.75$.}
\end{subfigure}
\begin{subfigure}[c]{0.56\textwidth}
\includegraphics[width=0.8\textwidth]{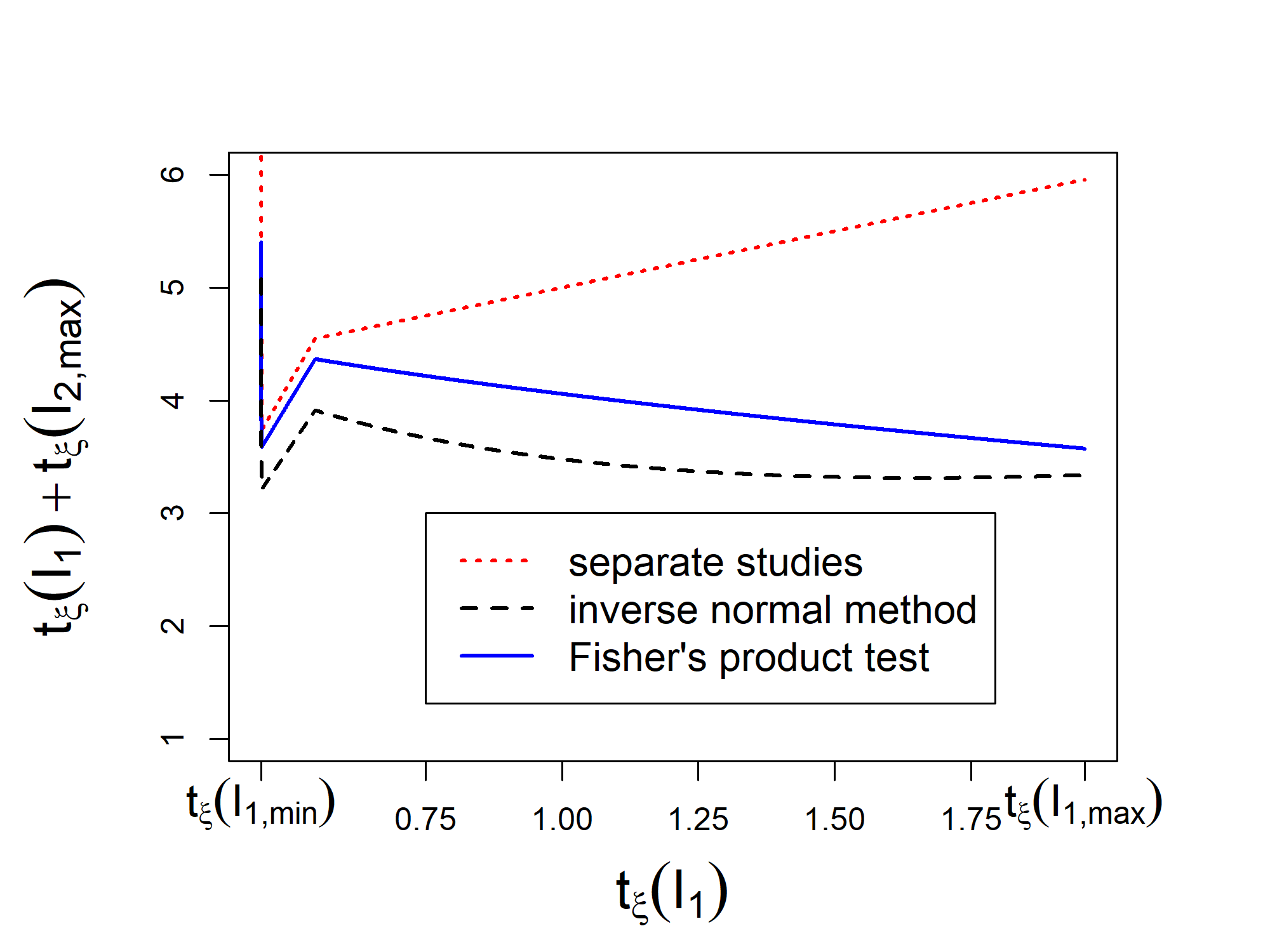}
\subcaption{$\xi=2$.}
\end{subfigure}
\caption{Maximum information over both stages.}
\end{figure}
\section{Additional Plots for Section~\ref{ChapterNonBindingFutility}}\label{AppendixPlotsCombinationStrategy}
The plots in this section complete the discussion in Section~\ref{ChapterNonBindingFutility}. Four further plots for the effect assumption $\xi=1.25$ are presented. Additionally, all efficiency parameters considered in Section~\ref{ChapterNonBindingFutility} for $\xi=1.25$ are calculated under the assumption $\xi=1.43=\xi_{\min}$.
\begin{figure}[h]
\includegraphics[width=1\textwidth]{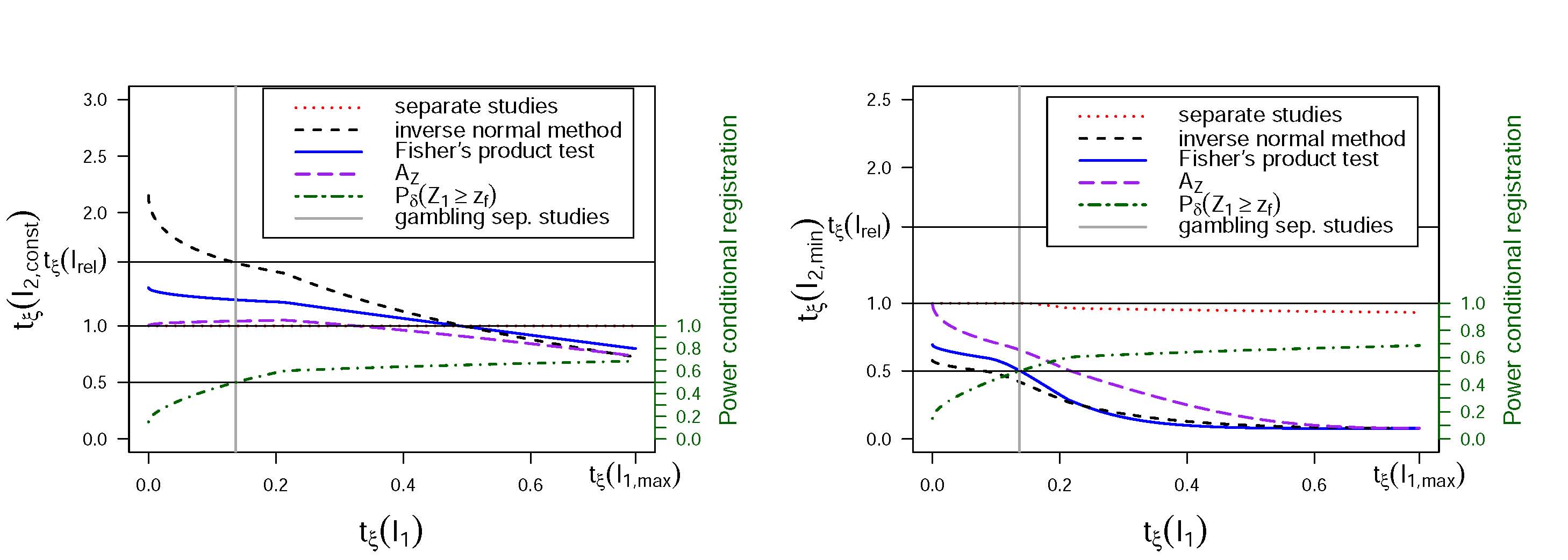}
\caption{Constant second stage information after non-successful conditional registration (left) and minimum second stage information after successful conditional registration (right) for $\xi=1.25$.}
\end{figure}
\begin{figure}[H]
\begin{subfigure}[c]{0.56\textwidth}
\includegraphics[width=.8\textwidth]{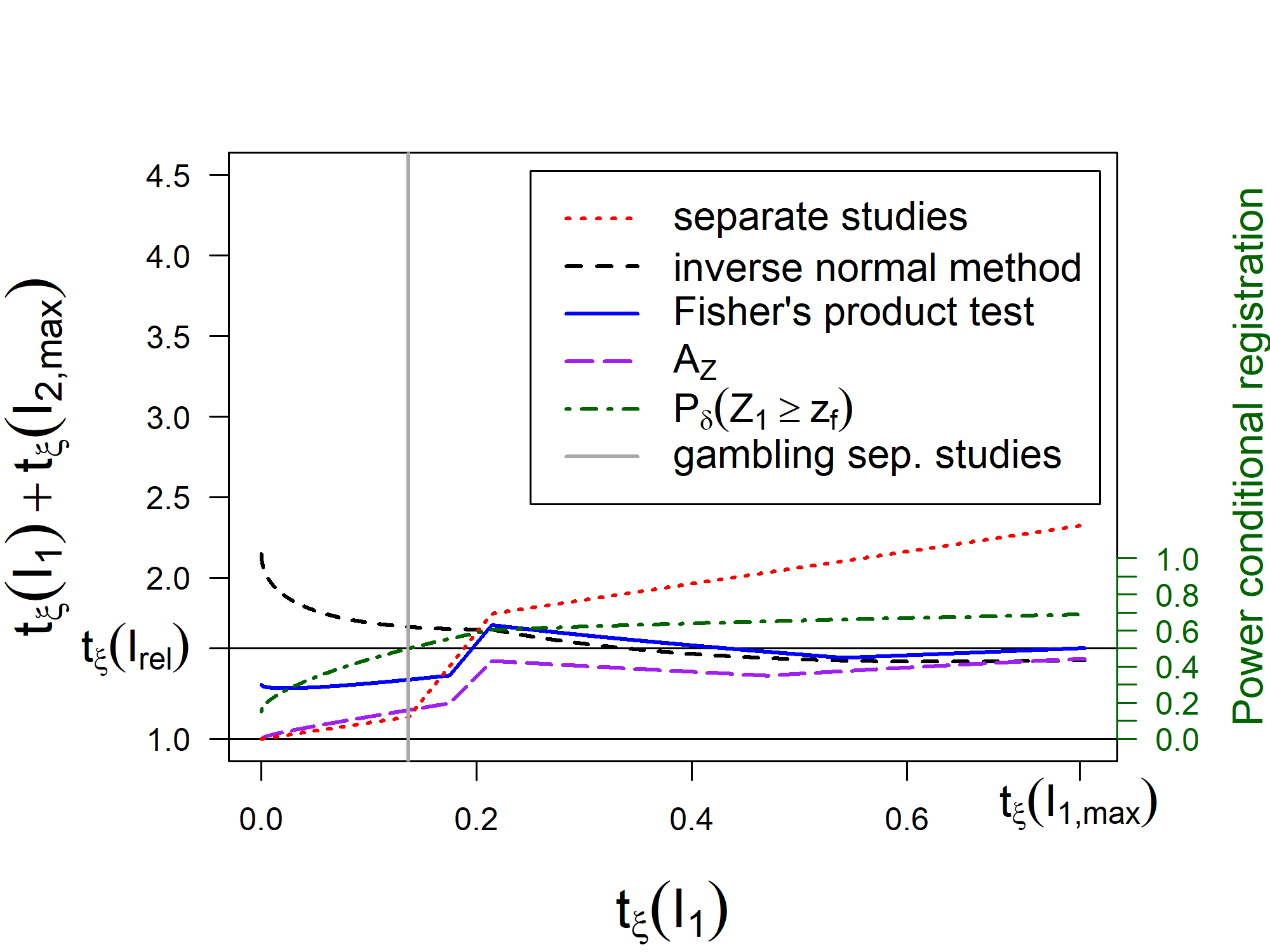}
\subcaption{Maximum Information.}\label{xi1.25MaxInformBothStages}
\end{subfigure}
\begin{subfigure}[c]{0.56\textwidth}
\includegraphics[width=.8\textwidth]{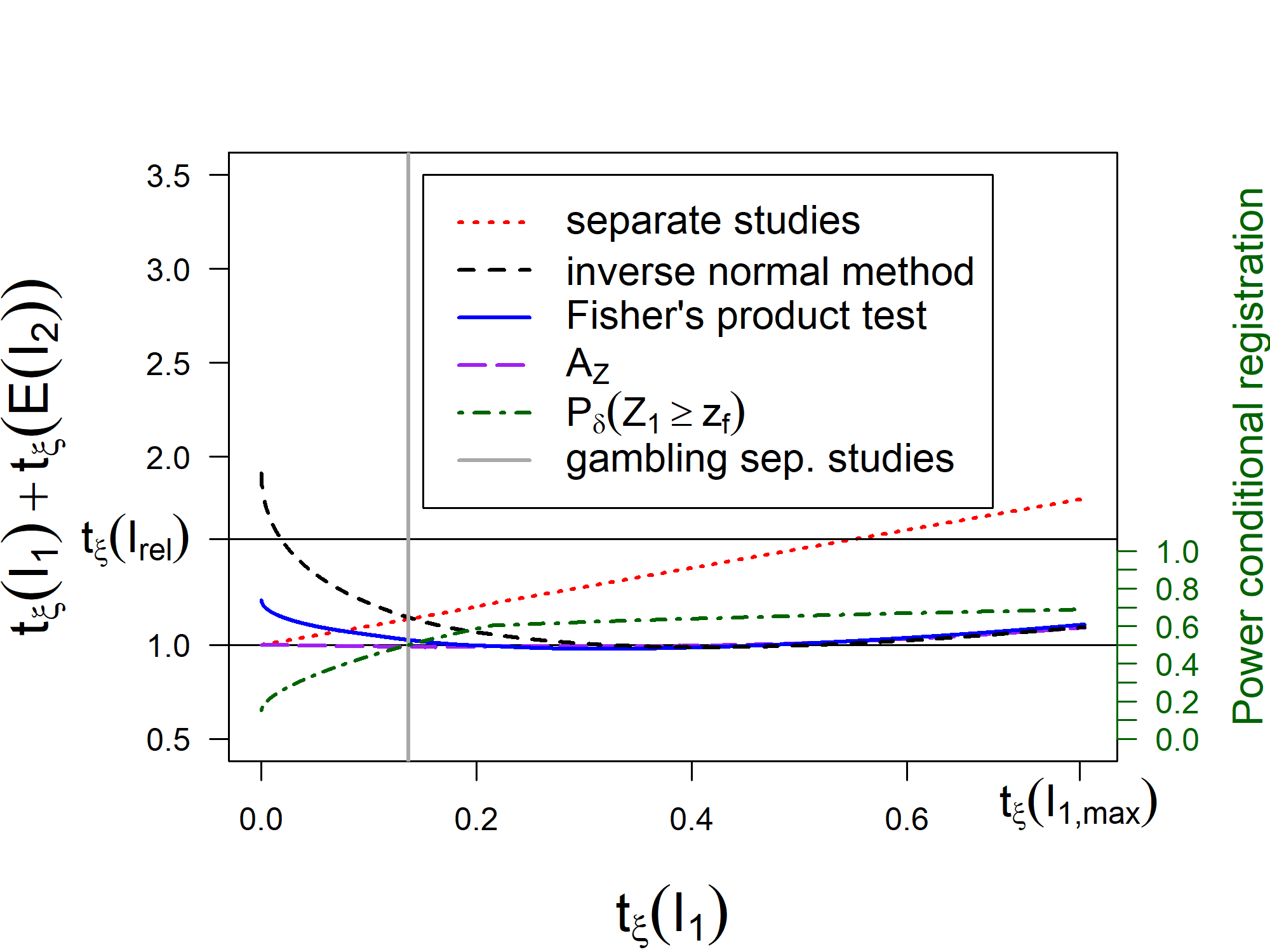}
\subcaption{Mean information.}\label{fig:BP07}
\end{subfigure}
\caption{Maximum and mean information over both stages for $\xi=1.25$.}
\end{figure}
\begin{figure}[H]
\begin{subfigure}[c]{0.5\textwidth}
\includegraphics[width=.95\textwidth]{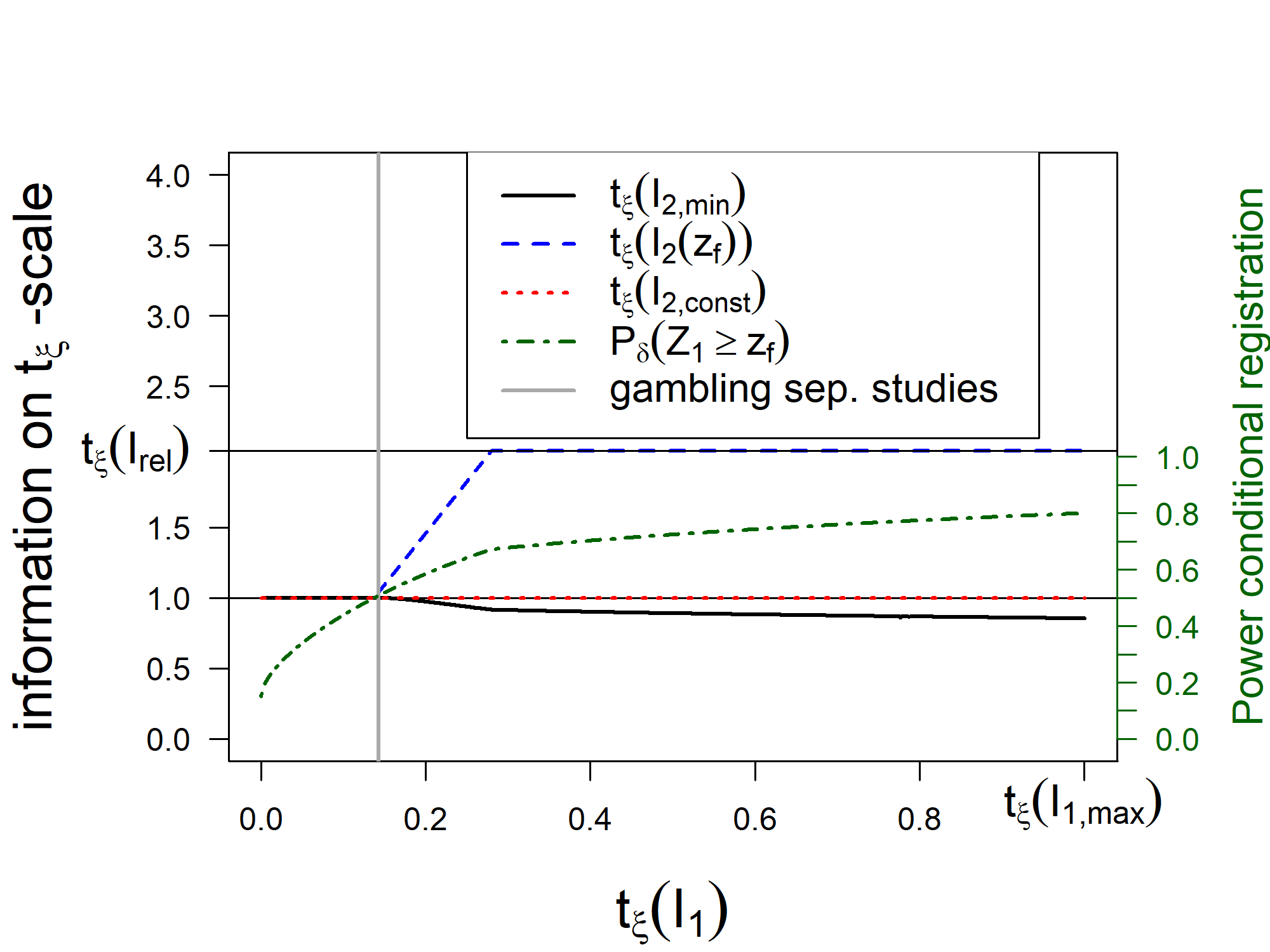}
\subcaption{Non-adaptive design.}
\end{subfigure}
\begin{subfigure}[c]{0.5\textwidth}
\includegraphics[width=.95\textwidth]{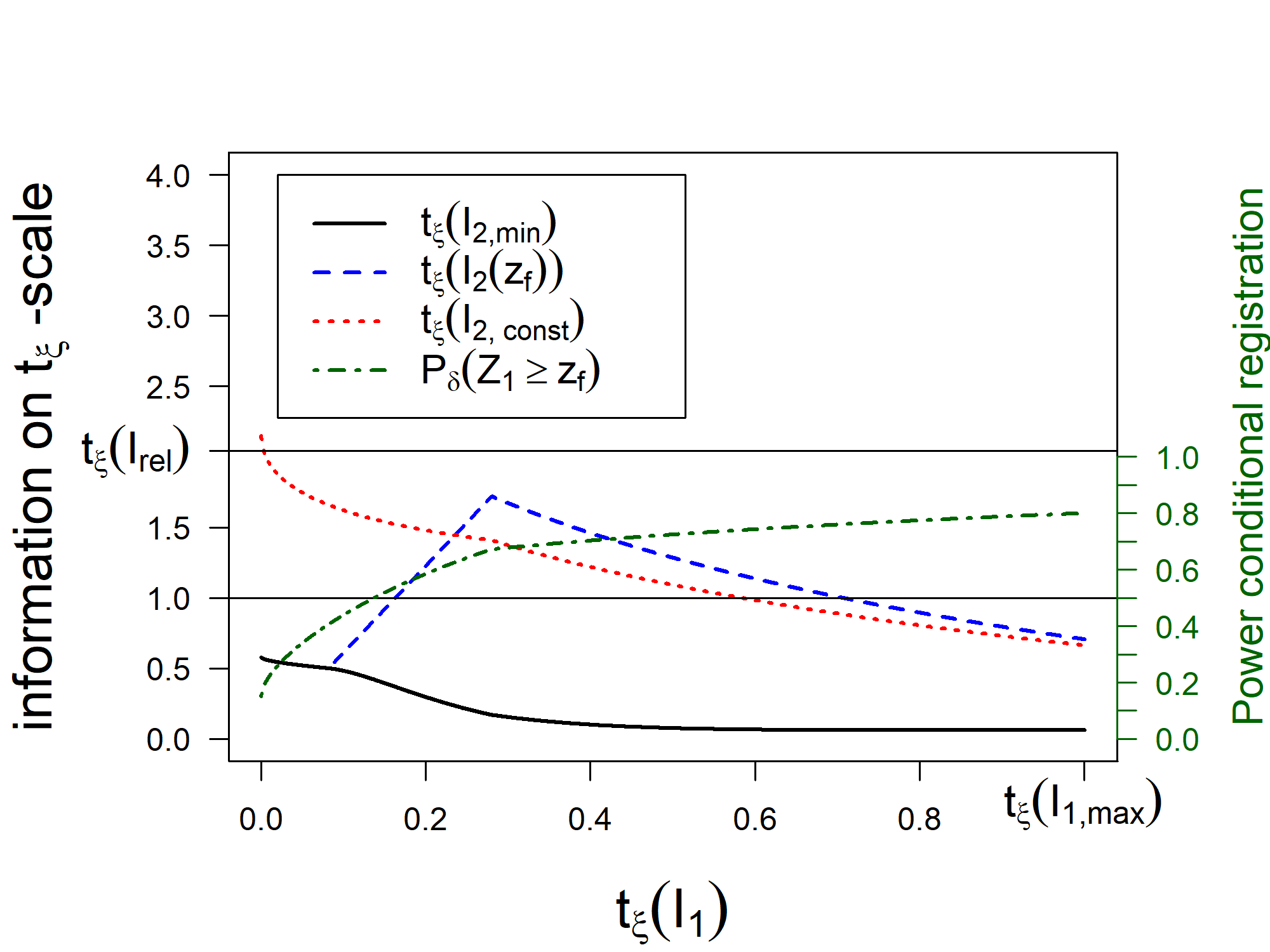}
\subcaption{Inverse normal method.}
\end{subfigure}
\begin{subfigure}[c]{0.5\textwidth}
\includegraphics[width=.95\textwidth]{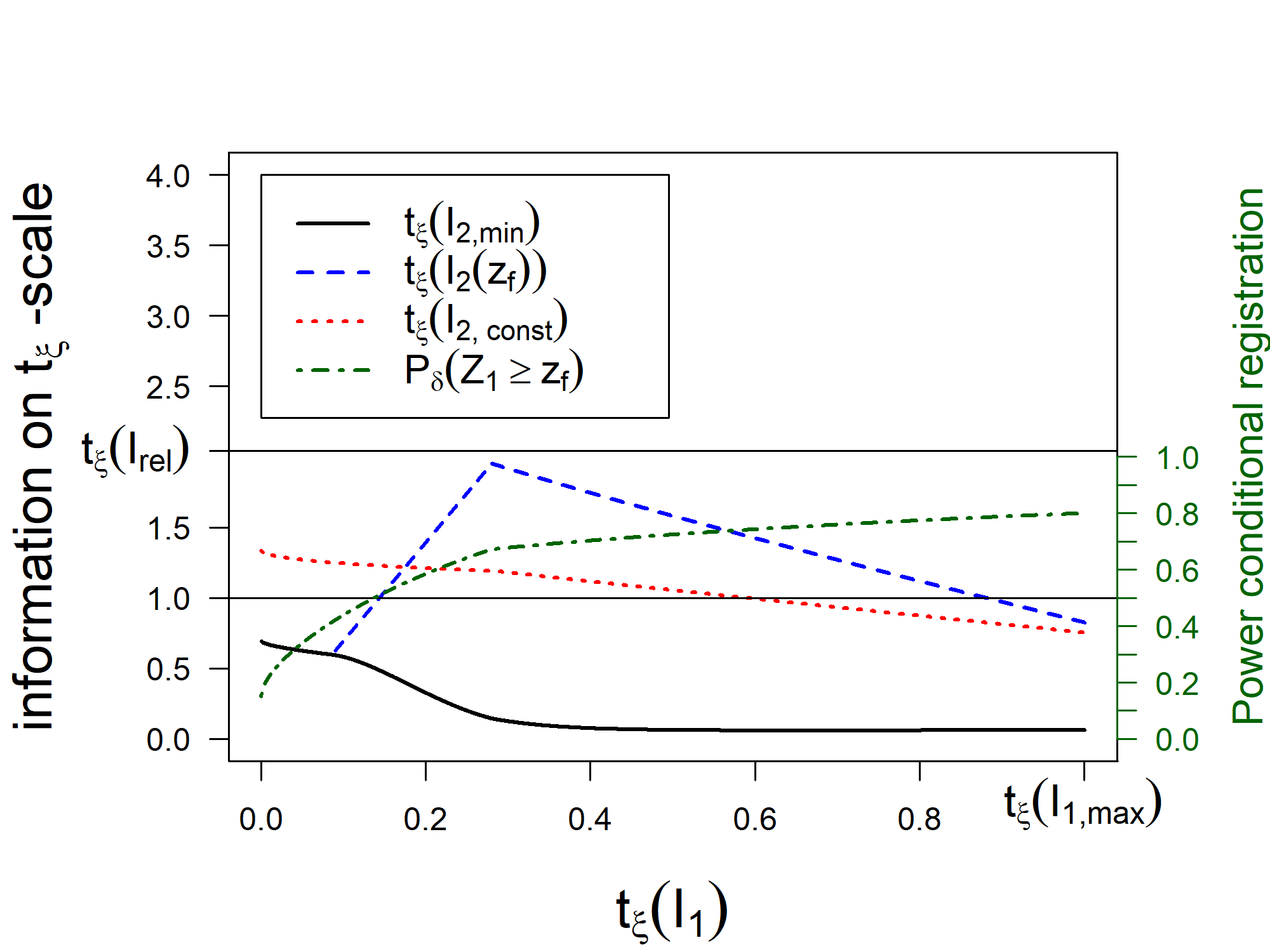}
\subcaption{Fisher's product test.}
\end{subfigure}
\begin{subfigure}[c]{0.5\textwidth}
\includegraphics[width=.95\textwidth]{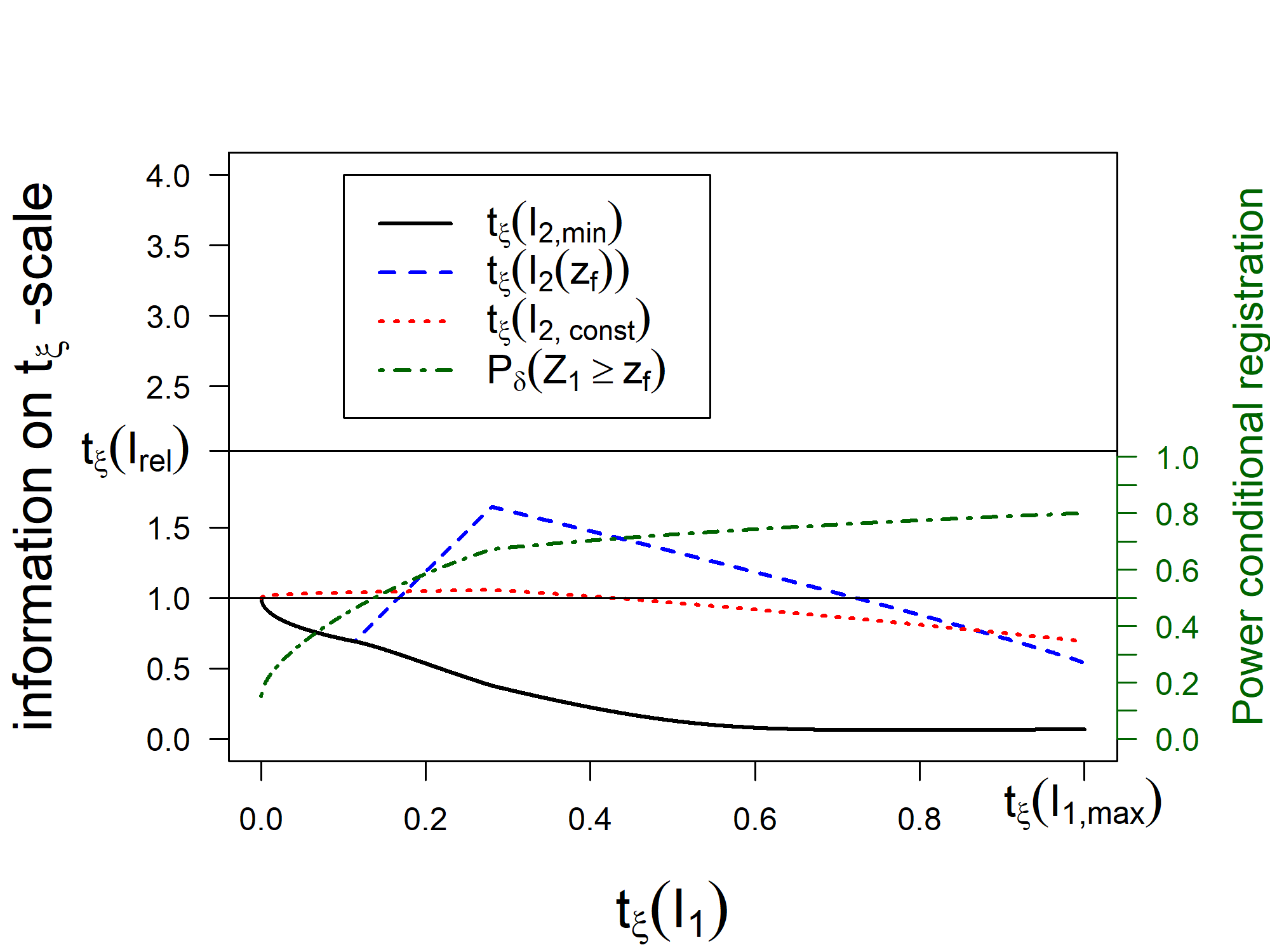}
\subcaption{$A_Z$.}
\end{subfigure}
\caption{Comparison of minimum and maximum information for the second stage of the suggested combination strategy (no conditional registration required for permanent registration) with a non-adaptive and the adaptive designs when $\alpha_c=0.15$ and $\xi=\xi_{\min}=1.43$. }
\end{figure}

\begin{figure}[H]
\begin{subfigure}[c]{0.56\textwidth}
\includegraphics[width=.8\textwidth]{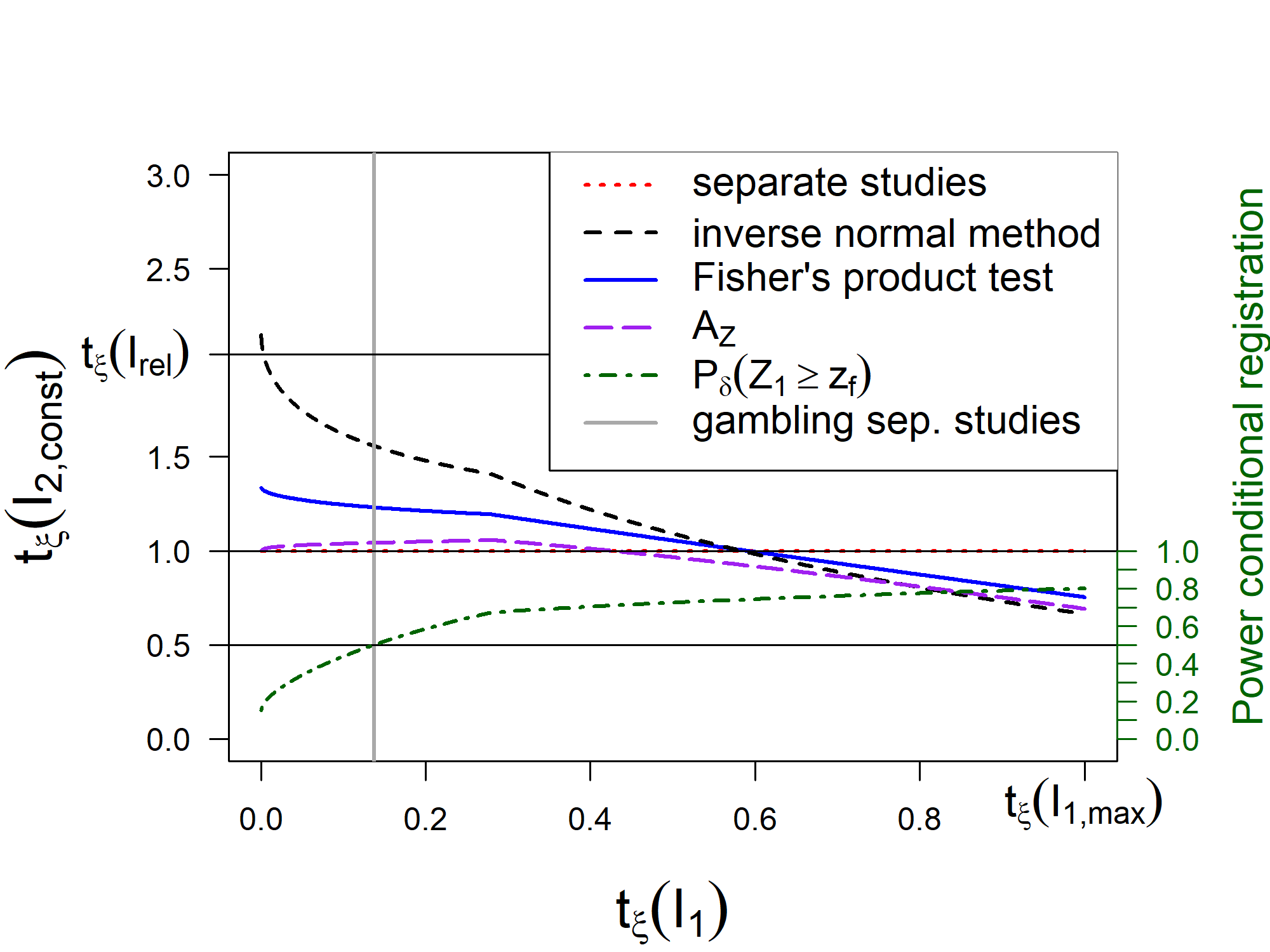}
\subcaption{Constant second stage information.}
\end{subfigure}
\begin{subfigure}[c]{0.56\textwidth}
\includegraphics[width=.8\textwidth]{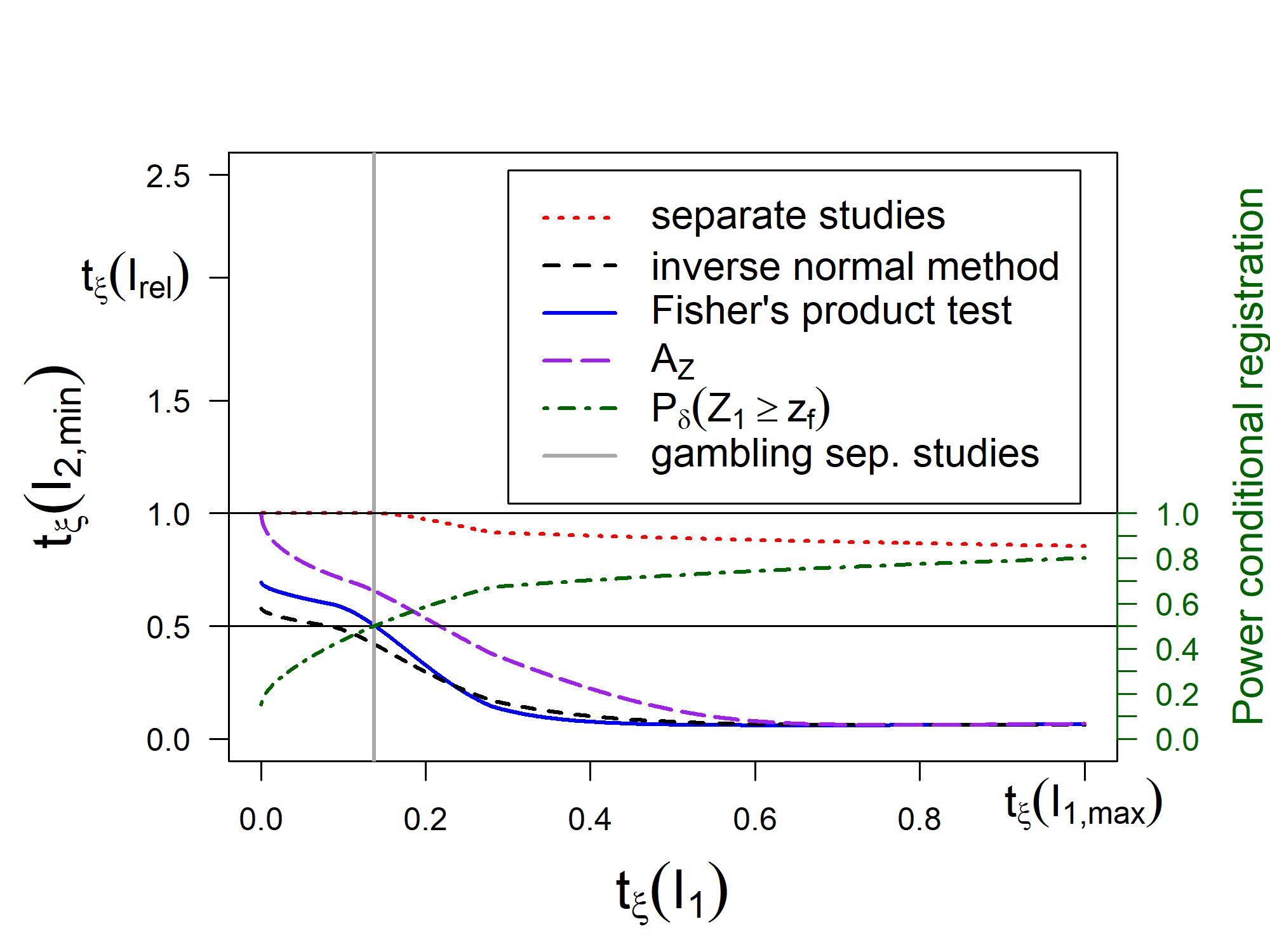}
\subcaption{Minimum second stage information\\ for upper branch.} 
\end{subfigure}
\caption{Constant second stage information after non-successful conditional registration and minimum second stage information after successful conditional registration for $\xi=\xi_{\min}=1.43$.}
\end{figure}

\begin{figure}[H]
\begin{subfigure}[c]{0.56\textwidth}
\includegraphics[width=.8\textwidth]{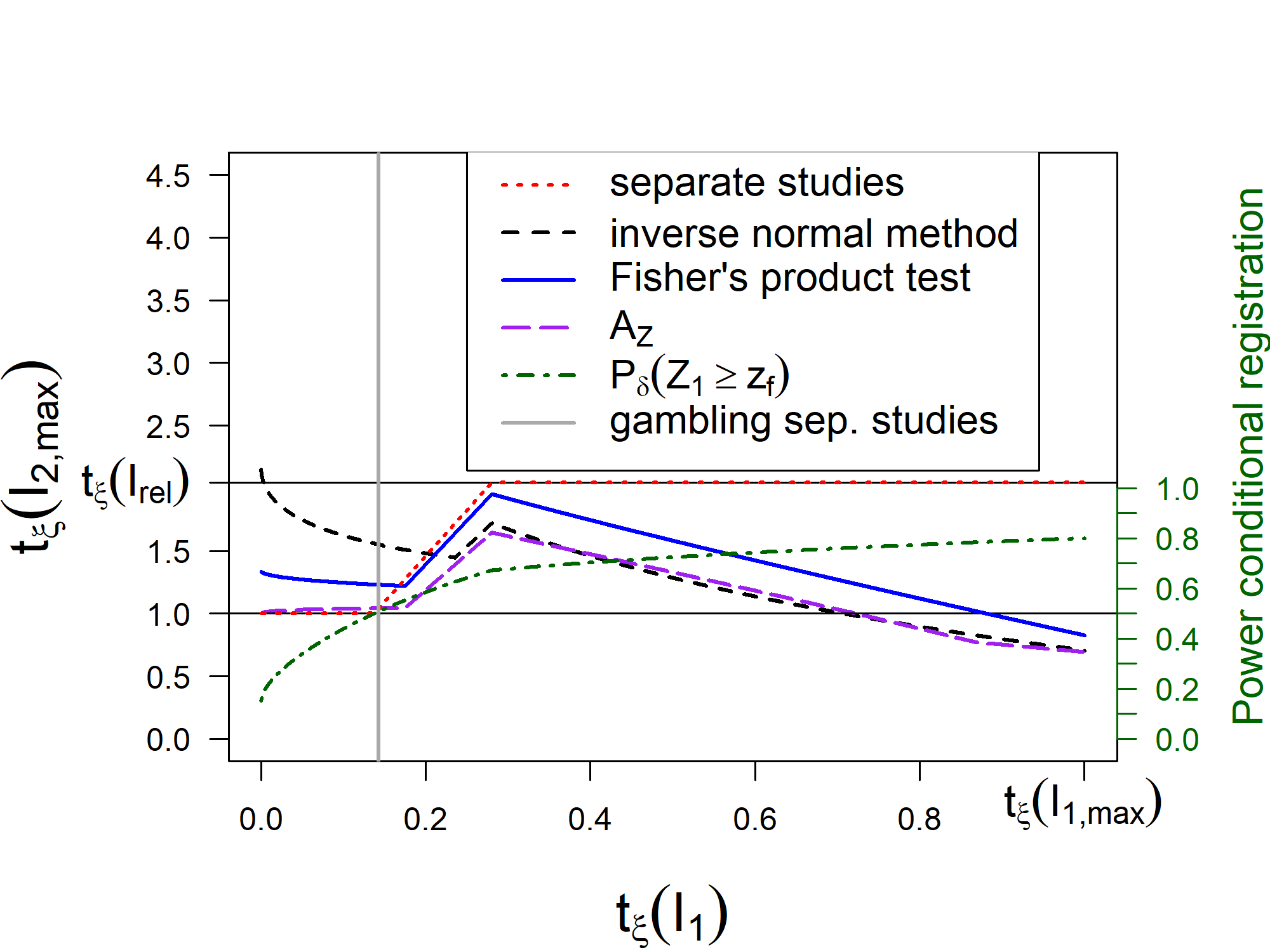}
\subcaption{Maximum information.}
\end{subfigure}
\begin{subfigure}[c]{0.56\textwidth}
\includegraphics[width=.8\textwidth]{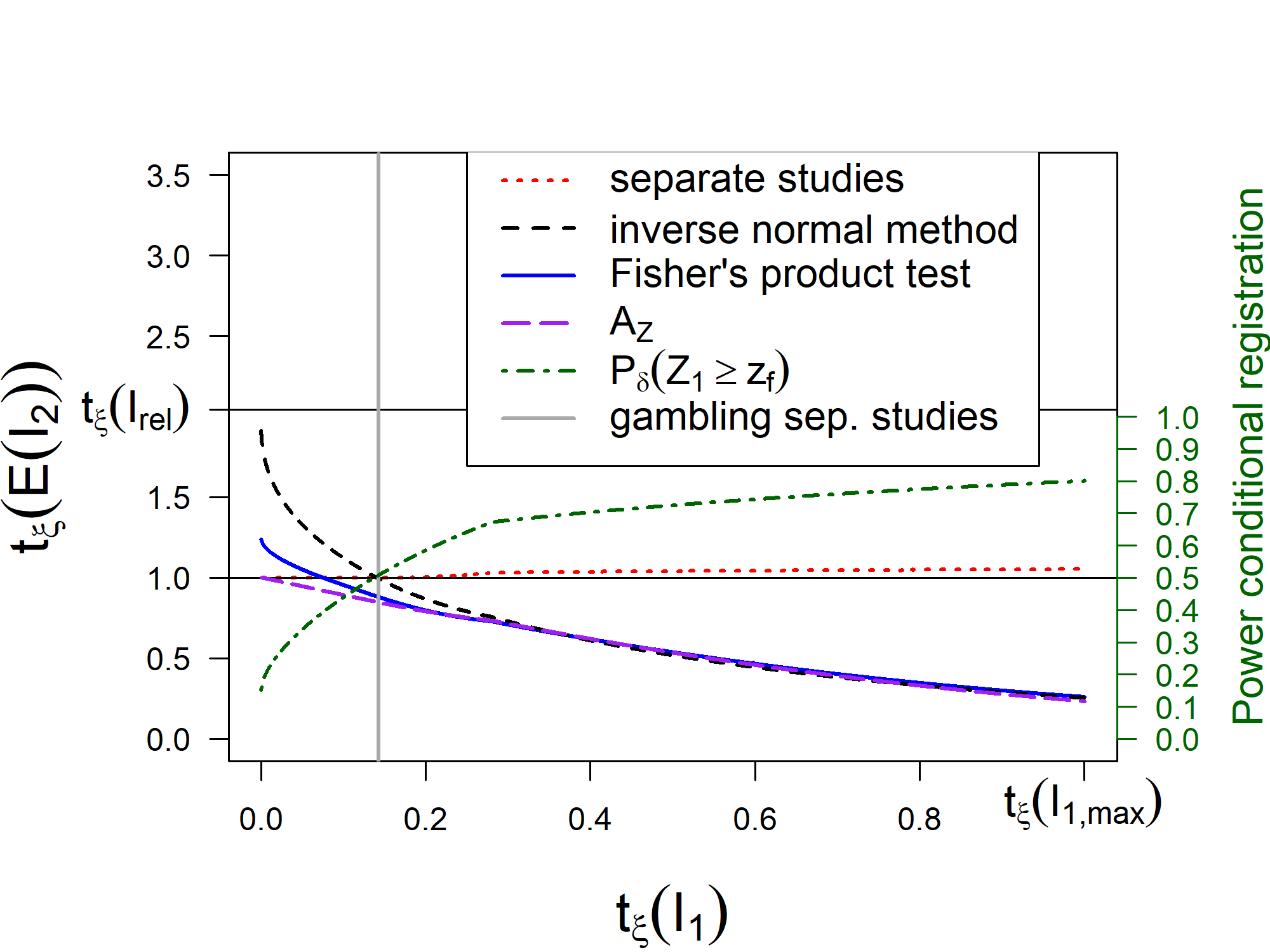}
\subcaption{Mean information.}
\end{subfigure}
\caption{Maximum and mean information of second stage for $\xi=\xi_{\min}=1.43$.}
\end{figure}

\begin{figure}[H]
\begin{subfigure}[c]{0.56\textwidth}
\includegraphics[width=.8\textwidth]{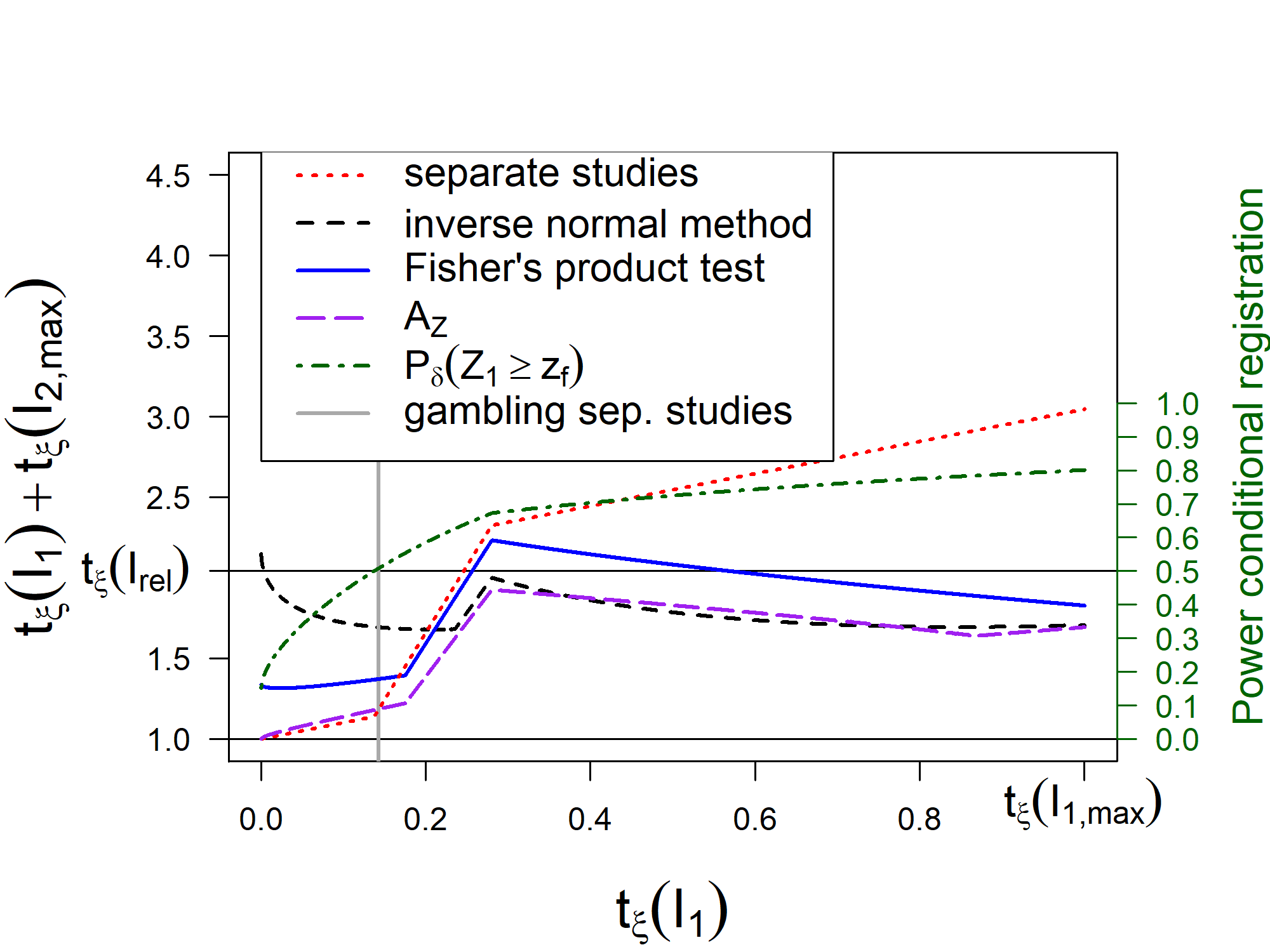}
\subcaption{Maximum Information.}
\end{subfigure}
\begin{subfigure}[c]{0.56\textwidth}
\includegraphics[width=.8\textwidth]{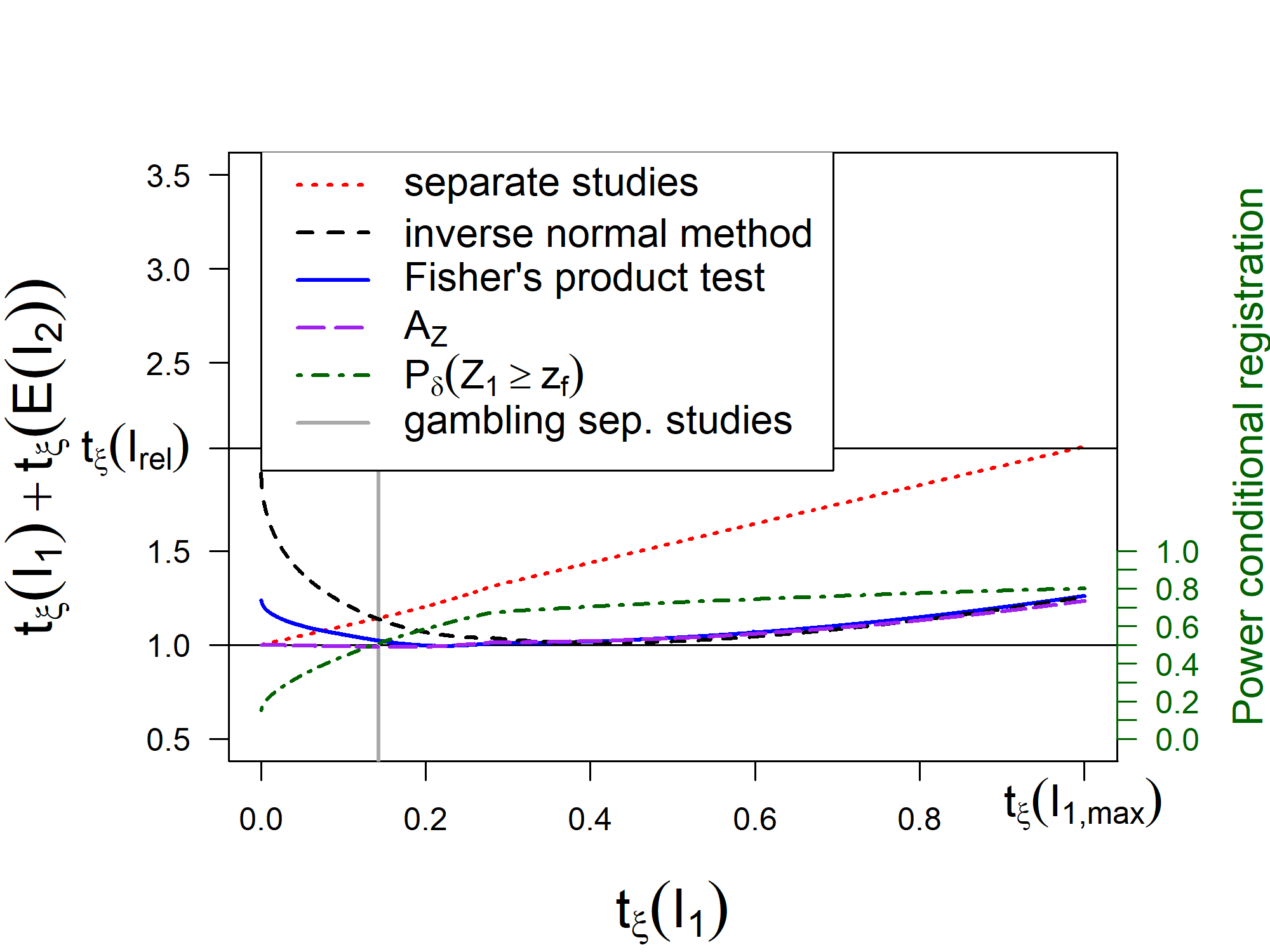}
\subcaption{Mean information.}
\end{subfigure}

\caption{Maximum and mean information over both stages for $\xi=\xi_{\min}=1.43$.}
\end{figure}

\end{document}